\documentclass[aps,superscriptaddress,amsmath,amssymb,onecolumn]{revtex4}

\usepackage{amsmath} 
\usepackage{amsxtra}     
\usepackage{amsthm}
\usepackage{amssymb}
\usepackage{amsfonts}
\usepackage{graphicx}
\usepackage[all, knot]{xy}
\usepackage{multirow}
\usepackage{hyperref}
\usepackage{color}

\usepackage{mathrsfs}
\usepackage{calrsfs}
\usepackage[mathscr]{eucal}

\definecolor{darkblue}{RGB}{8,81,156}
\hypersetup{colorlinks,breaklinks,
            linkcolor=darkblue,urlcolor=darkblue,
            anchorcolor=darkblue,citecolor=darkblue}

\newcommand{\norm}[1]{\left\Vert#1\right\Vert}

\theoremstyle{plain}
\newtheorem{thm}{Theorem}
\newtheorem{cor}[thm]{Corollary}
\newtheorem{lem}[thm]{Lemma}
\newtheorem{prop}[thm]{Proposition}
\theoremstyle{definition}
\newtheorem{defn}{Definition}
\newtheorem{example}{Example}
\theoremstyle{remark}
\newtheorem{rem}{Remark}

\begin{document}

\title{Strong Orientational Coordinates and Orientational Order Parameters For Symmetric Objects}

\author{Amir Haji-Akbari}
\affiliation{Department of Chemical and Biological Engineering, Princeton University, Princeton, NJ 08544}
\affiliation{Department of Chemical Engineering, University of Michigan, Ann Arbor, MI 48109}


\author{Sharon C. Glotzer}
\email{sglotzer@umich.edu}
\affiliation{Department of Chemical Engineering, University of Michigan, Ann Arbor, MI 48109}
\affiliation{Department of Materials Science and Engineering, University of Michigan, Ann Arbor, MI 48109}

\date{\today}

\begin{abstract}
Recent advancements in the synthesis of anisotropic macromolecules and nanoparticles have spurred an immense interest in theoretical and computational studies of self-assembly. The cornerstone of such studies is the role of shape in self-assembly and in inducing complex order. The problem of identifying different types of order that can emerge in such systems can, however, be challenging. Here, we revisit the problem of quantifying orientational order in systems of building blocks with non-trivial rotational symmetries. We first propose a systematic way of constructing orientational coordinates for such symmetric building blocks. We call the arising tensorial coordinates strong orientational coordinates (SOCs) as they fully and exclusively specify the orientation of a symmetric object. We then use SOCs to describe and quantify local and global orientational order, and spatiotemporal orientational correlations in systems of symmetric building blocks. The SOCs and the orientational order parameters developed in this work are not only useful in performing and analyzing computer simulations of symmetric molecules or particles, but can also be utilized for the efficient storage of rotational information in long trajectories of evolving many-body systems. 
\end{abstract}

\maketitle

\section{Introduction\label{section:introduction}}

It is typically easier to visually detect 'order` in a particular structure or pattern, than to describe it mathematically. What is more difficult-- if not impossible-- is to rigorously define what constitutes order, as the task of distinguishing an ordered structure from a disordered structure can be subjective at times. However, certain types of order such as periodicity can be rigorously defined and characterized.  With recent advancements in the synthesis of anisotropic particles~\cite{Ahmadi1996, PengAdvMatSemiconductor2003, JinMirkinNature2003, MurphyJACS2004, ChampoinPolymerPNAS2007, YanJMaterChem2008, MirkinAngewChemIntEd2009, SauAdvancedMaterialRev2010, CarboneNanoToday2010, MurrayJACS2012, GeisslerYang2012, MirkinScience2012, MirkinJACS2013}, it is now possible to assemble more complex forms of ordered structures~\cite{GlotzerSolomonNatureMatrial2007, HajiAkbariEtAl2009, HajiAkbaricondmat2011, HajiAkbariDQC2011, PabloACSNanot2012, PabloScience2012, HajiAkbariPRE2013, MurrayNatChem2013, MilanACSNano2014, GregPNAS2014}. The problem of identifying, distinguishing and quantifying different types of order is therefore of immense practical interest to materials science.

The types of global order that can arise in many-body systems can be loosely classified into two distinct categories. \emph{Translational order} is realized by the positions of the building blocks in a system. Crystalline~\cite{vonLaue1912} and quasicrystalline~\cite{ShechtmanPRL1984} order in atomic systems are examples of stand-alone translational order. \emph{Rotational order} is, however, realized by the orientations of the constituent building blocks. In systems of anisotropic building blocks, there is usually a strong coupling between translational and rotational order as the position of each building block is dictated by the respective shapes and orientations of its neighbors. However, stand-alone rotational order is possible and can, for instance, arise in systems of building blocks with large aspect ratios~\cite{PabloScience2012}. Nematic liquid crystals are the most notable examples, observed in a variety of systems~\cite{Onsager1949, EppengaFrenkel1984,BatesFrenkel1998}. 

There are several well-established methods for characterizing and quantifying translational order. The most popular method is to measure or compute the diffraction image~\cite{vonLaue1912}, $\hat{\rho}(\textbf{q})$, which is related to the density profile, $\rho(\textbf{r})$, via a simple Fourier transform:
\begin{eqnarray}
\hat{\rho}(\textbf{q}) &=& \int_{\mathbb{R}^d} \rho(\textbf{r}) e^{-i\textbf{q}\cdot\textbf{r}}d^d\textbf{r}\notag
\end{eqnarray}
The relationship that exists between the symmetries of $\rho(\textbf{r})$ and $\hat{\rho}(\textbf{q})$ is used for determining the symmetry group of the crystal. Diffraction images are, however, not sufficient for quantifying the extent of translational order. Bond order parameters proposed by Steinhardt~\emph{et al}~\cite{SteinhardtPRB1983} are widely used for that purpose~\cite{FrenkelPRL1995, ChopraJCP2006, DelagoJCP2008, MolineroPCCP2011, HajiAkbariFilmMolinero2014, PalmerNature2014, MolineroJPCB2014, HajiAkbariPNAS2015}. The extent of translational order is quantified by comparing the relative local arrangements of neighboring molecules with that of the ideal crystal. The shape matching algorithm  proposed by Keys~\emph{et al}~\cite{KeysAnnRevCondMatPhys2011, KeysJCompPhys2011} is a generalization of Steinhardt's bond order parameters, and can be used to identify and quantify different types of global and local translational order.  

Characterizing and quantifying rotational order is, however, more difficult. Historically, orientational order parameters are obtained from distribution functions~\cite{Zannoni1979}. Since the orientation of an arbitrary object is fully determined by one polar angle in $\mathbb{R}^2$ and three Euler angles in $\mathbb{R}^3$, any form of global orientational order can be mathematically represented by a distribution function in terms of these angles. Let $f(\Lambda)$ be such a function with $\Lambda=\theta$ for $\mathbb{R}^2$ and $\Lambda=(\theta,\phi,\psi)$ for $\mathbb{R}^3$. Then $f(\Lambda)d\Lambda$ is the infinitesimal probability that an an arbitrary particle assumes an orientation specified by $\Lambda$. If $f(\cdot)$ is a smooth function of $\Lambda$, it can be expanded using a complete set of orthogonal functions in the $\Lambda$ space. For instance for $d=2$ we have:
\begin{eqnarray}
f(\theta) &=& \sum_nf_ne^{in\theta}\label{eq:expansion_fourier}
\end{eqnarray}
which is the familiar Fourier series expansion of $f(\theta)$. For $d=3$, we have~\cite{Zannoni1979}:
\begin{eqnarray}
f(\theta,\phi,\psi) &=& \sum_{l,m,n} f_{m,n}^{l}D_{m,n}^l(\phi,\theta,\psi)\label{eq:expansion_wigner}
\end{eqnarray}
where $D^l_{m,n}$'s are the Wigner matrices that form an orthogonal basis for smooth functions of $\Lambda$. If a structure is rotationally isotropic, the coefficients of Eq.~(\ref{eq:expansion_fourier}) and~(\ref{eq:expansion_wigner}) will all be zero for nonzero integers. Therefore, the presence of any nontrivial nonzero terms in Eq.~(\ref{eq:expansion_fourier}) and~(\ref{eq:expansion_wigner}) is a symptom of broken rotational symmetry. Such nonzero scalars can thus be considered as \emph{orientational order parameters (OOPs)} and collectively describe the orientational distribution of the system. In the presence of nontrivial orientational symmetries in the ordered structure, some additional terms in (\ref{eq:expansion_fourier}) and (\ref{eq:expansion_wigner}) might vanish for symmetry reasons.
 For instance, in a uniaxial nematic liquid crystal with a director along the $z$ axis, $f(\theta,\phi,\psi)$ will only depend on $\theta$, and thus all the terms corresponding to $D^l_{m,n}$'s that explicitly depend of $\phi$ and $\psi$ will vanish. That will correspond to terms with $m,n\neq 0$. The number of order parameters are similarly reduced for other point symmetries~\cite{Zannoni1979, RossoaMolCrystLiqCryst2010}. This approach has been successfully used to derive order parameters for unixial~\cite{Zannoni1979} and biaxial nematics~\cite{Rosso2007} and the cubatic phase~\cite{JohnEscobedo2008}. It has also been used for obtaining suitable expansions for energetic interactions between symmetric molecules and particles~\cite{StoneMolPhys1978, RossoaMolCrystLiqCryst2010}.
 
 The order parameters obtained from the Wigner expansion of a distribution function can be alternatively represented as symmetric traceless tensors of different ranks~\cite{Indenbom1976SovPhysCrystallogr, Fel1995PhysRevE}. Therefore each point symmetry group has a collection of traceless tensors as its order parameters.  The smallest-rank traceless tensor for that particular group is usually described as the order parameter for that symmetry group.  A list of such tensors for major three-dimensional rotation groups are given in Table.~1 of Ref.~\cite{Fel1995PhysRevE}.
 
This classical approach has, however, its own shortcomings. First of all, the mathematical framework for deriving such order parameters has been primarily developed for theoretical studies of liquid crystals-- such as the ones given in~\cite{Fel1995PhysRevE}-- and not for computational studies. Most importantly, those scalar order parameters only show the extent of orientational order and provide little evident information about its geometry. For instance, if the director of a nematic liquid crystal is not along the $z$ axis-- something that  almost always happens in a molecular simulation-- Eq.~(\ref{eq:expansion_wigner}) will be of little practical utility and applying the symmetry constraints to it will not be trivial. For instance, the terms with nonzero $m$ or $n$ will no longer vanish under such circumstances. Although certain procedures are in place for extracting geometric information for certain phases such as the uniaxial and biaxial nematics, no general framework exists for performing this task for an arbitrary type of orientational order.

Another problem lies in the inherent degeneracy of polar (and Euler) angles in describing the	 orientations of symmetric objects, i.e.,~objects with nontrivial rotation groups. We schematically depict this in an example in Fig.~\ref{fig:OOP:sq_deg} where the orientation of a square in two dimensions is unchanged after a $90^{\circ}$ rotation while the value of the polar angle describing its orientation changes from $45^{\circ}$ to $135^{\circ}$. In the conventional approach, these symmetry restrictions are imposed on $f(\Lambda)$, henceforth limiting the number of non-vanishing terms in Eqs.~(\ref{eq:expansion_fourier}) and (\ref{eq:expansion_wigner}). A more useful alternative, which is the focus of this work, is to replace $\Lambda$ with orientational coordinates that are invariant under such symmetry operations. This way, any distribution function in terms of such coordinates will automatically satisfy the symmetry of the underlying building block(s). 

The main purpose of this work is to propose a systematic procedure for constructing such non-degenerate orientational coordinates for objects with arbitrary rotational symmetries. The proposed coordinates are tensors with ranks depending on the rotational symmetry group of the building block. Such coordinates are bijective,~i.e.,~each coordinate is associated with one and only one distinct orientation. They are therefore \emph{strong} descriptors of order, and are henceforth called strong orientational coordinates (SOCs).  Orientational distribution functions are then expressed in terms of such SOCs and not the degenerate polar (or Euler) angles. Order parameters for orientationally ordered structures are accordingly derived as the ensemble averages of certain moments of SOCs. The extent and geometry of order is then determined from solving a generalized nonlinear optimization problem.

\begin{figure}
	\begin{center}
		\includegraphics[width=.5\textwidth]{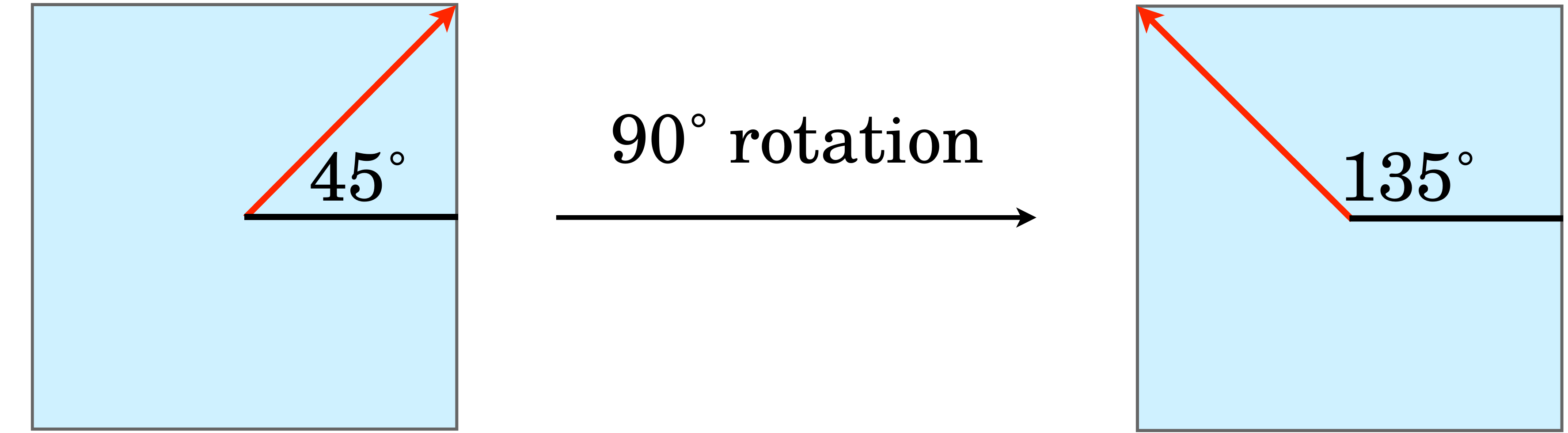}
	\end{center}
	\caption{\label{fig:OOP:sq_deg}Degeneracy of polar angles for describing the orientation of a square in two dimensions.}
\end{figure}

Besides characterizing and quantifying rotational order, SOCs can be used for measuring spatial and/or temporal orientational correlations in computer simulations, something that cannot be easily achieved with non-bijective coordinates (Section~\ref{subsection:spatialCorrelation}). Therefore, SOCs can complement the existing procedures~\cite{GregPNAS2014} used for defining spatial correlation functions for arbitrarily shaped objects, and thanks to their bijectivity, they can provide more accurate information about spatial orientational correlations in liquids, glasses and crystals of high-symmetry building blocks. Similarly, such coordinates can be the basis of defining interaction potentials between anisotropic symmetric objects. They can also be used for time-averaging the rotational behavior of a collection of particles and/or the entire system (Section~\ref{subsection:tAvgOr}). This can be used for efficient storage of rotational information for trajectories of evolving many-body systems. 

This paper is organized as follows. We introduced the utilized notations and conventions in Section~\ref{section:OOP:notations}. We define the notion of strong orientational coordinates in Section~\ref{subsection:homogeneous_tensors}. In Section~\ref{subsection:symmetric_objects}, we use the machinery of strong orientational coordinates and group theory to derive rotational coordinates for symmetric objects. We then derive the corresponding SOCs for all two- and three-dimensional rotation groups. In Section~\ref{subsection:OrderParameter}, we use distribution functions of SOCs to quantify different types of rotational order. We then use this framework to define orientational order parameters for several liquid crystalline phases alongside a few numerical examples. In Section~\ref{section:other_applications}, we  discuss further potential applications of SOCs other than the quantification of global orientational order. Finally, Section~\ref{section:conclusion} is reserved for concluding remarks.

We would like to close this introduction by noting that our approach is purely geometrical, i.e.,~we are not concerned about the physical realizability of the described structures. Instead, our aim is to develop computational tools for quantifying order in such structures if they ever emerge in simulations. Finally, we will use the terms 'building block', 'particle' and 'object' interchangeably, all referring to individual symmetric particles in an evolving many-body system.


\section{Notations And Definitions\label{section:OOP:notations}}
The orientational coordinates derived in this paper are all contra-variant tensors with  ranks depending on the symmetry of the object. In this context, we denote all scalars, i.e.,~rank-0 tensors, by Greek letters, e.g.,~$\alpha,\beta$, etc, all vectors, i.e.,~rank-1 tensors, by small italics, like $v,w$, etc, and all tensors of rank two or higher by capitalized italics-- e.g.~$S,T$. Also, we will use the conventional Einstein notation for tensors whenever necessary.

We denote the $r$-adic power of a vector $v\in\mathbb{R}^d$ as $v^r$,~i.e.,~$(\textbf{v}^r)^{i_1i_2\cdots i_r}=v^{i_1}v^{i_2}\cdots v^{i_r}$. Similarly, the $r$-adic product of $r$ distinct vectors $v_1,v_2,\cdots,v_r$ is denoted by $v_1v_2\cdots v_r$. For every set of vectors $V=\{v_1,v_2,\cdots,v_n\}$, its rank-$r$ homogeneous tensor, $\mathcal{H}_r(V)$, is defined as:
\begin{eqnarray}
\mathcal{H}_{r}(V) &:=& \sum_{p=1}^nv_p^r\label{eq:homogenerous_tensor_form}
\end{eqnarray}
For two tensors $S$ and $T$ of equal rank, a generalized inner product $S\odot T$ is defined as the full contraction between them:
\begin{eqnarray}\label{eq:tetrahedral_OP}
S\odot T&=&\underline{S^{i_1,i_2,\cdots,i_r}}T^{i_1,i_2,\cdots,i_r}
\end{eqnarray}
Here $\underline{u}$ is the complex conjugate of $u$. Note that $\odot$ reduces to the Euclidean inner product for $r=1$, i.e.,~when $S$ and $T$ are both vectors.  Associated with this inner product, a generalized Frobenius norm of a rank-$r$ tensor is defined as:
\begin{eqnarray}\label{eq:F_norm}
\norm{S}_F=\left(S\odot S\right)^{1/2}
\end{eqnarray}
The Frobenius norm reduces to the familiar Euclidean norm for $r=1$ and the matrix Frobenius norm for $r=2$. 


\section{Tensor Order Parameters- General Construction\label{section:tensor_OP}}
We derive proper orientational order parameters for arrangements of symmetric particles as follows. In Section~\ref{subsection:homogeneous_tensors}, we develop the machinery of homogeneous tensors as tools of describing the orientation of an arbitrary set of vectors in $\mathbb{R}^d$ and define the notion of \emph{strong orientational coordinates} of a collection of vectors. We then map the orientation of a symmetric object to a set of equivalent vectors in Section~\ref{subsection:symmetric_objects} and use the SOCs derived in Section~\ref{subsection:homogeneous_tensors} as rotational coordinates of symmetric objects. Finally in Section~\ref{subsection:OrderParameter} we use distribution functions of such SOCs to identify and quantify order in orientationally ordered arrangements of symmetric objects.


\subsection{\label{subsection:homogeneous_tensors}Homogeneous Tensors}
Let $V=\{v_1,v_2,\cdots,v_p\}$ be a finite set of unit vectors in $\mathbb{R}^d$ and let $W=\{Qv: v\in V\}$ with $Q\in O(d)$ an orthogonal transformation. A function $\mathcal{F}(V)$ is called a \emph{strong orientational coordinate} of $V$ if $\mathcal{F}(QV)=\mathcal{F}(V)$ implies $QV=V$. Such a function must be invariant under the permutation of the elements of $V$, a property satisfied by homogeneous tensor forms defined in Eq.~(\ref{eq:homogenerous_tensor_form}). With the following chain of theorems, we establish that for any $V$, there exists an even number $2q\le p$ and an odd number $2r-1\le p$ so that $V=QV$ if and only if $\mathcal{H}_{2q}(V)=\mathcal{H}_{2q}(QV)$ and $\mathcal{H}_{2r-1}(V)=\mathcal{H}_{2r-1}(QV)$.
\begin{lem}\label{lemma:sum_of_powers_roots}
Let $a_1,a_2,\cdots,a_n,b_1,b_2,\cdots,b_n\in\mathbb{R}$, Then $\{a_1,a_2,\cdots,a_n\}=\{b_1,b_2,\cdots,b_n\}$ (up to multiplicities) if and only if $\sum_{i=1}^na_i^k=\sum_{i=1}^nb_i^k$ for $1\le k\le n$.
\end{lem}
\begin{proof}
The  proof is given in Appendix~\ref{proof:lemma:sum_of_powers_roots}.
\end{proof}

\begin{thm}\label{thm:p_tuple}
Let $V=\{v_1,v_2,\cdots,v_p\}$ be an arrangement of $p$ unit vectors in $\mathbb{R}^d$ and let $W=\{Qv: v\in V\}$ where $Q\in O(d)$ is an orthogonal transformation. Then $V=W$ if and only if $\mathcal {H}_k(V)=\mathcal{H}_k(W)$ for every $k\le p$.
\end{thm}
\begin{proof}
$\mathcal{H}_k(V)=\mathcal{H}_k(W)$ for $1\le k\le p$ implies that $[\mathcal{H}_k(V)-\mathcal{H}_k(W)]\odot S=0$ for any $S$. Set $S=v_i^k$ and obtain:
\begin{eqnarray}
v_i^k\odot[\mathcal{H}_k(V)-\mathcal{H}_k(W)] &=& \sum_{j=1}^p\left[(v_i^Tv_j)^k-(v_i^Tw_j)^k\right]= \sum_{j=1}^p\left[\mu_{ij}^k-\xi_{ij}^k\right]\notag
\end{eqnarray}
where $\mu_{ij}=v_i^Tv_j$ and $\xi_{ij}=v_i^Tw_j$. Applying Lemma \ref{lemma:sum_of_powers_roots} yields $\{\xi_{i1},\cdots,\xi_{ip}\}=\{\mu_{i1}, \cdots,\mu_{ip}\}$.
However, for every $1\le i\le p$, there exists a $\xi_{ij}=1$ and this is only possible if there is a vector $w_j\in W$ so that $w_j=v_i$.
\end{proof}

This means that the $p-$tuple $(\mathcal{H}_1,\mathcal{H}_2,\cdots,\mathcal{H}_p)$ uniquely specifies the orientation of $V$. However, we are interested in decreasing the number of necessary coordinates, ideally to one. The following two theorems refine the scope of our search to two coordinates.
\begin{thm}\label{thm:jump_2}
Let $V$ and $W$ as defined in Theorem~\ref{thm:p_tuple} and let $\mathcal{H}_r(V)\neq \mathcal{H}_r(W)$ for some integer $r$. Then $\mathcal{H}_{r+2}(V)\neq \mathcal{H}_{r+2}(W)$.
\end{thm}
\begin{proof}
Let $S^{i_1i_2\cdots i_{r+2}}:=\delta^{i_1i_2}\left[\mathcal{H}_r(V)-\mathcal{H}_r(W)\right]^{i_3i_4\cdots i_{r+2}}$ and observe that:
\begin{eqnarray}
S\odot [\mathcal{H}_{r+2}(V)-\mathcal{H}_{r+2}(W)] &=& \norm{\mathcal{H}_r(V)-\mathcal{H}_r(W)}_F^2>0
\label{eq:inproof_jump_2}
\end{eqnarray}
since $v_i^{r+2}\odot S=v_i^{i_1}\delta^{i_1i_2}v_i^{i_2}v_i^k\odot [\mathcal{H}_r(V)-\mathcal{H}_r(W)]=v_i^k\odot [\mathcal{H}_r(V)-\mathcal{H}_r(W)]$. Use~(\ref{eq:inproof_jump_2}) and the Cauchy Schwartz inequality to conclude that $\mathcal{H}_{r+2}(V)\neq\mathcal{H}_{r+2}(W)$ since $\norm{S}_F>0$.
\end{proof}
\begin{cor}\label{thm:even_odd_sufficiency}
Let $V$ and $W$ be as defined in Theorem (\ref{thm:p_tuple}) then there exists integers $q,r$ with $2q\le p$ and $2r-1\le p$ so that $V=W$ if and only if $\mathcal{H}_{2q}(V)=\mathcal{H}_{2q}(W)$ and $\mathcal{H}_{2r-1}(V)=\mathcal{H}_{2r-1}(W)$.
\end{cor}
\begin{proof}
Let $Q\in O(d)$ so that $QV\neq V$. According to Theorem \ref{thm:p_tuple}, there exists $s\le p$ so that $\mathcal{H}_s(V)\neq \mathcal{H}_s(QV)$. Denote the smallest such integer with $s_Q$ and define $E=\{Q; s_Q\text{ is even}\}$ and $O=\{Q; s_Q\text{ is odd}\}$. Taking $q=\frac12\max_{Q\in E}s_Q$ and $r=\frac12(1+\max_{Q\in O}s_Q)$ completes the proof. In the case of either $E$ or $O$ being empty, the proof is completed by taking an arbitrary $q$ and $r$, respectively.
\end{proof}
As will be explained below, this is the smallest number of coordinates that can be obtained for an arbitrary set. However, further refinement is possible for the subclasses of sets defined below.
\begin{defn}
A set $V$ is \textbf{even} if for every $v\in V$, $-v\in V$ and \textbf{odd} if for every $v\in V,-v\not\in V$.
\end{defn}
Intuitively, one expects the SOC of an even set to be an even-ranked and the SOC of an odd set to be an odd-ranked homogeneous form. This is proven in the following theorems.
\begin{cor}\label{cor:even_sets}
Let $V$ and $W$ be as defined in Theorem~\ref{thm:p_tuple}. If $V$ and $W$ are even, then there exists an integer $2q\le p$ so that $V=QV$ if and only if $\mathcal{H}_{2q}(V)=\mathcal{H}_{2q}(QV)$.
\end{cor}
\begin{proof}
Apply Corollary (\ref{thm:even_odd_sufficiency}) and note that $\mathcal{H}_{2r+1}(V)=0$ for every $r$.
\end{proof}

\begin{lem}
\label{lemma:odd_moments}
Let $x_1,\cdots,x_n,y_1,\cdots,y_n\in\mathbb{R}$ with the property that for every distinct $i,j\le n$, if $x_i\neq0$, then $x_i+x_j\neq0$. Then $\{x_1,\cdots,x_n\}=\{y_1,\cdots,y_n\}$ (up to multiplicities) if and only if $\sum_{i=1}^nx_i^{2k-1}=\sum_{i=1}^ny_i^{2k-1}$ for every $k\le n$.
\end{lem}
\begin{proof}
As explained in Proposition~4.1 of Ref.~\cite{Leclerc1996DiscMath}, the coefficients of $p(z):=\prod_{i=1}^n(z-x_i)$ and $q(z):=\prod_{i=1}^n(z-y_i)$ can be written as rational functions of Schur staircase functions that only depend on sums of odd powers of $x_i$'s (and $y_i$'s).
\end{proof}

\begin{thm}\label{thm:odd_sets}
Let $V$ and $W$-- as defined in Theorem~\ref{thm:p_tuple}-- be odd sets, then their exists $q\le p$ so that $\mathcal{H}_{2q-1}(V)=\mathcal{H}_{2q-1}(W)$ if and only if $V=W$.
\end{thm}
\begin{proof}
Let $\mathcal{H}_{2q-1}(V)=\mathcal{H}_{2q-1}(W)$ for every $q\le p$ and observe that $v_i^{2q-1}\odot [\mathcal{H}_{2q-1}(V)-\mathcal{H}_{2q-1}(W)]=\sum_{j=1}^p(\mu_{ij}^{2q-1}- \xi_{ij}^{2q-1})=0$ for every $i\le p$. Due to the oddness of $V$, however, $\mu_{ij}$'s and $\xi_{ij}$'s satisfy the conditions specified in Lemma~\ref{lemma:odd_moments}. Therefore, $\{\xi_{ij}\}=\{\mu_{ij}\}$ (up to multiplicities) and $V=W$. The proof is completed by applying Theorem~\ref{thm:jump_2}.
\end{proof}

\begin{rem}
Note that a set $V$ that is neither even nor odd, can be partitioned into two nonempty even and odd sets $V=V_e\cup V_o$. The SOC for $V$ will therefore be a pair of an even-ranked homogenous form (for $V_e$) and an odd-ranked homogeneous form (for $V_o$). An even-ranked form of $V$ cannot exclusively specify the orientation of $V$  as it will be invariant under an inversion that will map $V_o$ to $-V_o$. Similarly, an odd-ranked homogeneous form will also be insufficient as it will be invariant under the orthogonal transformations that alter $V_e$ while keeping $V_o$ unchanged. 
\end{rem}

So far, we have established upper bounds on the tensorial rank of strong orientational coordinates of even and odd sets. These upper bounds are generic in the sense that they only depend on the cardinality of the underlying set and not its structure. In general, tensor SOCs of smaller ranks might be possible for sets with certain symmetries. As will become evident in Section~\ref{subsection:symmetric_objects}, however, the bounds proposed here can still be tight for certain types of symmetries.


\subsection{Symmetric Objects\label{subsection:symmetric_objects}}
The SOCs developed in Section~\ref{subsection:homogeneous_tensors} are permutation-invariant. They can thus be used for describing the orientations of a symmetric object if a bijection can be established between the orientation of the object and a set of equivalent vectors. We will achieve this by generating orbits of suitable vectors under the action of the corresponding rotation group. It is necessary to emphasize that the construct outlined in Section~\ref{subsection:homogeneous_tensors} can be applied to any point symmetry group, and our consideration of  rotation groups is due to the fact that the orientation of a physical object (chiral or achiral) can only change as a result of a rotation. In order to establish the connection, it is necessary to re-introduce some standard concepts in group theory.

A rigid body $\mathcal{R}\subset\mathbb{R}^d$ is symmetric if there is a non-identity $Q\in SO(d)$ that $Q\mathcal{R}=\mathcal{R}$. The \emph{rotation group} of a rigid body $\mathcal{R}$ is given by:
\begin{eqnarray}
\mathcal{G}_{\mathcal{R}} &:=& \{Q\in SO(d):Q\mathcal{R}=\mathcal{R}\}
\end{eqnarray}
For every vector $v\in\mathcal{R}$, its \emph{orbit} is given by:
\begin{eqnarray}
\mathcal{O}_{v} &:=& \{Qv: Q\in\mathcal{G}_{\mathcal{R}}\}
\end{eqnarray}
In a Lie group, an orbit can be an uncountably infinite set. 

\begin{thm}
\label{thm:equiv_class}
$\mathcal{G}_{\mathcal{R}}$ has the following properties:
\begin{enumerate}
	\item $\mathcal{G}_{\mathcal{R}}$ partitions $\mathcal{R}$ into equivalency classes i.e. $\mathcal{O}_v=\mathcal{O}_{Qv}$ for every $v\in\mathcal{R}$ and $Q\in\mathcal{G}_{\mathcal{R}}$.
	\item $\mathcal{O}_v$ is either even or odd.
\end{enumerate}
\end{thm}
\begin{proof}
Here we outline the proof for (2) as (1) directly follows from group properties. If for every $u\in\mathcal{O}_v, -u\not\in\mathcal{O}_v$, then $\mathcal{O}_v$ is odd an (2) follows. Suppose there is some $u\in\mathcal{O}_v$ with $-u\in\mathcal{O}_v$ and let $w\in\mathcal{O}_v$, then there exist $Q_+,Q_-,Q\in\mathcal{G}_{\mathcal{R}}$ with $u=Q_+v,-u=Q_-v$ and $w=Qv$. We thus have:
\begin{eqnarray}
-w=-Qv=-QQ_+^*u=QQ_+^*(-u)=QQ_+^*Q_-v\notag
\end{eqnarray}
but $QQ_+^*Q_-\in\mathcal{G}_{\mathcal{R}}$ and hence $-w\in\mathcal{O}_v$ and (2) follows.
\end{proof}

All vectors in an orbit are equivalent in the sense that they can be mapped onto one another via rotations that leave the orientation of $\mathcal{R}$ unchanged. Each orbit can therefore be uniquely described by an SOC derived in Section~\ref{subsection:homogeneous_tensors} as it is either an even or an odd set. A single orbit is not necessarily sufficient for describing the orientation of $\mathcal{R}$. We might therefore need to map the orientation of a rigid body to a collection of distinct orbits $\mathcal{W}_{\mathcal{R}}=\cup_{i=1}^N\mathcal{O}_i$. We call such a collection a \emph{symmetric descriptor} of $\mathcal{R}$. An \emph{irreducible symmetric descriptor} of $\mathcal{R}$ is defined as a symmetric descriptor that none of its subsets of constituent orbits is sufficient for describing the orientation of $\mathcal{R}$.  Since each orbit in an irreducible symmetric descriptor, $\mathcal{W}$, is uniquely specified by $\mathcal{H}_i$, a strong coordinate derived in~\ref{subsection:homogeneous_tensors}, the orientation of $\mathcal{R}$ can be uniquely specified by the $N$-tuple $(\mathcal{H}_1,\mathcal{H}_2,\cdots,\mathcal{H}_N)$. 

There is no unique way of constructing an irreducible symmetric descriptor for a symmetric object as the size and the structure of a given orbit will depend on the vector that generates it. In general, one would prefer orbits with fewer vectors as their associated SOCs will be of smaller ranks, and will therefore be easier to compute and store on a computer. The cardinality of an orbit is at most equal to the order of $\mathcal{G}_{\mathcal{R}}$, the rotation group of $\mathcal{R}$. However, if the initial generating vector is chosen so that it is invariant under certain group operation,  a smaller orbit will be obtained. Such vectors are equivalent to Wyckoff positions in a space group~\cite{IntTabCrystWyckoff}. In general, the orientation of a rigid object in $\mathbb{R}^d$ can be uniquely specified by \emph{at most} $d-1$ linearly independent vectors. Note that fewer vectors might be needed if the rotation group of $\mathcal{R}$ is a Lie group. A symmetric descriptor should therefore have the following properties:
\begin{itemize}
\item It should have sufficient number of linearly independent vectors.
\item It should not be invariant under any rotation that changes the orientation of $\mathcal{R}$.
\end{itemize}
 Here we derive irreducible symmetric descriptors and the associated SOCs for all two- and three-dimensional rotation groups. The results are summarized in Table.~\ref{table:2D3DSOCs}.  

\bigskip


\subsubsection{Trivial Rotation Group}
For a non-symmetric object, $\mathcal{G}_{\mathcal{R}}=\{I\}$ and $\mathcal{O}_{v}=\{v\}$ for every $v\in\mathcal{R}$. An irreducible symmetric descriptor is therefore given by $\cup_{i=1}^{d-1}\mathcal{O}_{v_i}$ with $v_i$'s being linearly independent. A strong coordinate is therefore given by a $(d-1)$-tuple $(\mathcal{H}_1,\mathcal{H}_2,\cdots,\mathcal{H}_{d-1})=(v_1,v_2,\cdots,v_{d-1})$. For a non-symmetric object in two and three dimensions, this will correspond, as expected, to one and two linearly independent vectors, respectively.


\begin{figure}
	\begin{center}
		\includegraphics[width=.5\textwidth]{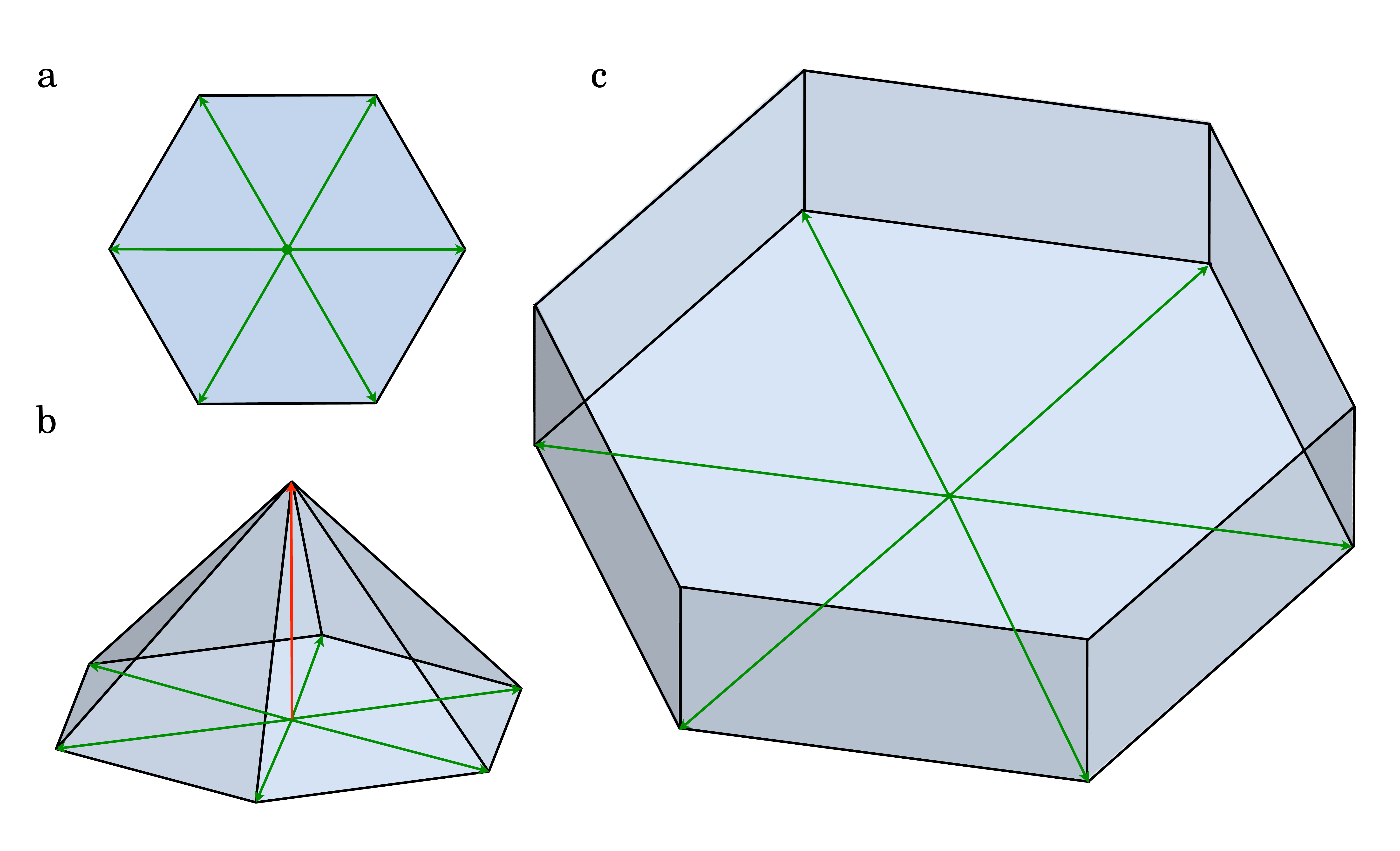}
		\caption{\label{fig:CnDn}$C_n$ and $D_n$ symmetries in $\mathbb{R}^2$ and $\mathbb{R}^3$. (a-b) $C_6$ symmetry in (a) $\mathbb{R}^2$ and (b) $\mathbb{R}^3$. (c) $D_6$ in $\mathbb{R}^3$. The vectors shown in green are irreducible symmetric descriptors in (a) and (c) but not in (b).}
	\end{center}
\end{figure}

\subsubsection{$C_n$ and $D_n$\label{section:SOC:CnDn}}
$C_n$ is the only non-trivial rotation group in $\mathbb{R}^2$ and corresponds to the symmetry of a regular $n$-gon (Fig.~\ref{fig:CnDn}a). In three dimensions, $C_n$ corresponds to the symmetry of a pyramid with a regular $n$-gonal base (Fig.~\ref{fig:CnDn}b). $D_n$, however, corresponds to the symmetry of a prism or a bipyramid with a regular $n$-gonal basis (Fig.~\ref{fig:CnDn}c). A characteristic orbit of both $C_n$ and $D_n$-- denoted by $\mathcal{O}_n$-- is generated by an arbitrary unit vector from within the plane that is perpendicular to the $n$-fold rotation axis. Such an orbit consists of $n$ vectors. In $\mathbb{R}^2$, $\mathcal{O}_n$ is an irreducible symmetric descriptor of $C_n$. In $\mathbb{R}^3$, however, it is only an irreducible symmetric descriptor of $D_n~(n\ge3)$. For a $C_n$ object, however, $\mathcal{O}_n$ is invariant under a $180^\circ$ rotation around one of its constituent vectors, while such a transformation will change the orientation of a $C_n$ object.  An irreducible symmetric descriptor of $C_n$  is therefore a union of the 'planar` orbit, $\mathcal{O}_n$, and a single vector that is parallel to the axis of rotation, e.g.,~the red vector in Fig.~\ref{fig:CnDn}c. $\mathcal{O}_n$ is also not an irreducible symmetric descriptor of an object with $D_2$ symmetry as it only contains two collinear vectors. In that case, an irreducible symmetric descriptor can be constructed as the union of $\{\pm z\}$ and $\mathcal{O}_n$ where $z$ is the rotation axis. 

In Appendix~\ref{appendix:derivation:CnDn}, we show that $\mathcal{H}_n$ is an SOC for the planar orbit $\mathcal{O}_n$. Note that for even $n$, $\mathcal{H}_n$, and for odd $n$, $\mathcal{H}_{2n-1}$ are guarranteed to be SOCs of $\mathcal{O}_n$. Therefore, the upper bound given in Corollary~\ref{cor:even_sets} is tight for even $n$s.

\begin{rem}
$C_3$ is a subgroup of $T$, the rotation group of a regular tetrahedron, which is a triangular pyramid. $D_4$ is a subgroup of $O$, the rotation group of an octahedron, which is a square bipyramid, and a cube, which is a square prism. These groups have different SOCs and will be discussed separately. 
\end{rem}


\subsubsection{$C_{\infty}$ and $D_{\infty}$\label{subsubsection:CinfDinf}}
$C_{\infty}$ is the symmetry of a cone with a circular base, while $D_{\infty}$ corresponds to the symmetry of a cylinder. For both symmetries, two types of orbits are possible. An orbit generated by a vector along the rotation axis will be finite and will have one and two elements for $C_{\infty}$ and $D_{\infty}$, respectively. All other orbits, however, will be uncountably infinite sets comprising of one circle for $C_{\infty}$ and one or two parallel circles for $D_{\infty}$. It can be noted that a finite orbit $\mathcal{O}_f$ is an irreducible symmetric descriptor for both $C_{\infty}$ and $D_{\infty}$. The corresponding SOC will therefore be $\mathcal{H}_1=z$ and $\mathcal{H}_2=zz$ for $C_{\infty}$ and $D_{\infty}$, respectively.


\subsubsection{Tetrahedral Symmetry $T$\label{section:SOC:tetrahedral}}
$T$ corresponds to the rotation group of a regular tetrahedron (Fig.~\ref{fig:tetrahedron}) and has twelve elements. A general orbit of $T$ will therefore have twelve elements as well. The high symmetry vectors that connect the centroid of a regular tetrahedron to its vertices, depicted in green in Fig.~\ref{fig:tetrahedron}, can, however, generate an orbit $\mathcal{O}_t=\{a_p\}_{p=1}^4$ that only has four elements. $\mathcal{O}_t$ is an irreducible symmetric descriptor of a regular tetrahedron. Therefore, its SOC will also be the SOC of a regular tetrahedron.  Note that $\mathcal{O}_t$ is an odd set, and, according to Theorem~\ref{thm:odd_sets}, $\mathcal{H}_1,\mathcal{H}_3,\mathcal{H}_5,\mathcal{H}_7$ will be its candidate SOCs. However, note that: 
\begin{eqnarray}
\mathcal{H}_1^i &=& \sum_{p=1}^4a_p^i = 0\label{eq:tet:H1}\\
\mathcal{H}_2^{ij} &=& \sum_{p=1}^4a_p^ia_p^j = \frac43\delta^{ij}\label{eq:tet:H2}
\end{eqnarray}
We prove that $\mathcal{H}_3$ is the SOC for $\mathcal{O}_T$ as follows. Let $\mathcal{O}_{T}=\{a_p\}_{p=1}^4$ and $\mathcal{O}_{T'}=\{b_q\}_{q=1}^4$ be the corresponding orbits for the two tetrahedra $T$ and $T'$ and observe that:
\begin{eqnarray}
\norm{\mathcal{H}_3(\mathcal{O}_T)-\mathcal{H}_3(\mathcal{O}_{T'})}_F^2 &=& 2\left[\frac{32}9-\sum_{p,q=1}^4\xi_{pq}^3\right]\label{eq:tet:objective}
\end{eqnarray}
with $\xi_{pq}=a_p^Tb_q$. Contracting (\ref{eq:tet:H1}) and (\ref{eq:tet:H2}) with $b_q^i$ and $b_q^ib_q^j$ yields:
\begin{eqnarray}
	\begin{array}{ll}
		\sum_{p=1}^4\xi_{pq} = 0 & q=1,\cdots,4\\
		\sum_{p=1}^4\xi_{pq}^2 = \frac43 & q=1,\cdots,4
	\end{array}\label{eq:tet:constraints}
\end{eqnarray}
As shown in Appendix~\ref{appendix:tetopt},~(\ref{eq:tet:objective}) can only be zero if $\xi_{pq}$'s are the permutations of $(1,-\frac13,-\frac13,-\frac13)$ for every $q$ i.e. if $\mathcal{O}_T=\mathcal{O}_{T'}$.

\begin{figure}
	\begin{center}
		\includegraphics[width=.4\textwidth]{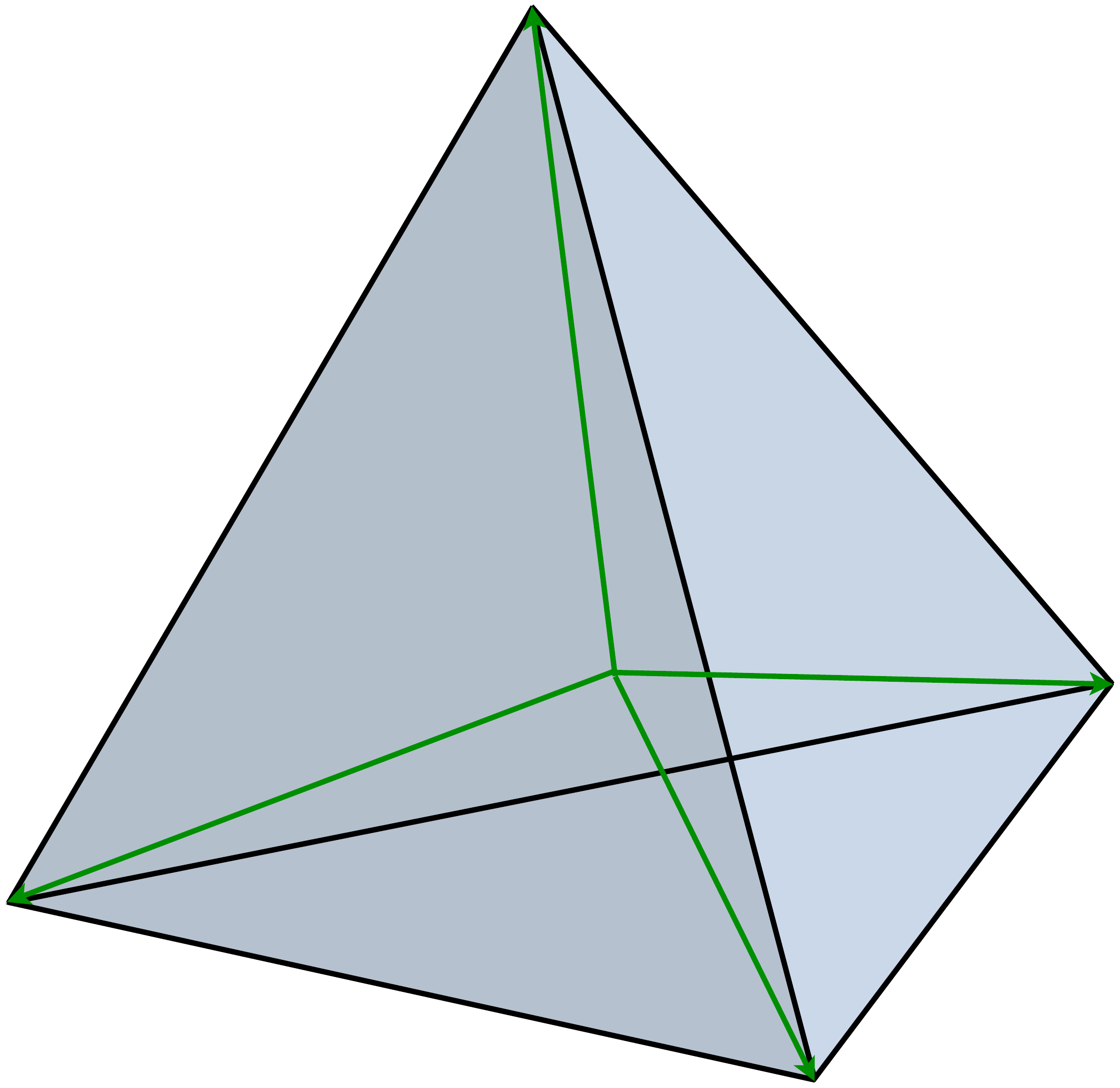}
		\caption{\label{fig:tetrahedron}A regular tetrahedron with the rotation group $T$. The smallest orbit has four elements that are depicted in green.}
	\end{center}
\end{figure}


\subsubsection{Octahedral Symmetry $O$\label{section:SOC:octahedral}}
$O$ corresponds to rotational symmetry of a cube and an octahedron (Fig.~\ref{fig:octahedron}) and has twenty-four elements. Orbits containing as few as six elements can, however, be constructed by choosing the high-symmetry vectors connecting the center of a regular octahedron to its vertices. These vectors are depicted in green in Fig.~\ref{fig:octahedron}. Such an orbit is an irreducible symmetric descriptor of a cube or an octahedron and has the form $\mathcal{O}_C=\{\pm x,\pm y, \pm z\}$ with $x,y$ and $z$ being mutually orthogonal. 

Being an even set, Corollary~\ref{cor:even_sets} implies that $\mathcal{H}_2,\mathcal{H}_4$ and $\mathcal{H}_6$ are the candidate SOCs for $\mathcal{O}_C$. We, however, have:
\begin{eqnarray}
	\mathcal{H}_2^{ij} &=& 2(x^ix^j+y^iy^j+z^iz^j) = 2\delta^{ij}
\end{eqnarray}
The only remaining candidates are thus $\mathcal{H}_4$ and $\mathcal{H}_6$. We prove that $\mathcal{H}_4$ is an SOC of $\mathcal{O}_C$ as follows. Let $\mathcal{O}_C=\{\pm a_p\}_{p=1}^3$ and $\mathcal{O}_{C'}=\{\pm b_q\}_{q=1}^3$ be the corresponding orbits for the two cubes $C$ and $C'$  with $a_p^ia_q^i=b_p^ib_q^i=\delta_{pq}$ and note that:
\begin{eqnarray}
\norm{\mathcal{H}_4(\mathcal{O}_C)-\mathcal{H}_4(\mathcal{O}_{C'})}_F^2 &=& 16\left[3-\sum_{p,q=1}^3\xi_{pq}^4\right]\label{eq:octa:objective}
\end{eqnarray}
with $\xi_{pq}=a_p^Tb_q$. For every $q$ we have $\sum_{p=1}^3\xi_{pq}^2=1$. As shown in Appendix~\ref{appendix:octaopt}, (\ref{eq:octa:objective}) is zero only if $\xi_{pq}^2$'s are a permutation of $(1,0,0)$ for each $q$ i.e. only if $\mathcal{O}_{C}=\mathcal{O}_{C'}$.

\begin{figure}
	\begin{center}
		\includegraphics[width=.4\textwidth]{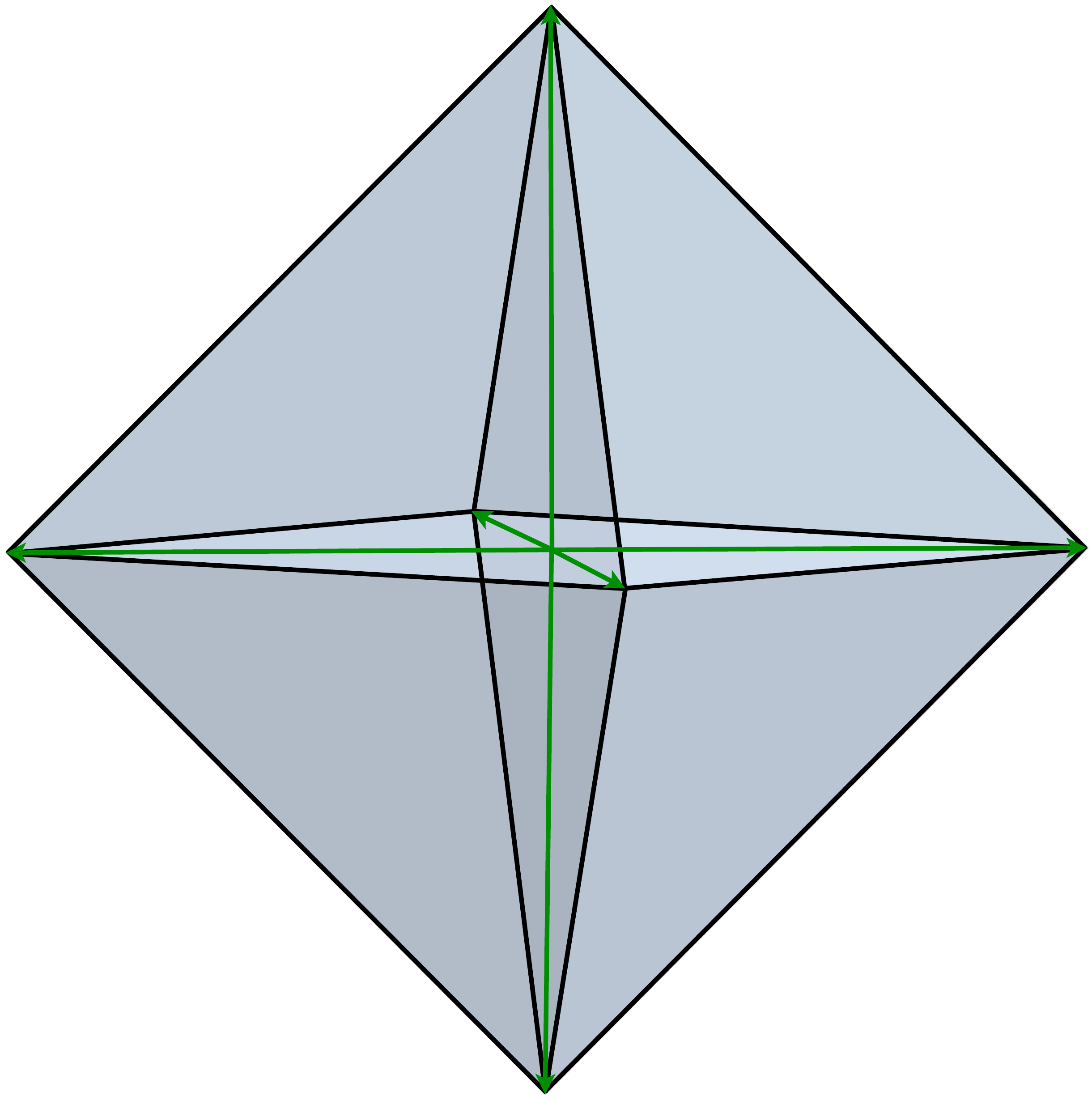}
		\caption{\label{fig:octahedron}A regular octahedron with the rotation group $O$. The smallest orbit has six elements that are depicted in green.}
	\end{center}
\end{figure}


\subsubsection{Icosahedral Symmetry $I$}
$I$ corresponds to the rotational symmetry of a regular icosahedron and a regular dodecahedron, and has $60$ elements. Orbits with as few as twelve elements can, however, be constructed from the vectors connecting the center of an icosahedron to its twelve vertices (Fig.~\ref{fig:icosahedron}). The corresponding orbit denoted by $\mathcal{O}_I=\{\pm a_i\}_{i=1}^6$ is an irreducible symmetric descriptor of an icosahedron (dodecahedron). A prototypical set of such unit $a_i$'s are given by $(\alpha,\pm\beta,0), (0,\alpha,\pm\beta), (\pm\beta,0,\alpha)$ with $\alpha=1/\sqrt{1+\phi^2}$, $\beta=\phi/\sqrt{1+\phi^2}$ and $\phi=(\sqrt5+1)/2$. The choice of $a_i$'s are arbitrary in the sense that both $+a_i$ and $-a_i$ are valid choices. Also note that $a_i^Ta_j=\pm1/\sqrt5$ for $i\neq j$.

As an even set, the SOC of $\mathcal{O}_I$ will be amongst $\mathcal{H}_2,\cdots,\mathcal{H}_{12}$ according to Corollary~\ref{cor:even_sets}. However one can show that:
\begin{eqnarray}
	\mathcal{H}_2^{ij} &=& 2\sum_{p=1}^6a_p^ia_p^j= 4\delta^{ij}\\
	\mathcal{H}_4^{ijkl} &=&2\sum_{p=1}^6a_p^ia_p^ja_p^ka_p^l=\frac45\left(\delta^{ij}\delta^{kl}+\delta^{ik}\delta{jl}+\delta^{il}\delta^{jk}\right)
\end{eqnarray}
With an approach similar to what was used for the tetrahedral and octahedral symmetries, we prove that $\mathcal{H}_6$ is the SOC for $\mathcal{O}_I$. Let $\mathcal{O}_I=\{\pm a_p\}_{p=1}^6$ and $\mathcal{O}_{I'}=\{\pm b_q\}_{q=1}^6$ be two such orbits. We have:
\begin{eqnarray}
	\norm{\mathcal{H}_6(\mathcal{O}_I)-\mathcal{H}_6(\mathcal{O}_{I'})}_F^2 &=& 128\sum_{p,q=1}^6\left[\frac{156}{125}-\xi_{pq}^3\right]\label{eq:ico:obj}
\end{eqnarray} 
with $\xi_{pq}=[a_p^Tb_q]^2$. Contracting $\mathcal{H}_2(\mathcal{O}_I)$ and $\mathcal{H}_4(\mathcal{O}_I)$ with $b_q^ib_q^j$ and $b_q^ib_q^jb_q^kb_q^l$ yields:
\begin{eqnarray}
	\sum_{p=1}^6 \xi_{pq} &=& 2\\
	\sum_{p=1}^6 \xi_{pq}^2 &=& \frac65
\end{eqnarray}
According to Appendix~\ref{appendix:icosa}, (\ref{eq:ico:obj}) can only be zero if $\xi_{pq}$'s are a permutation of $(1,\frac15,\frac15,\frac15,\frac15,\frac15)$ for each $q$ i.e. if $\mathcal{O}_I=\mathcal{O}_{I'}$. 

\begin{figure}
	\begin{center}
		\includegraphics[width=.4\textwidth]{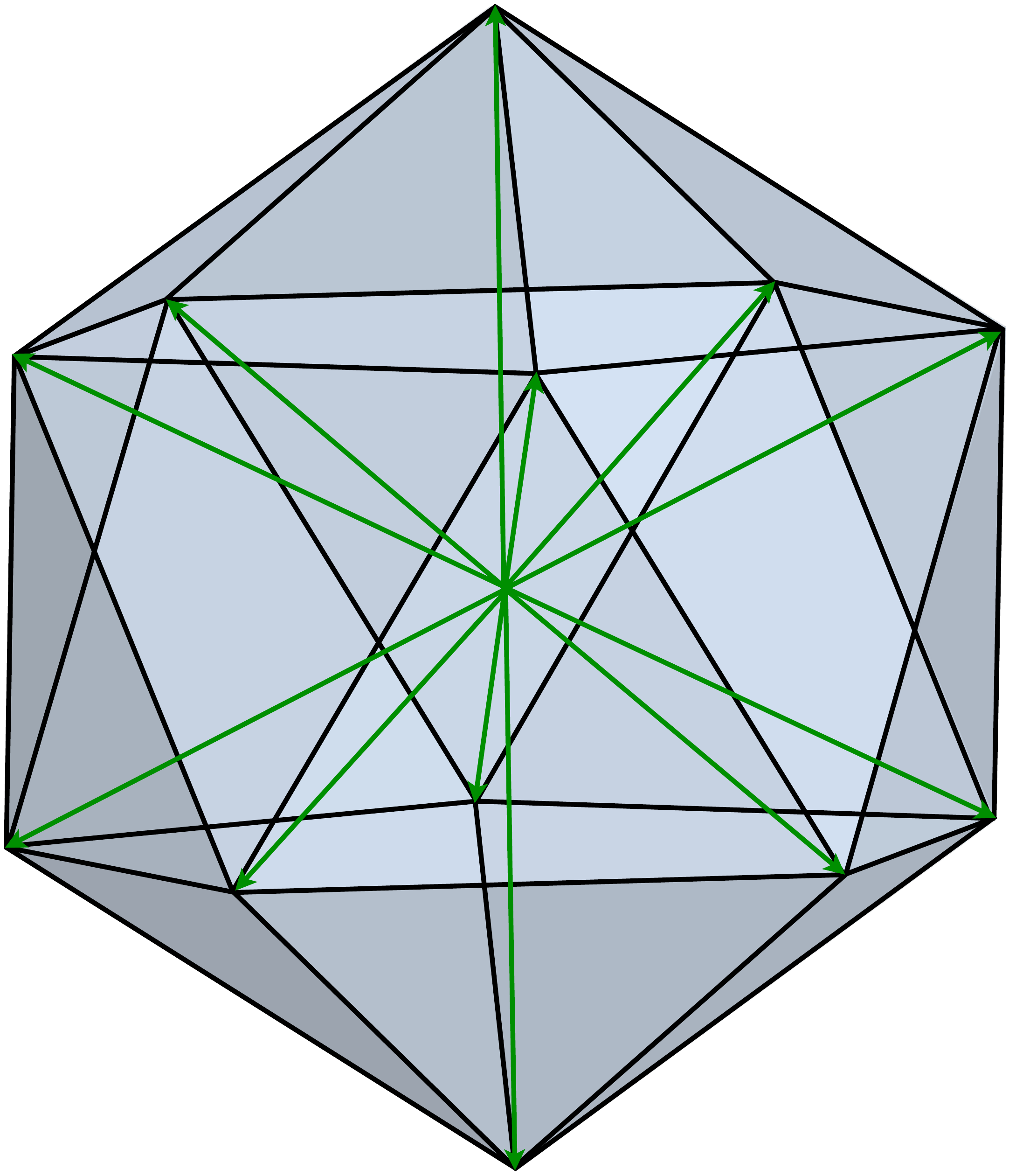}
		\caption{\label{fig:icosahedron}A regular icosahedron with the rotation group $I$. The smallest orbit has twelve elements that are depicted in green.}
	\end{center}
\end{figure}

\begin{table*}
	\caption{
	\label{table:2D3DSOCs}
	Strong Orientational Coordinates of Two- and Three-dimensional Rotation Groups. The `optimal SOC' corresponds to the smallest-rank strong rotational coordinate. The `upper-bound SOC', however, corresponds to the upper-bound rank predicted in Section~\ref{subsection:homogeneous_tensors}. The vectors given in the column $\mathcal{W}$ correspond to a typical orbit. Also note that $\alpha=1/\sqrt{1+\phi^2}, \beta=\phi/\sqrt{1+\phi^2}$ with $\phi=(\sqrt5+1)/2$.}
	\begin{tabular}{cccccc}
		\hline\hline 
		~~~$d$~~~~ & $\mathcal{G}$ & Prototypes & $\mathcal{W}$ & Optimal SOC & Upper bound SOC \\
		\hline
		\multirow{2}{*}{$2$} & \multirow{2}{*}{$C_n$} & \multirow{2}{*}{regular $n$-gon} & $\{v_p\}_{p=1}^n$ & \multirow{2}{*}{$\mathcal{H}_n=\sum_{p=1}^nv_p^n$} & $\mathcal{H}_n$, $n$ even \\
		&&& $v_p\left(\cos\frac{2\pi k}{n},\sin\frac{2\pi k}{n}\right)$ && $\mathcal{H}_{2n-1}$, $n$ odd\\
		\hline
		\multirow{3}{*}{$3$} & \multirow{3}{*}{$C_n$} & \multirow{3}{*}{regular $n$-gonal pyramid} & $\{v_p\}_{p=1}^n\cup\{z\}$ &$(\mathcal{H}_n,z)$  & $(\mathcal{H}_n,z)$, $n$ even \\
		&&& $v_p\left(\cos\frac{2\pi k}{n},\sin\frac{2\pi k}{n},0\right)$ &$\mathcal{H}_n=\sum_{p=1}^nv_p^n$& $(\mathcal{H}_{2n-1},z)$, $n$ odd\\
		&&& $z(0,0,1)$ &  &\\
		\hline
		\multirow{2}{*}{$3$} & \multirow{2}{*}{$D_n, n>2$} & regular $n$-gonal pyramid & $\{v_p\}_{p=1}^n$ & \multirow{2}{*}{$\mathcal{H}_n=\sum_{p=1}^nv_p^n$} & $\mathcal{H}_n$, $n$ even \\
		&&regular $n$-gonal prism& $v_p\left(\cos\frac{2\pi k}{n},\sin\frac{2\pi k}{n},0\right)$ && $\mathcal{H}_{2n-1}$, $n$ odd\\
		\hline
		\multirow{3}{*}{$3$} & \multirow{3}{*}{$D_2$} & \multirow{3}{*}{rectangular parallelepiped} & $\{\pm z\}\cup\{\pm y\}$ & \multirow{3}{*}{$(zz,yy)$} & \multirow{3}{*}{$(zz,yy)$} \\
		&&& $z(0,0,1)$ && \\
		&&& $y(0,1,0)$ &&\\
		\hline
		$3$ & $C_{\infty}$ & Cone, Hemisphere & $\{z\}, z(0,0,1)$ & $z$ & $z$\\ 
		\hline
		$3$ & $D_{\infty}$ & Cylinder & $\{\pm z\}, z(0,0,1)$ & $zz$ & $zz$\\ 
		\hline
		\multirow{5}{*}{$3$} & \multirow{5}{*}{$T$} & \multirow{5}{*}{Regular tetrahedron} & $\{a_p\}_{p=1}^4$ & \multirow{5}{*}{$\mathcal{H}_3=\sum_{p=1}^4a_pa_pa_p$} & \multirow{5}{*}{$\mathcal{H}_7$} \\ 
		&&& $a_1(\frac{1}{\sqrt3},\frac{1}{\sqrt3},\frac{1}{\sqrt3})$&\\
		&&& $a_2(\frac{1}{\sqrt3},-\frac{1}{\sqrt3},-\frac{1}{\sqrt3})$&\\
		&&& $a_3(-\frac{1}{\sqrt3},-\frac{1}{\sqrt3},\frac{1}{\sqrt3})$&\\
		&&& $a_4(-\frac{1}{\sqrt3},\frac{1}{\sqrt3},-\frac{1}{\sqrt3})$&\\
		\hline
		\multirow{4}{*}{$3$} & \multirow{4}{*}{$O$} & & $\{\pm a_p\}_{p=1}^3$ & \multirow{4}{*}{$\mathcal{H}_4=2\sum_{p=1}^3a_pa_pa_pa_p$} & \multirow{4}{*}{$\mathcal{H}_6$}\\
		&& Octahedron & $a_1(1,0,0)$ & & \\
		&& Cube & $a_2(0,1,0)$ & & \\
		&&  & $a_3(0,0,1)$ & & \\
		\hline
		\multirow{4}{*}{$3$} & \multirow{4}{*}{$I$} & & $\{\pm a_p\}_{p=1}^6$ & \multirow{4}{*}{$\mathcal{H}_6=2\sum_{p=1}^6a_pa_pa_pa_pa_pa_p$} & \multirow{4}{*}{$\mathcal{H}_{12}$}\\
		&& Regular Icosahedron & $a_{1,2}(\alpha,\pm\beta,0)$ & & \\
		&& Regular Dodecahedron & $a_{3,4}(0,\alpha,\pm\beta)$ & & \\
		&&  & $a_{5,6}(\pm\beta,0,\alpha)$ & & \\
		\hline
		\hline
	\end{tabular}
\end{table*}


\subsection{Ordered Arrangements and Order Parameters\label{subsection:OrderParameter}}
As explained in Section~\ref{section:introduction}, orientationally-ordered arrangements of symmetric objects can be mathematically described using the distribution functions that have been historically expressed in terms of non-bijective orientational coordinates such as polar (or Euler) angles. Here, we express such functions in terms of the SOCs derived in Section~\ref{subsection:symmetric_objects}, and we quantify the extent of orientational order by computing the ensemble averages of certain moments of such SOCs. 

An arrangement of (symmetric) objects is called \textbf{rotationally isotropic} if each particle can take all permissible orientations with equal probability. This can be characterized by a uniform distribution function. However in an orientationally-ordered arrangement of (symmetric) objects, orientational symmetry is broken and each particle tends to preferentially take certain orientations more frequently. Such a preference can be characterized by a non-uniform distribution function in terms of an SOC. We therefore define a structure- or a phase- as follows.

\begin{defn}\label{defn:phase}
A \textbf{structure} or \textbf{phase} of an object $\mathcal{R}$ is characterized by a probability density function $p_0(\mathcal{H}_{m_1},\cdots, \mathcal{H}_{m_N}; \Omega)$ where $\Omega$ stands for all the geometric features needed for the macroscopic characterization of the structure.
\end{defn}

It is necessary to emphasize that the notion of a phase in Definition~\ref{defn:phase} entails only the global orientational characteristics of an arrangement, and should not be confused with the thermodynamic notion of a phase that can entail both translational and orientational order.

\begin{example}\label{example:complete_match_structure}
Let $\mathcal{R}$ be a symmetric object with a normalized irreducible symmetric descriptor containing a single equivalence class ($N=1$) and let $\mathcal{H}_m$ be a strong orientational coordinate constructed from $\mathcal{N}_{\mathcal{R}}$. The density function $p_0(\mathcal{H}_m;\widehat{\mathcal{H}}_m)=\delta(\mathcal{H}_m-\widehat{\mathcal{H}}_m)$ defines an arrangement of the object $\mathcal{R}$ where all objects have the same orientation with $\widehat{\mathcal{R}}$ and $\Omega=\widehat{\mathcal{H}}_m$.
\end{example}

We now outline the procedure that can be used for deriving orientational order parameters (OOPs) for a phase. Let $\mathcal{H}_{\mathcal{R}}:\equiv(\mathcal{H}_{m_1},\cdots,\mathcal{H}_{m_N})$ be an SOC for $\mathcal{R}$ and let $p_0(\mathcal{H}_{\mathcal{R}}; \Omega)$ be a phase. Also consider $M(\mathcal{H}_{\mathcal{R}})$, a tensorial function of $\mathcal{H}_{\mathcal{R}}$, with the property that $\mathfrak{M}(\Omega)=\langle M(\mathcal{H}_{\mathcal{R}})\rangle_0$ satisfies the condition that it is a strong descriptor of $\Omega$, i.e.,~$\mathfrak{M}_1=\mathfrak{M}_2$ if and only if $\Omega_1=\Omega_2$. In other words, $\mathfrak{M}(\Omega)$ must be invariant under the transformations that keep $\Omega$ unchanged. This is to assure that all distinct geometric instances of a phase are distinguishable by $\mathfrak{M}$. However, this condition can be relaxed if one is only interested in certain structural features of a phase. 

Now let $\mathcal{R}_1,\mathcal{R}_2,\cdots,\mathcal{R}_n$ be an arrangement of (symmetric) particles and define the \textbf{experimental order estimator} as:
\begin{eqnarray}
\overline{M} &=& \frac1n\sum_{i=1}^nM(\mathcal{H}_{\mathcal{R}_i})
\end{eqnarray}
The more perfect the ordering of $\mathcal{R}_1,\mathcal{R}_2,\cdots,\mathcal{R}_n$ is, the closer will $\overline{M}$ be to $\mathfrak{M}$. Therefore, one would expect $||\overline{M}-\mathfrak{M}||_F$ to be smaller in more perfect arrangements. Henceforth, one can formulate the problem of identifying the underlying geometric features of a phase as:
\begin{eqnarray}
\min_{\Omega}\frac{\norm{\overline{M}-\mathfrak{M}_{\Omega}}_F}{\norm{\mathfrak{M}_I-\mathfrak{M}_{\Omega}}_F}\label{eq:core_optimization}
\end{eqnarray}
 with $\mathfrak{M}_I=\langle M\rangle_{\text{isotropic}}$. One can define \textbf{scalar order parameter} of the phase as:
\begin{eqnarray}
\kappa_{\Omega}&=& 1-\frac{\norm{\overline{M}-\mathfrak{M}_{\Omega^*}}^2_F}{\norm{\mathfrak{M}_I-\mathfrak{M}_{\Omega^*}}^2_F}\label{eq:scalar_OP}
\end{eqnarray}
where $\Omega^*$ is the minimizer in (\ref{eq:core_optimization}). Note that $\kappa=0$ for a completely isotropic system while $\kappa=1$ when matching is perfect. For a certain subclass of phases where $\norm{\mathfrak{M}_{\Omega}}$ is constant, one can simplify (\ref{eq:core_optimization}) to:
\begin{eqnarray}
\max_{\Omega}\overline{M}\odot\mathfrak{M}_{\Omega}\label{eq:OOP_maximize}
\end{eqnarray}

In general if the conditions outlined above are established for a given $\mathfrak{M}$, one can use (\ref{eq:core_optimization}) or (\ref{eq:OOP_maximize}) to obtain $\mathfrak{M}_{\Omega^*}$ and (\ref{eq:scalar_OP}) to calculate the scalar order parameter. We will show the procedure of solving (\ref{eq:core_optimization}) or (\ref{eq:OOP_maximize}) in multiple examples at the end of this section. Before doing so, however, we give explicit formulae for calculating $\mathfrak{M}_I$ for the $M$'s that are sums of $r$-adic products.

\begin{defn}\label{defn:symmetrized_tensor_polynom}
Let $v_1,v_2,\cdots,v_m\in\mathbb{R}^d$ be arbitrary vectors and $a_1,a_2,\cdots,a_m$ be given nonnegative integers adding up to $a$. A \textbf{symmetrized tensor polynomial} $\mathscr{S}_{a_1,a_2,\cdots,a_m}^a(v_1,v_2,\cdots,v_m)$ is defined as the sum of all possible direct products of the form $v_{i_1}v_{i_2}\cdots v_{i_a}$ where exactly $a_1$ of $i_j$'s are one, $a_2$ of $i_j$'s are two, etc. The number of distinct terms in such a polynomial is $\binom{a}{a_1,a_2,\cdots,a_m}=a!/\prod_{q=1}^ma_q!$. Table~\ref{table:symm:poly} gives the list of symmetrized tensor polynomials for $a<4$ obtained from the following proposition. 
\end{defn}

\begin{prop}
The symmetrized tensor polynomial defined in Definition~\ref{defn:symmetrized_tensor_polynom} is given by:
\begin{eqnarray}
\mathscr{S}^{a}_{a_1,a_2,\cdots,a_m}(v_1,v_2,\cdots,v_m) &=& \left\{
	\begin{array}{ll}
	\mathscr{S}^{a}_{a_1,a_2,\cdots,a_{k-1},0,a_{k+1},\cdots,a_m}(v_1,\cdots,v_{k-1},v_{k+1},\cdots,v_m) & a_k=0\\
	\sum_{k=1}^mv_k\mathscr{S}^{a-1}_{a_1,\cdots,a_{k-1},a_k-1,a_{k+1},\cdots,a_m}(v_1,\cdots,v_m)& a_i>0\text{~for all i}
	\end{array}
\right.\notag
\end{eqnarray}
with $\mathscr{S}^0_{0,0,\cdots,0}(v_1,v_2,\cdots,v_m)=1$.
\end{prop}
\begin{proof}
The first assertion follows from the definition. For the second assertion look at each term of $\mathscr{S}^{a}_{a_1,a_2,\cdots,a_m}(v_1,v_2,\cdots,v_m)$ and group them based on their initial multiplier. 
\end{proof}

\begin{widetext}
\begin{prop}\label{prop:isotropic}
Let $t\in\mathbb{R}^d,d=2,3$ be a unit vector and $k\in\mathbb{Z}^{\ge0}$, then :
\begin{enumerate}
\item $\langle{t}^{2k+1}\rangle_{\text{isotropic}} = 0$.
\item $\displaystyle\langle t^{2k}\rangle_{\text{isotropic}}=\sum_{p=0}^k\frac{(2p)!(2k-2p)!}{4^kp!k!(k-p)!}\mathscr{S}^{2k}_{2p,2k-2p}({x},{y})$~~~~~for $d=2$.
\item $\displaystyle\langle{t}^{2k}\rangle_{\text{isotropic}}=\frac{1}{2k+1}\sum_{m=0}^k\sum_{n=0}^{k-m}\frac{\binom{k}{m,n,k-m-n}}{\binom{2k}{2m,2n,2k-2m-2n}}\mathscr{S}^{2k}_{2m,2n,2k-2m-2n}({x},{y},{z})$~~~for $d=3$.
\end{enumerate}
where $({x},{y})$ and $({x},{y},{z})$ are orthonormal bases for $\mathbb{R}^2$ and $\mathbb{R}^3$ respectively. 
\end{prop}
\end{widetext}
\begin{proof}
Observe that $\langle{t}^n\rangle=(1/S_{d-1})\int_{S^d}\textbf{t}^nd\Omega$ where $S_{d-1}$ is the surface of the $d$-sphere. Contributions from two hemispheres to the integral cancel out for odd $n$ and (1) follows. For $n=2k$ and $d=2$ we have ${t}=\cos\theta{~x}+\sin\theta{~y}$ and:
\begin{eqnarray}
\langle{t}^{2k}\rangle_I &=&\frac1{2\pi}\sum_{p=0}^kI_{2p,2k-2p}\mathscr{S}^{2k}_{2p,2k-2p}({x},{y})\notag
\end{eqnarray}
where $I_{2p,2k-2p}=\int_0^{2\pi}\cos^{2p}\theta\sin^{2k-2p}\theta d\theta$. The odd terms vanish because $I_{2p,2q+1}=I_{2p+1,2q}=I_{2p+1,2q+1}=0$ and (2) follows from Eq.~(\ref{eq:app:int_I_2m_2n}). For $d=3$, ${t}(\theta,\phi)=\sin\theta\cos\phi~{x}+\sin\theta\sin\phi~{y}+\cos\theta~{z}$ and:
\begin{widetext}
\begin{eqnarray}
\langle{t}^{2k}\rangle_I &=& \frac1{4\pi}\sum_{p=0}^{k}\sum_{q=0}^{k-p}I_{2p,2q}J_{2k-2p-2q,2p+2q+1}\mathscr{S}_{2p,2q,2k-2p-2q}^{2k}({x},{y},{z})\notag\\
&\overset{(a)}{=}& \sum_{p=0}^k\sum_{q=0}^{k-p}\frac{(2p)!(2q)!}{4^{p+q}p!q!(p+q)!}\frac{4^{p+q}(p+q)!(2k-2p-2q)!k!}{(k-p-q)!(2k+1)!}\mathscr{S}_{2p,2q,2k-2p-2q}^{2k}({x},{y},{z})\label{eq:prop3temp}
\end{eqnarray}
$(a)$ follows from~(\ref{eq:app:int_I_2m_2n}) and~(\ref{eq:app:int_J_2m_2np1}). Rearranging~(\ref{eq:prop3temp}) completes the proof.
\end{widetext}
\end{proof}

This procedure yields the widely known rotation-invariant isotropic tensors given in Table~\ref{tab:isotropic}, and can be thought of an algorithmic way of constructing such isotropic tensors for large values of $k$. Using this proposition, one can thus calculate $\mathfrak{M}_I$ for any given tensor that is a sum of $n$-adic products of unit vectors, including the moments of homogeneous tensors defined in this work. One can therefore always subtract $\mathfrak{M}_I$ in the definition of $M$ so that $\norm{M}_F$ on its own can be used as a measure of how anisotropic a certain arrangement of particles is. The rest of this section is devoted to some examples of how (\ref{eq:core_optimization}) and (\ref{eq:OOP_maximize}) can be formulated and solved. But before doing so, we outline the following useful result that can be used to calculate the expected value of a $k$-adic power of a vector that can uniformly rotate around a rotation axis.
\begin{prop}\label{prop:axial}
Let $z,t\in\mathbb{R}^3$ be unit vectors with $t$  uniformly distributed on the plane perpendicular to $z$. For a vector $v=\alpha z+\beta t$, $\langle v^k\rangle$ is given by:
\begin{eqnarray}
\langle v^k\rangle_z&=&\sum_{l=0}^{\lfloor k/2\rfloor}\sum_{m=0}^l\alpha^{k-2l}\beta^{2l}\frac{(2m)!(2l-2m)!}{4^ll!m!(l-m)!}\notag\\ &&\times\mathscr{S}^{k}_{2m,2l-2m,k-2l}(x,y,z)\notag
\end{eqnarray}
with $x$ and $y$ being a pair of orthonormal vectors perpendicular to $z$.
\end{prop}
\begin{proof}
Expand $v^k$ in terms of $t$ and $z$ and use case 2 of Proposition~\ref{prop:isotropic} to complete the proof.
\end{proof}

\begin{table}
	\caption{\label{table:symm:poly}List of $\mathscr{S}_{a_1,a_2,a_3}^{a}(x,y,z)$'s for $a<4$ and $x,y,z\in\mathbb{R}^3$.}
	\begin{tabular}{ccccl}
		\hline\hline
		$a$~~&~~$a_1$~~&~~$a_2$~~&~~$a_3$~~~~~&$\mathscr{S}_{a_1,a_2,a_3}^a(x,y,z)$\\
		\hline
		$0$ & $0$ & $0$ & $0$ & $1$ \\
		$1$ & $1$ & $0$ & $0$ & $x$ \\
		$1$ & $0$ & $1$ & $0$ & $y$ \\
		$1$ & $0$ & $0$ & $1$ & $z$ \\
		$2$ & $2$ & $0$ & $0$ & $xx$ \\
		$2$ & $0$ & $2$ & $0$ & $yy$ \\
		$2$ & $0$ & $0$ & $2$ & $zz$ \\
		$2$ & $1$ & $1$ & $0$ & $xy+yx$ \\
		$2$ & $1$ & $0$ & $1$ & $xz+zx$ \\
		$2$ & $0$ & $1$ & $1$ & $yz+zy$ \\
		$3$ & $3$ & $0$ & $0$ & $xxx$ \\
		$3$ & $0$ & $3$ & $0$ & $yyy$ \\
		$3$ & $0$ & $0$ & $3$ & $zzz$ \\
		$3$ & $2$ & $1$ & $0$ & $xxy+xyx+yxx$\\
		$3$ & $2$ & $0$ & $1$ & $xxz+xzx+zxx$\\
		$3$ & $1$ & $2$ & $0$ & $xyy+yxy+yyx$ \\
		$3$ & $1$ & $0$ & $2$ & $xzz+zxz+zzx$\\
		$3$ & $0$ & $2$ & $1$ & $yyz+yzy+zyy$\\
		$3$ & $0$ & $1$ & $2$ & $yzz+zyz+zzy$\\
		$3$ & $1$ & $1$ & $1$ & $xyz+xzy+yxz+yzx+zxy+zyx$ \\
		\hline
	\end{tabular}
\end{table}

\begin{table}
	\caption{\label{tab:isotropic}$\langle{t}^{2k}\rangle_{\text{isotropic}}$ for a few values of $k$. Here $S^{ijkl}=\delta^{ij}\delta^{kl}+\delta^{ik}\delta^{jl}+\delta^{il}\delta^{jk}$ and $T^{ijklmn}=\delta^{ij}S^{klmn}+\delta^{ik}S^{jlmn}+\delta^{il}S^{jkmn}+\delta^{im}S^{jkln}+\delta^{in}S^{jklm}$.}
	\begin{tabular}{ccc}
	\hline\hline
	$k$ & $\langle{t}^{2k}\rangle_{\text{isotropic}}^{\mathbb{R}^2}$ & $\langle{t}^{2k}\rangle_{\text{isotropic}}^{\mathbb{R}^3}$\\
	\hline
	$1$ & $\frac12\delta^{ij}$ & $\frac13\delta^{ij}$ \\
	$2$ & $\frac18S^{ijkl}$ & $\frac1{15}S^{ijkl}$\\
	$3$ & $\frac1{48}T^{ijklmn}$ & $\frac1{105}T^{ijklmn}$\\
	\hline
	\end{tabular}
\end{table}


\subsubsection{Uniaxial Nematics\label{section:nematic}}
Rodlike molecules or nanoparticles can assemble into a rotationally ordered phase known as the \emph{uniaxial nematic} phase in which the rotation axes of all particles are on average aligned to a common vector called a \emph{director}~\cite{Onsager1949,EppengaFrenkel1984,BatesFrenkel1998}. As explained in Section~\ref{subsubsection:CinfDinf}, $\mathcal{H}_2(\{\pm z_i\})=2z_i^2$ is the proper SOC for a rod. As for the uniaxial nematic phase, it is fully specified by $u$, the director, and henceforth, $\Omega=\{u\}$. In a perfect nematic phase, all particles will align along the same director. This perfectly fits into the class of structures described in Example~\ref{example:complete_match_structure}. We therefore have:
\begin{eqnarray}
M^{ij} &=& \frac12\mathcal{H}_2^{ij}-\frac12\mathcal{H}_{2,I}^{ij} = z^iz^j-\frac13\delta^{ij}\notag\\
\mathfrak{M}_{\Omega}^{ij} &=& u^iu^j-\frac13\delta^{ij}\notag\\
\overline{M}^{ij} &=& \frac1N\sum_{p=1}^Nz_p^iz_p^j-\frac13\delta^{ij}\notag
\end{eqnarray}
Since $\norm{\mathfrak{M}_{\Omega}}_F$ is constant, we can use the optimization problem (\ref{eq:OOP_maximize}) which takes the form:
\begin{eqnarray}
	\begin{array}{ll}
		\max & u^T\overline{M}u\\
		\text{subject to} & u^Tu=1
	\end{array}
\end{eqnarray} 
which can be solved by using Lagrange multipliers:
\begin{eqnarray}
\mathcal{L}(u,\lambda) &=& u^T\overline{M}u-\lambda(u^Tu-1)\notag\\
\nabla_u\mathcal{L} &=& 2\overline{M}u-2\lambda u = 0 \implies (\overline{M}-\lambda I)u=0\notag
\end{eqnarray}
This implies that $u$ should be an eigenvector of $\overline{M}$. The largest eigenvalue of $\overline{M}$ , $\lambda_1$ maximizes $u^T\overline{M}u$. The scalar order parameter is given by:
\begin{eqnarray}
	\kappa_{\text{nematic}} &=& 1-\tfrac{3}{2}\left[\left(\lambda_1-\tfrac23\right)^2+\left(\lambda_2+\tfrac13\right)^2+\left(\lambda_3+\tfrac13\right)^2\right]\notag\\&&
\end{eqnarray}
with $\lambda_1\ge\lambda_2\ge\lambda_3$ being the eigenvalues of $\overline{M}$. This formula penalizes any divergence of $\lambda_i$'s from their 'optimal' values of $(2/3,-1/3,-1/3)$ in a perfect nematics. 

\begin{figure}
	\begin{center}
		\includegraphics[width=.4\textwidth]{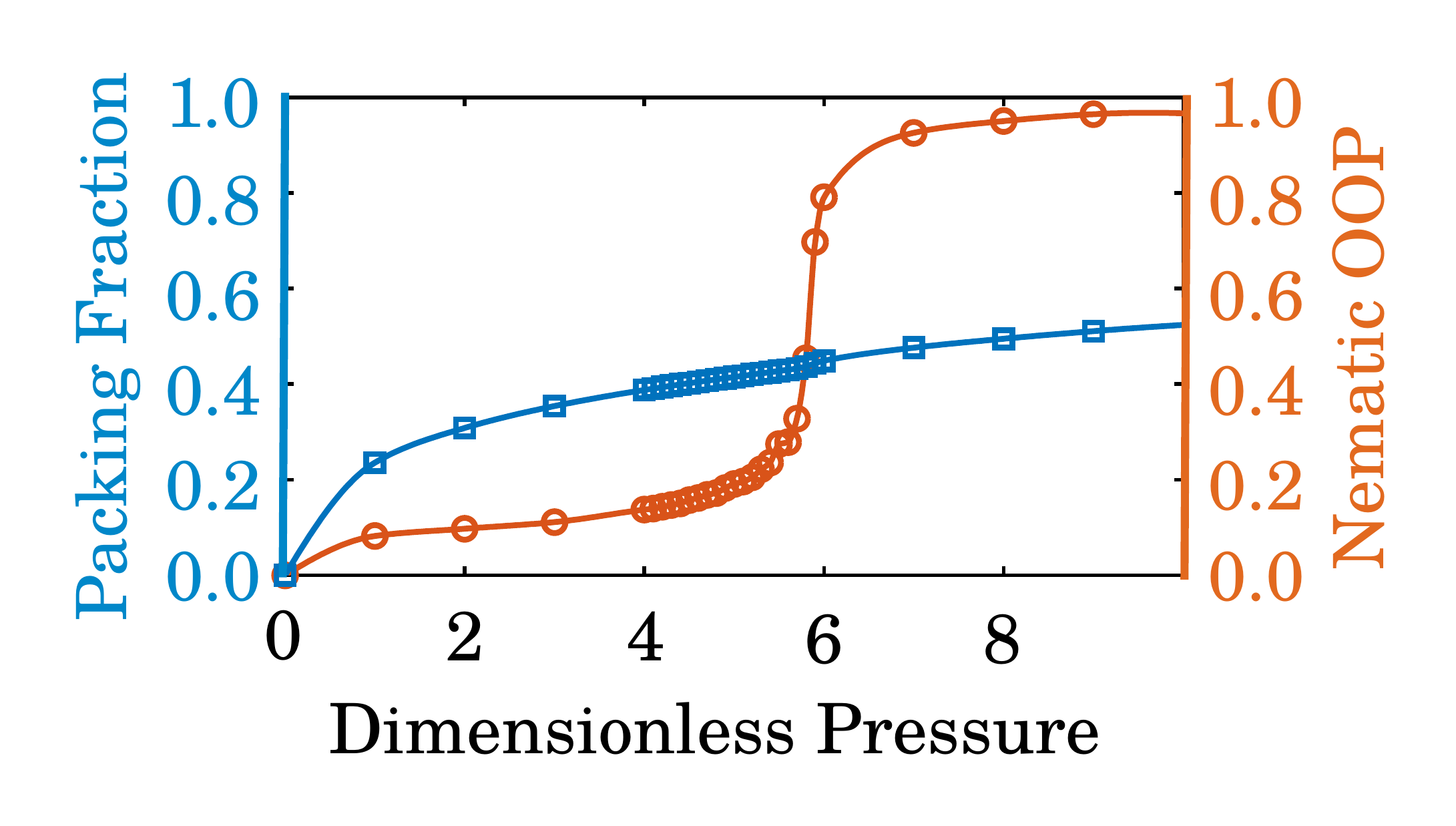}
		\caption{\label{fig:ellipsoid_OOP} Nematic OOP (circles) and packing fraction (squares) as a function of dimensionless pressure for a system of $512$ hard ellipsoids.}
	\end{center}
\end{figure}

As a numerical example, we conduct isothermal isobaric Monte Carlo simulations of a system of $512$ hard ellipsoids with an aspect ratio $a/b=4$. These highly elongated building blocks are known to assemble into the nematic liquid crystal at moderate pressures~\cite{RaduSchilling2009}. Each MC step consists of 512 trial particle sweeps (translation and rotation with equal probability), and one box rescale move, on average. The dimensionless pressure is defined based on the shortest axis of an ellipsoid that is chosen as the length scale. Snapshots are stored every $20,\!000$ MC steps and the nematic order parameter is computed for each configuration. The calculated nematic OOPs are depicted in Fig.~\ref{fig:ellipsoid_OOP}. A pronounced increase in OOP is observed at around $P^*=PV/k_BT\approx5.9$, even though the change in packing fraction is not as pronounced.  


\subsubsection{Cubatic Phase}
At sufficiently large densities, hard cubes can assemble into the \emph{cubatic phase}; a structure in which all particles assume the same orientation, while no long-range translational order exists~\cite{EscobedoNatureMaterials2011}. As explained in Section~\ref{section:SOC:octahedral}, $\mathcal{H}_4(\{\pm x,\pm y,\pm z\})$ is the SOC of a cube. The cubatic phase can  be described by $\Omega=\{\pm v_p\}_{p=1}^3$, with $v_p$'s being three orthonormal vectors, and $M$, $\mathfrak{M}_{\Omega}$ and $\overline{M}$ can be defined as:
\begin{eqnarray}
M^{ijkl} &=& 2\sum_{p=1}^3u_p^iu_p^ju_p^ku_p^l-\frac25\left(\delta^{ij}\delta^{kl}+\delta^{ik}\delta^{jl}+\delta^{il}\delta^{jk}\right)\notag\\
\mathfrak{M}_{\Omega} &=& 2\sum_{p=1}^3v_p^iv_p^jv_p^kv_p^l-\frac25\left(\delta^{ij}\delta^{kl}+\delta^{ik}\delta^{jl}+\delta^{il}\delta^{jk}\right)\notag\\
\overline{M}^{ijkl} &=& \frac2N\sum_{p=1}^3\sum_{q=1}^Nu_{p,q}^iu_{p,q}^ju_{p,q}^ku_{p,q}^l-\frac25\left(\delta^{ij}\delta^{kl}+\delta^{ik}\delta^{jl}+\delta^{il}\delta^{jk}\right)
\end{eqnarray}
Here, $u_{p,q}$'s correspond to the orthogonal vectors describing the orientation of particle $q$. Since $\mathfrak{M}_{\Omega}$ is constant, the following optimization problem can be solved for quantifying cubatic order:
\begin{eqnarray}
	\begin{array}{lll}
		\max & \overline{M}^{ijkl}\sum_{p=1}^3v_p^iv_p^jv_p^kv_p^l\\
		\text{subject to} & v_p^iv_q^j = \delta_{pq} & p,q=1,2,3
	\end{array}\label{eq:cubatic_opt}
\end{eqnarray}
The constraints of (\ref{eq:cubatic_opt}) can, however, be equivalently reformulated as $\sum_{p=1}^3v_pv_p^T=I$. (Multiply both sides by $v_q$ and use the linear independence of $v_p$'s to conclude that $v_p^Tv_q=\delta_{pq}$.) The Lagrangian of the optimization problem is therefore given by:
\begin{eqnarray}
\mathcal{L} &=& \overline{M}^{ijkl}\sum_{p=1}^3v_p^iv_p^jv_p^kv_p^l - \sum_{i,j=1}^3\mu_{ij}\left[\sum_{p=1}^3v_p^iv_p^j-\delta^{ij}\right]\notag
\end{eqnarray}
with $\mu_{ij}=\mu_{ji}$. The derivative of the Lagrangian is given by:
\begin{eqnarray}
\frac{\partial\mathcal{L}}{\partial v_q^s} &=& \overline{M}^{sjkl}v_q^jv_q^kv_q^l+\overline{M}^{iskl}v_q^iv_q^kv_q^l+\overline{M}^{ijsl}v_q^iv_q^jv_q^l+\overline{M}^{ijks}v_q^iv_q^jv_q^k-\mu_{sj}v_q^j-\mu_{is}v_q^i\notag\\
&\overset{(a)}{=}& 4\overline{M}^{sijk}v_q^iv_q^jv_q^k-2\mu_{si}v_q^i\notag\\
\nabla_q\mathcal{L} &=& b_q-Cv_q
\end{eqnarray}
with:
\begin{eqnarray}
b_q = 4\left(
	\begin{matrix}
		\overline{M}^{1ijk}v_q^iv_q^jv_q^k\\
		\overline{M}^{2ijk}v_q^iv_q^jv_q^k\\
		\overline{M}^{3ijk}v_q^iv_q^jv_q^k
	\end{matrix}
\right),~C =2\left(
	\begin{matrix}
		\mu_{11} & \mu_{12} & \mu_{13}\\
		\mu_{21} & \mu_{22} & \mu_{23}\\
		\mu_{31} & \mu_{32} & \mu_{33}
	\end{matrix}
\right)
\end{eqnarray}
Note that $(a)$ follows from the invariance of  $\overline{M}$ under index permutation. $\nabla_q\mathcal{L}=0$ implies that $b_q=Cv_q$. Multiplying both sides by $v_q$ and summing over $q$ yields:
\begin{eqnarray}
	\sum_{q=1}^3b_qv_q^T &=& C\sum_{q=1}^3v_qv_q^T=C
\end{eqnarray}
Note that $C$ is not symmetric for an arbitrary set of orthogonal $v_q$'s and its symmetry is achieved when $v_q$'s are amongst the Karush-Kuhn-Tucker (KKT) solutions of (\ref{eq:cubatic_opt}). The global maximum of (\ref{eq:cubatic_opt}) can only be obtained if \emph{all} such KKT solutions are identified. We achieve this by adopting a Newton-Raphson scheme that is described below, and by performing a sufficient number of attempts using different initial guesses. First, we define $\zeta_{st}:=2(\mu_{st}-\mu_{ts})$ for $s\neq t$ and observe that:
\begin{eqnarray}
\frac{\partial\zeta_{st}}{\partial v_p^m} &=& 12\left(\overline{M}^{smij}v_p^iv_p^jv_p^t-\overline{M}^{tmij}v_p^iv_p^jv_p^s \right)+4\left(b_p^s\delta_{tm}-b_p^t\delta_{sm}\right)
\end{eqnarray}
We can therefore approximate $\zeta_{st}$ as
\begin{eqnarray}
\zeta_{st} &\approx& \zeta_{st,0}+\sum_{p=1}^3(\nabla_p\zeta_{st})^T_0(v_p-v_{p,0})+\cdots
\end{eqnarray}
The Newton-Raphson iteration can therefore be carried out by simultaneously solving the three equations given by $\zeta_{st}=0$ and the six equations ensuring the orthonormality of $v_q$'s, i.e.,~$\sum_qv_q^iv_q^j=\delta^{ij}$. We can, however, decrease the number of equations from nine to four by using quaternions, which are widely used in simulations of non-spherical particles. Unit quaternions are used to describe rigid-body rotations in three dimensions. Rotating a rigid body $\mathcal{R}$ using a unit quaternion $\textbf{q}=(q_1,q_2,q_3,q_4)$ maps every vector $v^{\circ}\in\mathcal{R}$ to $v=R(\textbf{q})v^{\circ}$. The rotation matrix $R(\textbf{q})$ is given by: 
\begin{eqnarray}
\left[\begin{array}{ccc}
        q_1^2+q_2^2-q_3^2-q_4^2 & 2(q_2q_3-q_1q_4) & 2(q_1q_3+q_2q_4) \\
        2(q_1q_4+q_2q_3) & q_1^2-q_2^2+q_3^2-q_4^2 & 2(q_3q_4-q_1q_2)\\
        2(q_2q_4-q_1q_3) & 2(q_1q_2+q_3q_4) & q_1^2-q_2^2-q_3^2+q_4^2
        \end{array}\right]\notag
\end{eqnarray}
One can thus consider $v_p$'s- and $\zeta_{st}$'s- as implicit functions $\textbf{q}$.  More specifically we have:
\begin{eqnarray}
\frac{\partial\zeta_{st}}{\partial q_u} &=& \sum_{p,m=1}^3\frac{\partial\zeta_{st}}{\partial v_p^m}\frac{\partial v_p^m}{\partial q_u}=\sum_{p=1}^3\left(\nabla_p\zeta_{st}\right)^T\frac{\partial v_p}{\partial q_u}\notag= \sum_{p=1}^3\left(\nabla_p\zeta_{st}\right)^T\frac{\partial R(\textbf{q})}{\partial q_u}v_p^{\circ}\notag
\end{eqnarray}
where $v_p^\circ$ correspond to a set of orthonormal vectors corresponding to $\textbf{q}=1$. In order to solve $C=C^T$ under the orthonormality constraint, one can solve the following four equations denoted by $\textbf{f}(\textbf{q})=\textbf{0}$:
\begin{eqnarray}
\textbf{f}(\textbf{q}) &=& \left(
		\begin{matrix}
			\zeta_{12}\\
			\zeta_{13}\\
			\zeta_{23}\\
			q_1^2+q_2^2+q_3^2+q_4^2-1
		\end{matrix}
\right)
\end{eqnarray}
The Newton-Raphson iteration can therefore be carried out using the following formula:
\begin{eqnarray}
\textbf{q}_{n+1} = \textbf{q}_n - \Phi(\textbf{q}_n)^{-1}\textbf{f}(\textbf{q}_n)\label{eq:cubatic:quatNewton}
\end{eqnarray}
with:
\begin{eqnarray}
\Phi(\textbf{q}) &=& \left(
	\begin{matrix}
		{\partial\zeta_{12}}/{\partial q_1} &
		{\partial\zeta_{12}}/{\partial q_2} &
		{\partial\zeta_{12}}/{\partial q_3} &
		{\partial\zeta_{12}}/{\partial q_4} \\
		{\partial\zeta_{13}}/{\partial q_1} &
		{\partial\zeta_{13}}/{\partial q_2} &
		{\partial\zeta_{13}}/{\partial q_3} &
		{\partial\zeta_{13}}/{\partial q_4} \\
		{\partial\zeta_{23}}/{\partial q_1} &
		{\partial\zeta_{23}}/{\partial q_2} &
		{\partial\zeta_{23}}/{\partial q_3} &
		{\partial\zeta_{23}}/{\partial q_4} \\
		2q_1 & 2q_2 & 2q_3 & 2q_4
	\end{matrix}
\right)\notag
\end{eqnarray}
Iteration (\ref{eq:cubatic:quatNewton}) can be carried out for a number of randomly-selected unit quaternions as initial guesses. Once the global maximum is attained, Eq~(\ref{eq:scalar_OP}) can be used to calculate the scalar cubatic order parameter. Neither the accuracy nor the convergence rate depends on the selection of $v_p^\circ$ as they are only benign parameters of $\textbf{f}(\textbf{q})$. 

\begin{figure}
	\begin{center}
		\includegraphics[width=.5\textwidth]{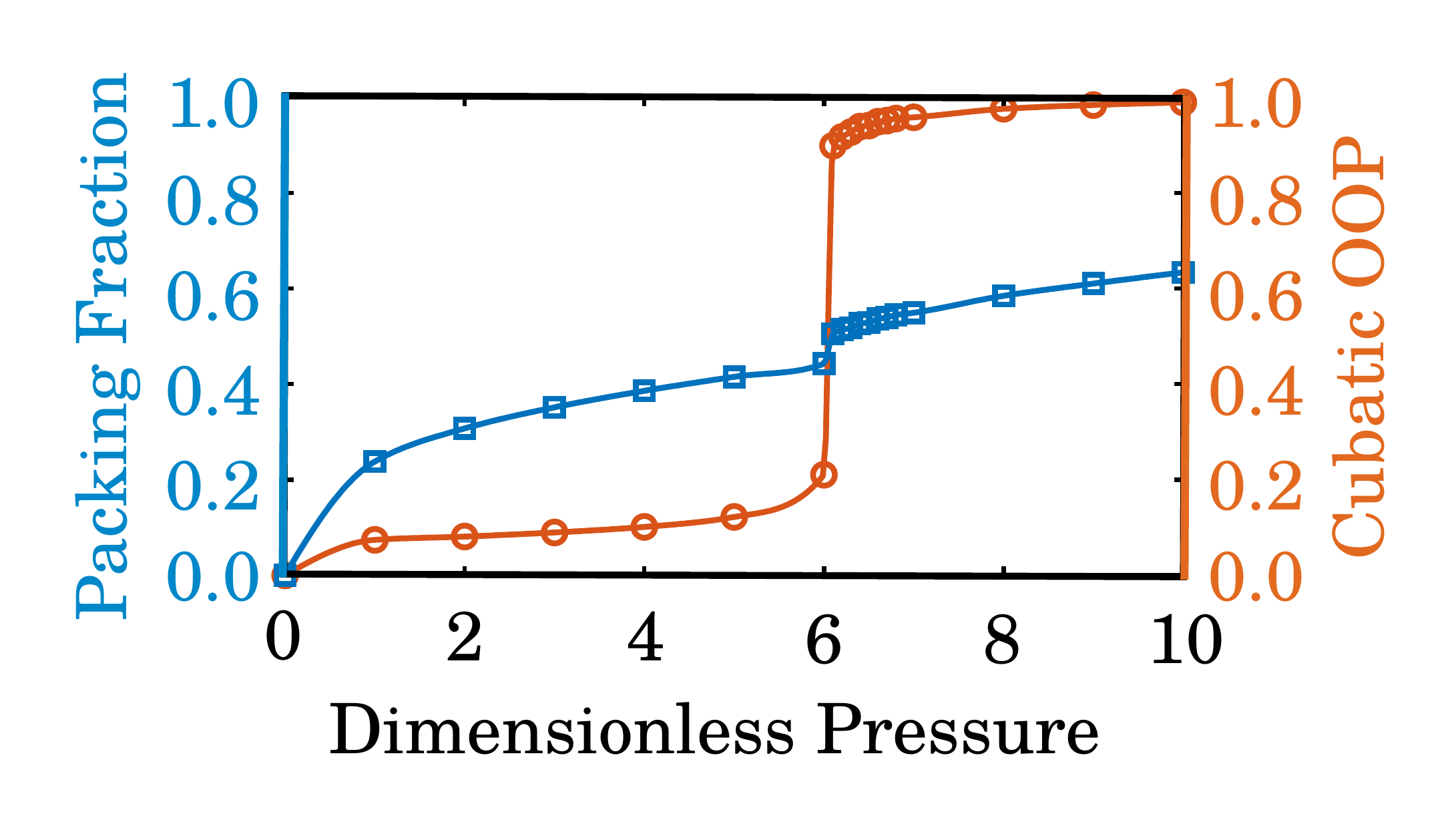}
		\caption{\label{fig:cubatic_OOP} Cubatic OOP (circles) and packing fraction (esquires) vs.~dimensionless pressure for a system of $512$ hard cubes.}
	\end{center}
\end{figure}

As a numerical example, we conduct isothermal isobaric Monte Carlo simulations of a system of $512$ hard cubes in the range of pressures at which the cubatic phase is known to exist. Technical specifications of the MC simulations are identical to what we  discussed earlier for the hard ellipsoid system. The dimensionless pressure is computed using the edge length of a cube as the length scale. In order to increase the numerical efficiency of the Newton-Raphson iterations, backtracking is also utilized. Fig.~\ref{fig:cubatic_OOP} depicts packing fraction and cubatic OOP as a function of dimensionless pressure.   The isotropic-to-cubatic transition occurs at $P^*=PV/kT\approx6.1$ and is marked by pronounced jumps in both the packing fraction and the cubatic OOP.


\subsubsection{Tetratic  Phase}
Similar to the cubatic phase, hard squares that are the two-dimensional equivalents of hard cubes can form a rotationally ordered phase with $C_4$ symmetry, known as the tetratic phase,  at sufficiently high pressures~\cite{FiakowskiPhysicaA2000, WojciechowskiCompMetSciTechnol2004}.  As explained in Section~\ref{section:SOC:CnDn}, $\mathcal{H}_4(\{\pm x,\pm y\})$ is the SOC for the $C_4$ symmetry. The tetratic phase can  be described by $\Omega=\{\pm v_p\}_{p=1}^2$, with $v_1$ and $v_2$ being two orthonormal vectors, and $M$, $\mathfrak{M}_{\Omega}$ and $\overline{M}$ can be defined as:
\begin{eqnarray}
M^{ijkl} &=& 2\sum_{p=1}^2u_p^iu_p^ju_p^ku_p^l-\frac12\left(\delta^{ij}\delta^{kl}+\delta^{ik}\delta^{jl}+\delta^{il}\delta^{jk}\right)\notag\\
\mathfrak{M}_{\Omega} &=& 2\sum_{p=1}^2v_p^iv_p^jv_p^kv_p^l-\frac12\left(\delta^{ij}\delta^{kl}+\delta^{ik}\delta^{jl}+\delta^{il}\delta^{jk}\right)\notag\\
\overline{M}^{ijkl} &=& \frac2N\sum_{p=1}^2\sum_{q=1}^Nu_{p,q}^iu_{p,q}^ju_{p,q}^ku_{p,q}^l-\frac12\left(\delta^{ij}\delta^{kl}+\delta^{ik}\delta^{jl}+\delta^{il}\delta^{jk}\right)
\end{eqnarray}
Here, $u_{p,q}$'s are the corresponding orthogonal vectors describing the orientation of particle $q$. In order to identify the $\mathfrak{M}$ that best matches $\overline{M}$, one can solve the following two-dimensional equivalentt of the optimization problem given by Eq.~(\ref{eq:cubatic_opt}):
\begin{eqnarray}
	\begin{array}{lll}
		\max & \overline{M}^{ijkl}\sum_{p=1}^2v_p^iv_p^jv_p^kv_p^l\\
		\text{subject to} & v_p^iv_q^j = \delta_{pq} & p,q=1,2
	\end{array}\label{eq:C4_opt}
\end{eqnarray}
with the solution given by $\zeta_{12}=\mu_{12}-\mu_{21}=0$. Numerical estimation of the Karush-Kuhn-Tucker solutions of~(\ref{eq:C4_opt}) can be performed using a similar Newton-Raphson scheme, based on $\theta$, the polar angle by noting that:
\begin{eqnarray}
\frac{d\zeta_{12}}{d\theta} &=& \sum_{p=1}^2 \left(\nabla_p\zeta_{12}\right)^T\frac{dR}{d\theta}v_p^{\circ}
\end{eqnarray}
with:
\begin{eqnarray}
R(\theta) &=& \left(
	\begin{matrix}
		\cos\theta & \sin\theta \\
		-\sin\theta & \cos\theta 
	\end{matrix}
\right)
\end{eqnarray}

\begin{figure}
	\begin{center}
		\includegraphics[width=.5\textwidth]{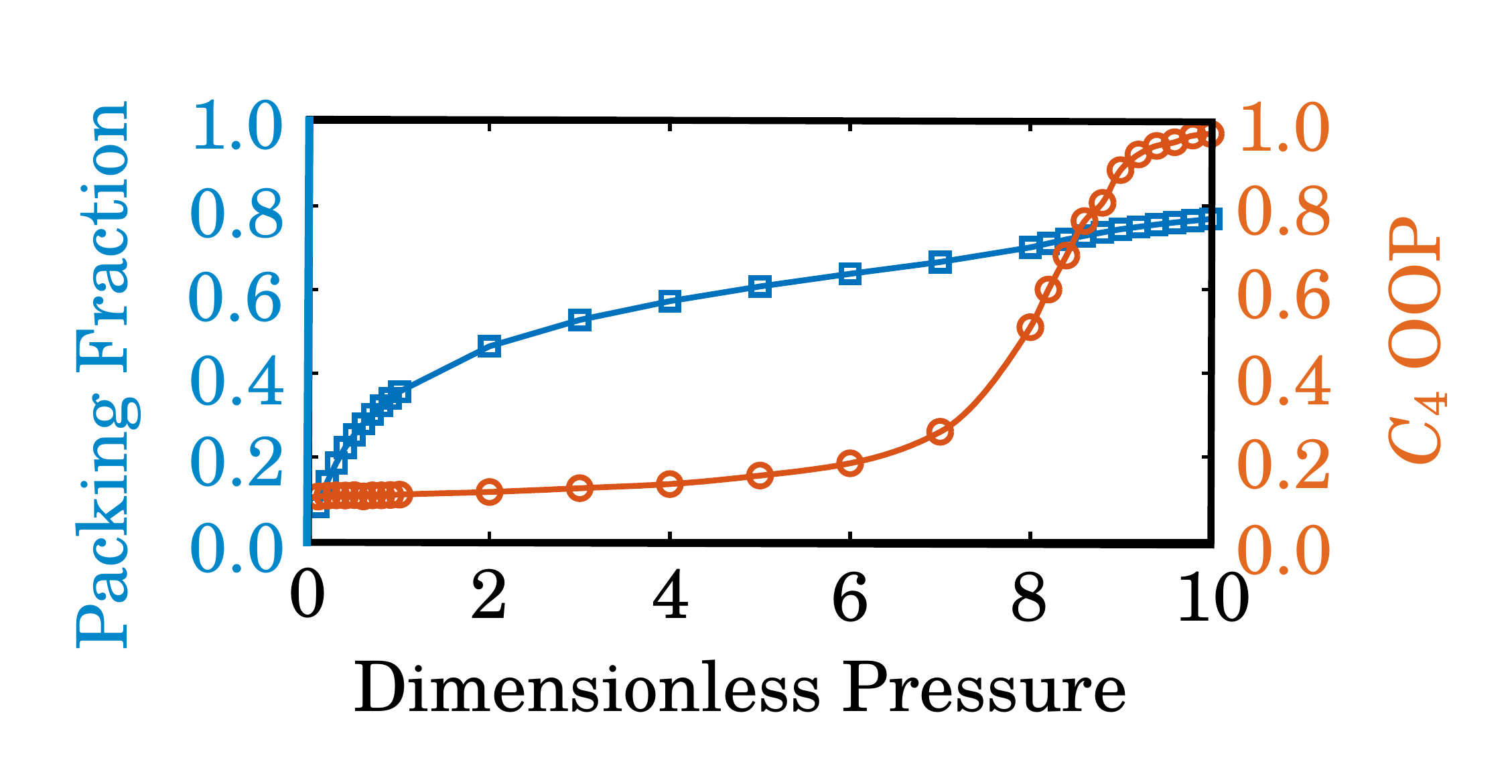}
		\caption{\label{fig:C4_OOP} $C_4$ OOP (cirlces) and packing fraction (squares) vs.~pressure for a system of $256$ hard squares.}
	\end{center}
\end{figure}

As a numerical example, we conduct isothermal isobaric Monte Carlo simulations of a system of $256$ hard squares in the range of pressures at which the tetratic phase emerges The simulation details, including what constitutes an MC step, is identical to what was discussed in the hard ellipsoid and the hard cube systems. Dimensionless pressure is computed from the edge length of a square as the length scale. Simulations are performed for $10^8$ MC cycle, and snapshots are gathered every $5\times10^4$ MC cycles.  Fig.~\ref{fig:C4_OOP} depicts the packing fraction and the tetratic OOP as a function of pressure. The isotropic-to-transition occurs gradually for $P^*=PA/kT\approx8.0$ and is not as pronounced as the transitions observed in three-dimensional systems. 


\subsubsection{Tetrahedral Nematics}
Under external fields, regular tetrahedra are capable of forming a structure in which all particles take the same orientation. The possibility and the thermodynamic nature of such a transition has been thoroughly discussed in theoretical works of liquid crystals~\cite{Fel1995PhysRevE}. However, such a transition has yet to be realized experimentally. The arising orientationally ordered structure is typically referred to as the \emph{tetrahedral nematics} phase, a structure that geometrically falls into the general category of structures outlined in Example~\ref{example:complete_match_structure}. According to Section~\ref{section:SOC:tetrahedral}, $\mathcal{H}_3$ is the SOC for the tetrahedral rotation group. For the tetrahedral nematic phase, $\Omega=\{u_p\}_{p=1}^4$ is the characteristic orbit of the reference tetrahedron with which the individual tetrahedra align. $M$, $\mathfrak{M}_{\Omega}$ and $\overline{M}$ are thus defined as:
\begin{eqnarray}
M^{ijk} &=& \sum_{p=1}^4v_p^iv_p^jv_p^k\notag\\
\mathfrak{M}^{ijk}_{\Omega} &=& \sum_{p=1}^4u_p^iu_p^ju_p^k\notag\\
\overline{M}^{ijk} &=& \frac1N\sum_{l=1}^N\sum_{p=1}^3v_{p,l}^iv_{p,l}^jv_{p,l}^k
\end{eqnarray}
Here, $v_{p,q}$'s are the elements of the high-symmetry orbit for particle $q$. The following optimization problem is to be solved for quantifying the extent of tetrahedral order:
\begin{eqnarray}
	\begin{array}{lll}
		\max & \overline{M}^{ijk}\sum_{p=1}^4u_p^iu_p^ju_p^k\\
		\text{subject to} & v_p^iv_q^i = \frac43\delta_{pq}-\frac13 & p,q=1,2,3,4
	\end{array}\label{eq:tetra:optimization}
\end{eqnarray}
As for the cubatic phase, we can replace the constraints of (\ref{eq:tetra:optimization}) with $\sum_{p=1}^4u_p^iu_p^j=(4/3)\delta^{ij}$ and $\sum_{p=1}^4u_p^i=0$. (See Appendix~\ref{appendix:tetrahedral_nematic:constraint_equivalency} for details.) The Lagrangian and its derivatives are thus given by:
\begin{eqnarray}
	\mathcal{L} &=& \overline{M}^{ijk}\sum_{p=1}^4u_p^iu_p^ju_p^k-\sum_{p=1}^4\left[\lambda_iu_p^i-\mu_{ij}u_p^iu_p^j\right]\notag\\
	\nabla_q\mathcal{L} &=& b_q-c-Cu_q\notag
\end{eqnarray}
with:
\begin{eqnarray}
b_q &=& \left(
	\begin{matrix}
		3\overline{M}^{1jk}u_q^iu_q^j\\
		3\overline{M}^{2jk}u_q^iu_q^j\\
		3\overline{M}^{3jk}u_q^iu_q^j
	\end{matrix}
\right), c=\left(
	\begin{matrix}\lambda_1\\ \lambda_2\\ \lambda_3\end{matrix}
\right), C = \left(
	\begin{matrix}
		\mu_{11} & \mu_{12} & \mu_{13}\\
		\mu_{21} & \mu_{22} & \mu_{23}\\
		\mu_{31} & \mu_{32} & \mu_{33}
	\end{matrix}
\right)\notag
\end{eqnarray}
with $C=C^T$. 
Multiplying $\nabla_q\mathcal{L}$  by $u_q^T$ on the right and summing over $q$ yields:
\begin{eqnarray}
	 C &=& \frac34\sum_{q=1}^4 b_qu_q^T
\end{eqnarray}
The symmetry of $C$ can be enforced using a method similar to what was explained for the cubatic phase and the scalar order parameter can be calculated accordingly. 


\subsubsection{Tetrahedral Axial Nematics}
We call an arrangement of regular tetrahedra a \emph{tetrahedral axial nematic} if a specific axis of each particle (modulo symmetry operations) aligns with a common director on average. The formation of tetrahedral axial nematics has neither been observed in experiments, nor has it been suggested in theoretical studies. The discussion that follows is therefore only a geometrical illustration of the ideas presented in this paper. 

\begin{table*}
\caption{\label{table:uniaxial_OOP_info} Projection parameters and $\omega$ for face, edge and \textit{z} uniaxial nematics. }
\begin{tabular}{lccccccccccccccccccccc}
\hline\hline
Phase && $\alpha_1$ & $\beta_1$ && $\alpha_2$ & $\beta_2$ && $\alpha_3$ & $\beta_3$ && $\alpha_4$ & $\beta_4$ && $\sum_{i=1}^4\alpha_i^4$ && $\sum_{i=1}^4\alpha_i^2\beta_i^2$ && $\sum_{i=1}^4\beta_i^4$ && $\omega$ \\
\hline
Face Nematics && $1$ & $0$ && $-\frac13$ & $\frac{\sqrt{8}}3$ && $-\frac13$ & $\frac{\sqrt{8}}3$ && $-\frac13$ & $\frac{\sqrt{8}}3$ && $\frac{28}{27}$ && $\frac{8}{27}$ && $\frac{64}{27}$ && $\frac{28}{27}$ \\
Edge Nematics && $\frac{\sqrt{6}}3$ & $\frac{\sqrt{3}}3$ && $-\frac{\sqrt{6}}3$ & $\frac{\sqrt{3}}3$ && $0$ & $1$ && $0$ & $1$ && $\frac89$ && $\frac49$ && $\frac{20}{9}$ && $\frac7{18}$\\
\textit{z} nematics && $\frac{\sqrt{3}}3$ & $\frac{\sqrt{6}}3$ && $\frac{\sqrt{3}}3$ & $\frac{\sqrt{6}}3$ && $-\frac{\sqrt{3}}3$ & $\frac{\sqrt{6}}3$ && $-\frac{\sqrt{3}}3$ & $\frac{\sqrt{6}}3$ && $\frac49$ && $\frac89$ && $\frac{16}9$ && $-\frac{14}9$ \\
\hline
\end{tabular}
\end{table*}

Like the uniaxial nematic phase described in Section~\ref{section:nematic}, $\Omega=\{\pm z\}$ for a tetrahedral axial nematic. Here, we consider three distinct plausible types of alignments with $z$. In the face nematic phase, each tetrahedron has one of the vectors that is normal to its faces aligned with $z$. In the edge nematics, however, one edge of each tetrahedron is aligned with $z$. The third structure is what we call the $z$ nematics, and a vector that connects the centers of two non-adjacent edges aligns with $z$. Note that all these phases have inversion symmetry. Therefore,  $M$ cannot be an odd-ranked moment of $\mathcal{H}_3$. Instead we define $M$ as:
\begin{eqnarray}
M^{ijkl}(\mathcal{H}_3) &=& \sum_{p=1}^4v_p^iv_p^jv_p^kv_p^l - \frac4{15}\left(\delta^{ij}\delta^{kl}+\delta^{ik}\delta^{jl}+\delta^{il}\delta^{jk}\right)\notag\\
\overline{M}^{ijkl} &=& \frac1N\sum_{q=1}^N\sum_{p=1}^4v_{p,q}^iv_{p,q}^jv_{p,q}^kv_{p,q}^l - \frac4{15}\left(\delta^{ij}\delta^{kl}+\delta^{ik}\delta^{jl}+\delta^{il}\delta^{jk}\right)\notag
\end{eqnarray}
Note that $\mathcal{H}_4^{ijkl}=(3/4)\mathcal{H}_3^{ijm}\mathcal{H}_3^{klm}+(4/9)\delta^{ij}\delta^{kl}$, so $M$ is indeed a function of $\mathcal{H}_3$. For each of the three phases introduced above, $\mathfrak{M}_{\Omega}$ can be calculated by expressing the elements of the high-symmetry orbit $\{v_p\}_{p=1}^4$ as $v_p=\alpha_pz+\beta_pt_p$ with $t_p\perp z$, and using Proposition~\ref{prop:axial} to obtain:
\begin{eqnarray}
\mathfrak{M}_{\Omega}^{ijkl} &=& \sum_{p=1}^4\alpha_p^4z^iz^jz^kz^l + \frac18\sum_{p=1}^4\beta_p^4P^{ijkl}+ \frac12\sum_{p=1}^4\alpha_p^2\beta_p^2Q^{ijkl}-\frac4{15}(\delta^{ij}\delta^{kl}+\delta^{ik}\delta^{jl}+\delta^{il}\delta^{jk})
\end{eqnarray}
with $P$ and $Q$ given by:
\begin{eqnarray}
P^{ijkl} &=& 3(x^ix^jx^kx^l+y^iy^jy^ky^l) + x^ix^jy^ky^l + y^iy^jx^kx^l +x^iy^jy^kx^l+y^ix^jx^ky^l+x^iy^jx^ky^l+y^ix^jy^kx^l \\
Q^{ijkl} &=& x^ix^jz^kz^l+y^iy^jz^kz^l+z^iz^jx^kx^l+z^iz^jy^ky^l +z^ix^jx^kz^l+z^iy^jy^kz^l\notag\\&&+x^iz^jz^kx^l+y^iz^jz^ky^l+z^ix^jz^kx^l+z^iy^jz^ky^l+x^iz^jx^kz^l+y^iz^jy^kz^l 
\end{eqnarray}
Here, $x$ and $y$ are mutually orthogonal unit vectors that are also perpendicular to $z$. It can be easily observed that $P$ and $Q$ are invariant under orthogonal transformations that keep $z$ unchanged. Therefore, $\mathfrak{M}_{\Omega}$ clearly satisfies the properties outlined in Section \ref{subsection:OrderParameter}. The extent of order can therefore be quantified by solving the following associated optimization problem:
\begin{eqnarray}\label{eq:uniaxial_generic_optimization}
	\begin{array}{ll} 
	\max & \omega\overline{M}^{ijkl}z^iz^jz^kz^l\\ 
	\text{subject to} & z^iz^i = 1	
	\end{array}
\end{eqnarray}
where $\omega=\sum_{p=1}^4[\alpha_p^4+\frac38\beta_p^4-3\alpha_p^2\beta_p^2]$. The derivation details can be found in Appendix \ref{appendix:nematic_OOP_derive}. Table \ref{table:uniaxial_OOP_info} gives the corresponding $\alpha_p$'s, $\beta_p$'s and $\omega$'s for the face, edge and \textit{z} nematics. For the $z$ nematic phase, $\omega<0$, and solving (\ref{eq:uniaxial_generic_optimization}) reduces to \emph{minimizing} $\overline{M}^{ijkl}z^iz^jz^kz^l$ under the same constraints. The scalar order parameter can also be defined using (\ref{eq:scalar_OP}).


\begin{figure}
\begin{center}
	\includegraphics[width=.4\textwidth]{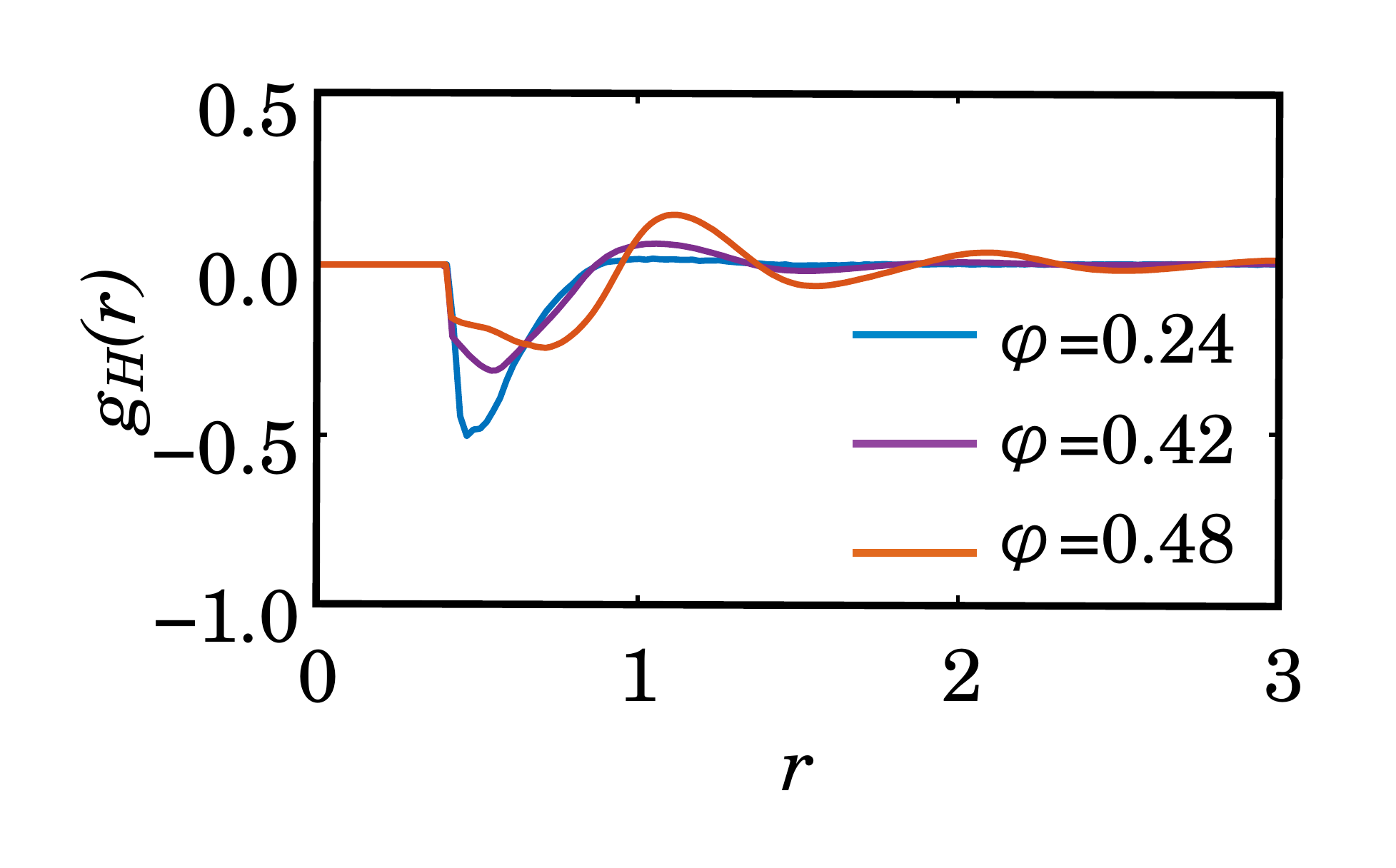}
	\caption{\label{fig:gHr}$g_H(r)$ calculated for a fluid of $4,\!096$ hard tetrahedra at different packing fractions.}
	\end{center}
\end{figure}

\begin{figure}
\begin{center}
	\includegraphics[width=.4\textwidth]{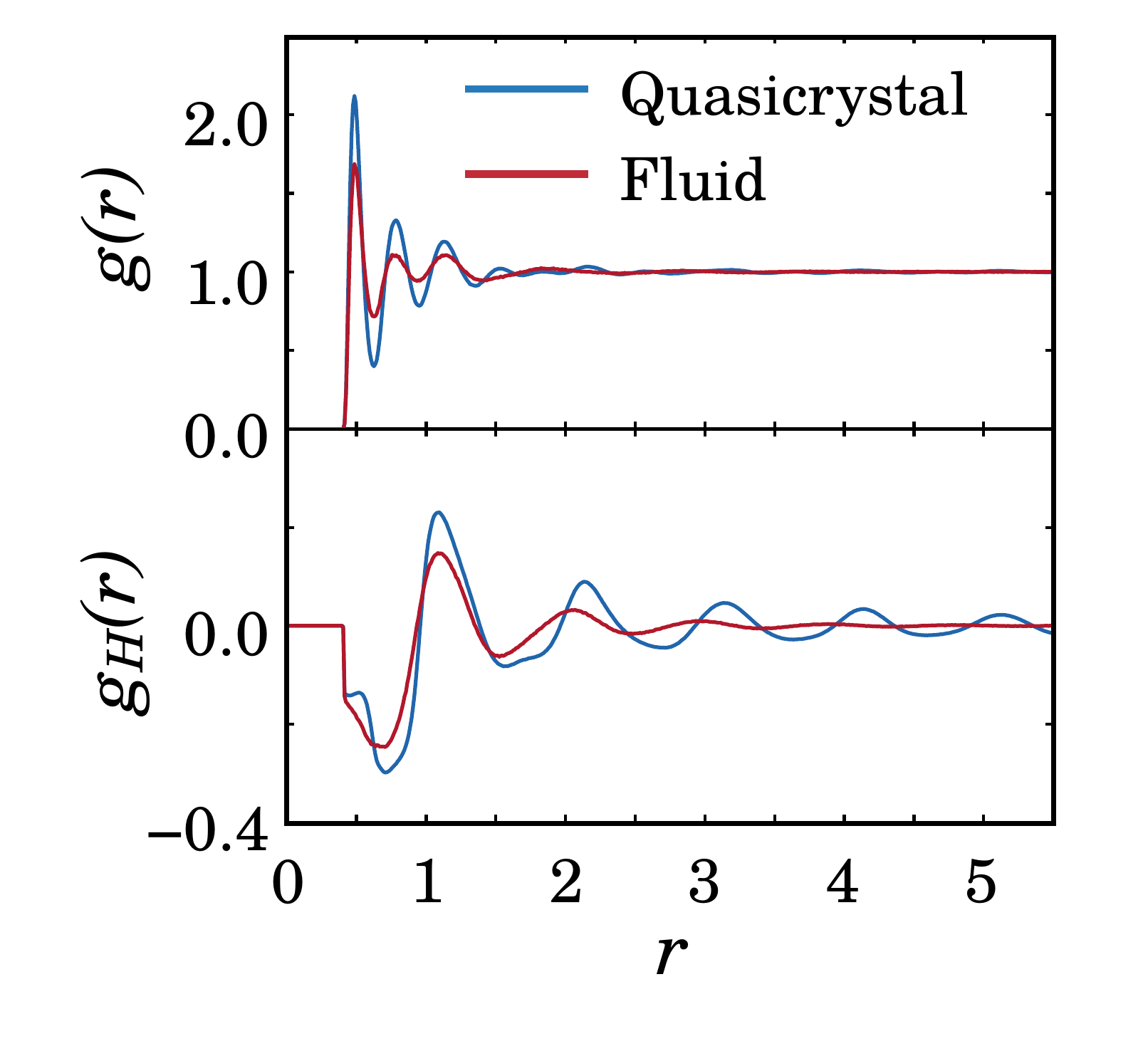}
	\caption{\label{fig:QC-vs-fluid} $g(r)$ and $g_H(r)$ calculated for a quasicrystal and a fluid of $8,\!000$ hard regular tetrahedra at $P\sigma^3/k_BT=64$, with $\sigma$  the edge length of a tetrahedron.}
	\end{center}
\end{figure}

\begin{table}
	\caption{\label{table:gHr}Expected values of $\mathcal{H}_{3,u}^{ijk}\mathcal{H}_{3,v}^{ijk}$ and $g_H(r)$ for different nearest neighbors shells in the hard tetrahedron system}
		\begin{tabular}{ccc}
			\hline\hline
				$n$th nearest neighbor~~~ & ~~~$\mathcal{H}_{3,u}^{ijk}\mathcal{H}_{3,v}^{ijk}$~~~ &~~~ $g_H(r)$~~~ \\
				\hline
				$1$ & $-\frac{32}{81}$ & $-\frac19$\\
				$2$ & $-\frac{2464}{2187}$ & $-\frac{77}{243}$ \\
				$3$ & $\frac{79072}{59049}$ & $\frac{2471}{6561}$ \\
				\hline\hline
		\end{tabular}
\end{table}

\section{Other applications of Strong Rotational Coordinates
\label{section:other_applications}
}
So far, we have only used SOCs to quantify global orientational order in systems of particles with non-trivial rotational symmetries. In this section, other potential applications of SOCs in computational studies of soft condensed matter are discussed.

\subsection{Time-averaged orientations
\label{subsection:tAvgOr}
}
The machinery of strong orientational coordinates can be used to calculate the average orientation of a particle in the course of a simulation. In general, this is done to get rid of thermal fluctuations, and can be considered the equivalent of determining the inherent structure in an energetic system~\cite{StillingerWeberJCP1985}. More precisely, if $M_p(t)$ is the SOC of particle $p$ at time $t$, one can maximize $\mathfrak{M}_p\odot\frac1T\int_0^TM_p(t)dt$ to solve for the best average orientation that matches the trajectory of $p$. This problem also fits into the general class of problems outlined in Example~\ref{example:complete_match_structure}.


\subsection{Spatial Correlation of Local Rotational Order
\label{subsection:spatialCorrelation}
}
There are different  ways of quantifying spatial orientational correlations, with different levels of generality~\cite{GregPNAS2014}.  The SOCs derived in this work offer a systematic way of quantifying orientational correlations in theoretical and computational studies of building blocks with nontrivial symmetries. We explain this through an example, i.e.,~the different phases formed by hard tetrahedra. At low packing fractions, hard tetrahedra form a simple fluid in which the orientations of different particles are not correlated. This, however, changes as the packing fraction increases, and a networked structure emerges in which all neighboring particles are in  face-to-face configurations. At sufficiently high densities, this network transforms into a dodecagonal quasicrystal~\cite{HajiAkbariEtAl2009}. In both the complex fluid and the quasicrystal, the neighboring tetrahedra are in face-to-face contact, a fact quantitatively characterized in earlier studies~\cite{GregPNAS2014}. The SOCs derived in this work can, however, allow us to get further information about the transformation of the simple fluid into the complex fluid vis-a-vis the angular distribution of the neighboring face-to-face tetrahedra. It can also be used for detecting and quantifying long-range orientational coherence in these networks.

It was shown in Section~\ref{section:SOC:tetrahedral} that $\mathcal{H}_3(\mathcal{O}_t)$ is an SOC for a regular tetrahedron. Now consider an arrangement of $N$ tetrahedra and let $N(r)$ be the number of pairs whose center-to-center distance lies in the interval $[r-dr/2,r+dr/2]$ and define the following correlation function:
\begin{eqnarray}
	g_H(r) &=& \left\langle\frac{1}{N(r)\norm{\mathcal{H}_3}_F^2}\sum_{<p,q>,|d(p,q)-r|<dr/2}\mathcal{H}_{3,p}^{ijk}\mathcal{H}_{3,q}^{ijk}\right\rangle\notag\\
	&&\label{eq:tetrahedral_correlation_function}
\end{eqnarray}
We can quantify  different types of spatial correlations in the orientations of tetrahedra using (\ref{eq:tetrahedral_correlation_function}). We first analyze different scenarios for the first nearest neighbor shell. Geometrically, the smallest distance between the centroids of two neighboring tetrahedra is achieved when their faces are touching one another. At low densities, the two neighboring tetrahedra can freely rotate around the axis perpendicular to this touching face, while at higher densities, rotation is restricted and the two faces tend to match perfectly. The value of $g_H(r)$ can be calculated for these two idealized configurations. Let $z$ be a unit vector perpendicular to the common plane of the two touching faces. The characteristic vectors of each tetrahedron can thus be expressed as:
\begin{eqnarray}
u_1 &=& z \notag\\
u_p &=& -\tfrac13z+\tfrac{\sqrt8}3s_p\notag\\
v_1 &=& -z \notag\\
v_p &=& \tfrac13z+\tfrac{\sqrt8}3t_p\notag
\end{eqnarray}
with:
\begin{eqnarray}
s_p &=& \cos[{2\pi(p-2)}/{3}]x+\sin[{2\pi(p-2)}/{3}]y\notag\\
t_p &=& \cos\left[{2\pi(p-2)}/{3}+\theta\right]x+\sin\left[{2\pi(p-2)}/{3}+\theta\right]y\notag
\end{eqnarray}
and $x$ and $y$ two orthonormal vectors perpendicular to $z$. For a perfect non-rotating face-to-face contact, $\theta=0$ and $\mathcal{H}_{3,u}^{ijk}\mathcal{H}_{3,v}^{ijk}=-\frac{32}{81}$. The corresponding value of $g_H(r)$ will thus be $-\frac19\approx-0.11111$. For a freely-rotating face-to-face contact, however, $\theta\sim U(0,2\pi)$, and $\langle\mathcal{H}_{3,u}^{ijk}\mathcal{H}_{3,v}^{ijk}\rangle=-\frac{160}{81}$ which yields a $g_H(r)$ value of $-\frac59\approx0.5555$. Investigating the magnitude of the first valley of $g_H(r)$ will elucidate the type of the face-to-face contact prevalent in that particular phase. Fig.~\ref{fig:gHr} shows $g_H(r)$ vs. $r$ for three different packing fractions. At $\phi=0.24$, the first valley of $g_H(r)$ has a value of $-0.503$ which is very close to the calculated value of $-\frac59$ for a perfect freely rotating face-to-face configuration. At a packing fraction of $0.48$ however, the first valley of $g_H(r)$ has a value of $-0.15$ which is  very close to the theoretically predicted value of $-\frac19$ for the perfect face-to-face configuration. At intermediate densities (the violet curve), the first valley is in between $-\frac59$ and $-\frac19$. This analysis can be extended to the second and third nearest neighbor shells. For perfect face-to-face configurations, the corresponding values of $\mathcal{H}_{3,u}^{ijk}\mathcal{H}_{3,v}^{ijk}$ and $g_H(r)$ are given in Table~\ref{table:gHr}. The $g_H(r)$ values shown in Fig.~\ref{fig:gHr} are consistent with these theoretical predictions, which confirms the existence of the network structure in which all neighbors are in  perfect face-to-face contact. 

Differences in spatial orientational correlation can be the basis of differentiating the quasicrystal and the complex fluid, both of which have a similar network structure with perfect face-to-face contacts between nearest neighbors~\cite{HajiAkbariEtAl2009, HajiAkbaricondmat2011}. They also have very similar radial distribution functions, $g(r)$'s, as depicted in Fig.~\ref{fig:QC-vs-fluid}. Note the lack of long-range order in both $g(r)$'s. To the contrary, these  phases have starkly different $g_H(r)$ functions.  For the disordered fluid, orientational correlations are virtually nonexistent at large separations. For the quasicrystal, however, such correlations never decay and are long-ranged. By using $g_H(r)$, we are therefore able to detect long-range orientational coherence in a networked phase of hard tetrahedra. 
It is noteworthy that the $g_H(r)$ values corresponding to the first, second and third nearest neighbor shells in the quasicrystal are closer to the theoretically predicted values given in Table~\ref{table:gHr} than the corresponding peaks in the fluid. Therefore, the face-to-face contacts in the quasicrystal are more perfect than the fluid. 

\section{Conclusion\label{section:conclusion}}
In this work, we revisit the problem of quantifying orientational order in arrangements of anisotropic building blocks, and we propose a systematic way of constructing symmetry-invariant coordinates for arbitrary building blocks. We call the arising tensorial coordinates strong orientational coordinates, and we discuss their potential applications in theoretical and computational studies of symmetric building blocks. For instance, we demonstrate that such SOCs can be used in the systematic quantification of spatiotemporal correlations in colloidal systems. Most importantly, the orientational distribution functions can be expressed in terms of such SOCs, and the problem of identifying, quantifying and describing long-range rotational order can be formulated as a generalized non-linear optimization problem. The arising scalar and tensorial order parameters can be efficient ways of storing rotational information on a computer, especially for objects with nontrivial symmetries. 

\acknowledgments
The authors gratefully acknowledge discussions with R. Greiss, S. Fomin, J. C. Lagarias, M. Engel, G. van Anders, P. Palffy-Muhoray and R. Petschek. Resources and support for S.G. and A.H.-A. provided in part by the U.S. Air Force Office of Scientific Research under Multidisciplinary University Research Initiative No. FA9550-06-1-0337, Subaward No. 444286-P061716. This material is based upon work also supported by the DOD/DDRE under Award No. N00244-09-1-0062 (S.G.). Any opinions, findings, and conclusions or recommendations expressed in this publication are those of the authors and do not necessarily reflect the views of the DOD/DDRE.  A.H.-A. acknowledges support from the University of Michigan Rackham Predoctoral Fellowship program.  

\appendix


\section{Proof of Lemma~\ref{lemma:sum_of_powers_roots}\label{proof:lemma:sum_of_powers_roots}}

Form the polynomials $p(z)=\prod_{i=1}^n(z-a_i)$ and $q(z)=\prod_{i=1}^n(z-b_i)$ and expand them to obtain $p(z)=z^n+\sum_{i=1}^n(-1)^i\alpha_iz^{n-i}$ and $q(z)=z^n+\sum_{i=1}^n(-1)^i\beta_iz^{n-i}$ where the coefficients are given by Newton's identities~\cite{KalmanMathMag2000}:
\begin{eqnarray}
\phi_k&=&\sum_{i=1}^{k-1}(-1)^{i-1}\alpha_i\phi_{k-i}+(-1)^{k-1}k\alpha_k\notag\\
\psi_k&=&\sum_{i=1}^{k-1}(-1)^{i-1}\beta_i\psi_{k-i}+(-1)^{k-1}k\beta_k\notag
\end{eqnarray}
where $\phi_k=\sum_{i=1}^na_i^k$ and $\psi_k=\sum_{i=1}^nb_i^k$. Note that $\phi_1=\psi_1$ if and only if $\alpha_1=\beta_1$ and by induction $\alpha_k=\beta_k$ for every $1\le k\le n$. Thus $p(z)$ and $q(z)$ have the same coefficients, and thus the same roots.

\section{Proof of $\mathcal{H}_n$ being an SOC for the planar orbit $\mathcal{O}_n$ in Section~\ref{section:SOC:CnDn}\label{appendix:derivation:CnDn}}
We prove this through the following chain of lemmas and theorems.

\begin{lem}
Let $v_1,v_2,\cdots,v_n\in\mathbb{R}^d$ be unit vectors so that $v_p^Tv_q=\cos\left[{2\pi}(p-q)/n\right]$; then for $n\ge3$, $v_i$'s span a two-dimensional subspace of $\mathbb{R}^d$.
\end{lem}
\begin{proof}
Let $x,y,z\in\mathbb{R}^d$ be unit vectors so that  $v_1=x, v_2=\cos({2\pi}/{n})x+\sin({2\pi}/{n})y$ and $v_k=\cos[{2\pi(k-1)}/{n}]x+\sin[{2\pi(k-1)}/{n}]z$ for some $k\ge3$. Use $v_k^Tv_2=\cos[{2\pi(k-2)}/{n}]$ to deduce $z=y$.
\end{proof}

\begin{lem}
\label{lem:cos_p_sin_q}
Let $p,q$ be nonnegative integers and $\phi_k=\theta+2\pi k/n$. Then $I_{p,q}(\theta)=\sum_{k=0}^{n-1}\cos^p\phi_k\sin^q\phi_k$ is independent of $\theta$ for $p+q<n$, and is a function of $\theta$ if $p+q=n$.
\end{lem}
\begin{proof}
Note that:
\begin{eqnarray}
\cos^p\phi_k &=& \left(\frac12\right)^p\sum_{l=0}^p\left(\begin{matrix}p \\ l\end{matrix}\right)\exp\left[i\phi_k(2l-p)\right]\label{eq:cos_p}\\
\sin^q\phi_k &=& \left(\frac{i}{2}\right)^q\sum_{m=0}^q\left(\begin{matrix}q\\ m\end{matrix}\right)(-1)^m\exp\left[i\phi_k(2m-q)\right] \label{eq:sin_q}
\end{eqnarray}
$I_{p,q}(\theta)$ can therefore be written as:
\begin{eqnarray}
I_{p,q}(\theta) &=& \frac{i^q}{2^{p+q}}\sum_{l=0}^p\sum_{m=0}^q\left(\begin{matrix}p\\l\end{matrix}\right)\left(\begin{matrix}q\\m\end{matrix}\right)(-1)^me^{i\theta(2l+2m-p-q)}\sum_{k=0}^{n-1}\left[e^{2\pi i(2l+2m-p-q)/n}\right]^k\notag
\end{eqnarray}
For $e^{2\pi i(2l+2m-p-q)/n}\neq1$ we have:
\begin{eqnarray}
\sum_{k=0}^{n-1}\left[e^{2\pi i(2l+2m-p-q)/n}\right]^k &=& \frac{1-e^{2\pi i(2l+2m-p-q)}}{1-e^{2\pi i(2l+2m-p-q)/n}}=0\notag
\end{eqnarray}
Thus the sum over $k$ survives only if $2l+2m-p-q$ is a multiple of $n$. However $|2l+2m-p-q|\le p+q$. Thus if $p+q<n$, the only possibility is zero, which will take away the $\theta$ dependence of $I_{p,q}(\theta)$. For $p+q=n$, two $\theta$-dependent terms will survive i.e. $e^{\pm in\theta}$ and the proof follows.
\end{proof}

\begin{thm}
$\mathcal{H}_n$ is an SOC of $\mathcal{O}_n$.
\end{thm}
\begin{proof}
Let $V=\mathcal{O}_n=\{v_k\}_{k=1}^n$. Without loss of generality suppose $v_k=\cos\left[\theta+{2\pi k}/{n}\right]e_x+\sin\left[\theta+{2\pi k}/{n}\right]e_y$ with $e_x$ and $e_y$ being the unit vectors along the $x$ and $y$ directions. The components of $\mathcal{H}_m$ are therefore either zero or $I_{p,m-p}$ which are always independent of $\theta$ for $m<n$ according to Lemma~\ref{lem:cos_p_sin_q}. For even $n$, the proof follows from Corollary~\ref{cor:even_sets}. The proof for odd $n$ is completed by noting that $\mathcal{H}_n(QV)$ will have some nonzero components that are zero for $\mathcal{H}_n(V)$ if $QV$ and $V$ are not in the same plane.
\end{proof}

\section{Derivation Details of Tetrahedral SOC\label{appendix:tetopt}}
\noindent We need to solve the following optimization problem:
\begin{eqnarray}\label{eq:appendix:optimization}
\begin{array}{lll} \text{maximize} & \sum_{l,m=1}^4\xi_{lm}^3 &\\ \text{subject to} & \sum_{m=1}^4\xi_{lm}^2=\frac43 & l=1,2,3,4\\ &\sum_{m=1}^4\xi_{lm}=0 & l=1,2,3,4\end{array}
\end{eqnarray}
The form of (\ref{eq:appendix:optimization}) suggests that we can solve it by breaking into four independent problems of the form:
\begin{eqnarray}\label{eq:appendix:subopt}
\begin{array}{ll} \text{maximize} & \sum_{i=1}^4x_i^3\\ \text{subject to} & \sum_{i=1}^4x_i=0\\ & \sum_{i=1}^4x_i^2=\frac43\end{array}
\end{eqnarray}
We solve this problem by enumerating all solutions satisfying the Karush-Kuhn-Tucker criteria~\cite{Karush1939,KuhnTucker1951} and show that global maximum is attained if $x$ is a permutation of $(1,-\frac13,-\frac13,-\frac13)$. The Lagrangian for (\ref{eq:appendix:subopt}) is:
\begin{eqnarray}
\mathcal{L}(x,\nu,\mu) &=& \sum_{i=1}^4x_i^3+\nu\sum_{i=1}^4x_i+\mu\left[\sum_{i=1}^4x_i^2-\frac43\right]\\
\frac{\partial\mathcal{L}}{\partial x_j} &=& 3x_j^2+\nu+2\mu x_j
\end{eqnarray}
KKT conditions requires that $\nabla_x\mathcal{L}=0$. Eliminating $\mu$ one easily obtains:
\begin{eqnarray}
\nu&=&-\frac34\sum_{i=1}^4x_i^2=-1
\end{eqnarray}
which implies that:
\begin{eqnarray}
x_i &=& \frac{-\mu\pm\sqrt{\mu^2+3}}3
\end{eqnarray}
Denote the roots of this equation with $\xi_+$ and $\xi_-$ and let $n_+(n_-)$ be the number of $x_i$'s equalling $\xi_+(\xi_-)$. Using $n_+\xi_++n_-\xi_-=0$ one gets:
\begin{eqnarray}
\mu &=& \frac{3m}{\sqrt{3(16-m^2)}}\\
\xi_{\pm} &=& -\frac{m\mp4}{\sqrt{3(16-m^2)}}
\end{eqnarray}
where $m=n_+-n_-$. The global maximum is thus obtained for $m=-2,\xi_+=1,\xi_-=-\frac13$, for which $\sum_{i=1}^4x_i^3=\frac89$. The global maximum of (\ref{eq:appendix:optimization}) is therefore $\frac{32}9$.


\section{Derivation Details of Octahedral SOC\label{appendix:octaopt}}
Letting $\zeta_{pq}=\xi_{pq}^2$, we need to solve the following optimization problem:
\begin{eqnarray}\label{eq:app:octa:opt}
	\begin{array}{lll}
		\max & \sum_{p,q=1}^3\zeta_{pq}^2\\
		\text{subject to} & \sum_{p=1}^3\zeta_{pq}=1& q=1,2,3\\
		& \zeta_{pq}\ge0& p,q=1,2,3
	\end{array}
\end{eqnarray}
Observe that (\ref{eq:app:octa:opt}) can be broken into three independent and yet identical optimization problems of the form:
\begin{eqnarray}
	\begin{array}{ll}
		\max & \sum_{i=1}^3x_i^2\\
		\text{subject to} & \sum_{i=1}^3x_i=1\\
		& x_i\ge0, i=1,2,3
	\end{array}
\end{eqnarray}
which is solved by identifying $x$'s that satisfy the Karush-Kuhn-Tucker criteria~\cite{Karush1939,KuhnTucker1951}. The Lagrangian is given by:
\begin{eqnarray}
\mathcal{L} &=& \sum_{i=1}^3x_i^2-\sum_{i=1}^3\nu_ix_i-\lambda\sum_{i=1}^3x_i\\
\frac{\partial\mathcal{L}}{\partial x_j} &=& 2x_j-\nu_j-\lambda=0
\end{eqnarray}
which yields $x_j=(\nu_j+\lambda)/2$. Let $n\le3$ be the number of nonzero $x_j$'s. For each such $x_j$, $\nu_j=0$ due to complementary slackness. We thus have $\lambda=\frac2n$ and $x_j=\frac1n$ for nonzero $x_j$'s. Choosing $n=1$ maximizes the objective function with $\sum_{i=1}^3x_i^2=1$. The global maximum for (\ref{eq:app:octa:opt}) is therefore $3$ with $\xi_{pq}$'s being a permutation of $(\pm1,0,0)$ for each $q$. 


\section{Derivation Details of Icosahedral SOC\label{appendix:icosa}}
We need to solve the following optimization problem:
\begin{eqnarray}\label{eq:app:icosa:opt}
	\begin{array}{lll}
		\max & \sum_{p,q=1}^6 \xi_{p,q}^3 & \\
		\text{subject to} & \sum_{p=1}^6\xi_{pq} = 2 & q=1,\cdots,6\\
		& \sum_{p=1}^6\xi_{pq}^2 = \frac65 & q=1,\cdots,6\\
		& \xi_{pq}\ge 0 & p,q=1,\cdots,6
	\end{array}
\end{eqnarray}
which can be broken into six independent and yet identical optimization problems of the form:
\begin{eqnarray}
	\begin{array}{lll}
		\max & \sum_{i=1}^6 x_i^3\\
		\text{subject to}  & \sum_{i=1}^6x_i^2=\frac65\\
		& \sum_{i=1}^6x_i=2\\
		& x_j \ge0 &  j=1,2,\cdots,6
	\end{array} \label{eq:appendix:ico:constrained}
\end{eqnarray}
In order to solve (\ref{eq:appendix:ico:constrained}), we first solve the following optimization problem:
\begin{eqnarray}
	\begin{array}{ll}
		\max & \sum_{i=1}^n x_i^3\\
		\text{subject to}  & \sum_{i=1}^nx_i^2=\frac65\\
		& \sum_{i=1}^nx_i=2
	\end{array}\label{eq:appendix:ico:unconstrained}
\end{eqnarray}
using Lagrange multipliers. The Lagrangian is given by:
\begin{eqnarray}
\mathcal{L}(x_1,\cdots,x_n;\lambda,\mu) &=& \sum_{i=1}^nx_i^3-\lambda\left[\sum_{i=1}^nx_j^2-\frac65\right]-\mu\left[\sum_{i=1}^nx_i-2\right]\notag
\end{eqnarray}
which can be differentiated to get:
\begin{eqnarray}
\frac{\partial\mathcal{L}}{\partial x_j}=3x_j^2-2\lambda x_j-\mu=0
\end{eqnarray}
Summing over $j$ yields $\mu=\frac1n\left[\frac{18}{5}-4\lambda\right]$. The roots of the quadratic equation (that we denote by $\xi_+$ and $\xi_-$) are thus given by: 
\begin{eqnarray}
\xi_{\pm} &=& \frac{1}{3}\left[\lambda\pm\sqrt{\lambda^2+\frac3n\left(\frac{18}5-4\lambda\right)}\right]
\end{eqnarray}
Let $n_+(n_-)$ be the number of $x_j$'s being equal to $\xi_+(\xi_-)$ and let $m=n_+-n_-$. From $n_+\xi_++n_-\xi_-=2$ we have:
\begin{eqnarray}
\lambda &=& \frac6n\left[1-m\sqrt{\frac{3n-10}{10(n^2-m^2)}}\right]\\
\xi_{\pm} &=& \frac2n\left[1-(m\mp n)\sqrt{\frac{3n-10}{10(n^2-m^2)}}\right]
\end{eqnarray}
which suggests that the problem is not feasible for $n<4$. The value of the objective function is given by:
\begin{eqnarray}
n_+\xi_+^3+n_-\xi_-^3 &=& \frac{36}{5n}-\frac{16}{n^2}-\frac{8m(3n-10)}{5n^2}\sqrt{\frac{3n-10}{10(n^2-m^2)}}
\end{eqnarray}
Since $d[m/\sqrt{n^2-m^2}]/dm=n^2/\sqrt{(n^2-m^2)^3}>0$, the function is maximized for the \emph{smallest} possible $m$ i.e. $m=2-n$ and the global maximum is given by:
\begin{eqnarray}
f(n) &=& \frac{36}{5n}-\frac{16}{n^2}+\frac{4(n-2)(3n-10)}{5n^2}\sqrt{\frac{3n-10}{10(n-1)}}\notag\\&\label{eq:appendix:ico:maximum}
\end{eqnarray}
Now we solve (\ref{eq:appendix:ico:constrained}) by identifying $x$'s that satisfy the Karush-Kuhn-Tucker criteria. Its Lagrangian is of the form:
\begin{eqnarray}
\mathcal{L} &=& \sum_{i=1}^6x_i^3-\nu_ix_i-\lambda\left[\sum_{i=1}^nx_j^2-\frac65\right]-\mu\left[\sum_{i=1}^nx_i-2\right]\notag\\
\frac{\partial\mathcal{L}}{\partial x_j} &=& 3x_j^2-\nu_j-2\lambda x_j-\mu\notag
\end{eqnarray}
It can be easily shown that for every $x_j=0$, $\nu_j=-\mu$ and is independent of $j$. For $x_j\neq0$, complementary slackness implies that $\nu_j=0$ and $\partial\mathcal{L}/\partial x_j=0$ implies $\mu\neq0$. 
Note that non-zero $x_j$'s are amongst the Karush-Kuhn-Tucker solutions of (\ref{eq:appendix:ico:unconstrained}) for $n=6-n_0$ with $n_0$ being the number of $x_j$'s that are zero. Since $f(4)<f(5)<f(6)$ from Eq.~(\ref{eq:appendix:ico:maximum}), $n_0=0$ gives the global maximum with $\sum_{i=1}^6x_i^3=\frac{26}{25}, \xi_+=1,\xi_-=\frac15$. The value of the objective function of (\ref{eq:app:icosa:opt}) is thus $\frac{156}{25}$ which can be \emph{only} achieved for $\xi_{pq}$'s being permutations of $(1,\frac15,\frac15,\frac15,\frac15,\frac15)$. 


\section{\label{appendix:trigonometric}Some Useful Trigonometric Integrals}
We are interested in integrals of the form:
\begin{eqnarray}
I_{m,n}&=&\int_0^{2\pi}\cos^m\theta\sin^n\theta d\theta\\
J_{m,n}&=&\int_0^{\pi}\cos^m\theta\sin^n\theta d\theta
\end{eqnarray}
Note that $I_{2m-1,2n}=I_{2m,2n-1}=0$ since the integrand is odd around $\theta=\pi$. Also:
\begin{eqnarray}
I_{2m-1,2n-1}=2\int_0^{\pi}\cos^{2m-1}\theta\sin^{2n-1}\theta d\theta=0\notag
\end{eqnarray}
since the integrand is odd around $\theta=\frac{\pi}{2}$. $I_{2m,2n}(m,n>0)$ can be calculated by considering the identity:
\begin{eqnarray}
I_{2m,0}=\frac{2\pi(2m)!}{4^m(m!)^2}\label{eq:I_2m_0}
\end{eqnarray}
and integration by part:
\begin{eqnarray}
\int_0^{2\pi}\cos^{2m}\theta\sin^{2n}\theta d\theta &=& \left[\frac{\cos^{2m-1}\theta\sin^{2n+1}\theta}{2n+1}\right]_0^{2\pi}+\tfrac{2m-1}{2n+1}\int_0^{2\pi}\cos^{2m-2}\theta\sin^{2n+2}\theta d\theta\notag\\&=&\frac{2m-1}{2n+1}I_{2m-2,2n+2}=  \frac{2\pi(2m)!(2n)!}{4^{m+n}m!n!(m+n)!}\label{eq:app:int_I_2m_2n}
\end{eqnarray}
(\ref{eq:I_2m_0}) can be proven as follows:
\begin{eqnarray}
I_{2m,0}&=&\frac{1}{4^m}\sum_{p=0}^{2m}\left(\begin{matrix}2m\\ p\end{matrix}\right)\int_0^{2\pi}\exp\left[2i(p-m)\right]d\theta= \frac{2\pi}{4^m}\left(\begin{matrix}2m\\ p\end{matrix}\right)\notag
\end{eqnarray}
since $\int_0^{2\pi}\exp\left[2i(p-m)\right]d\theta=0$ for $p\neq m$.
Note that $J_{2m+1,n}=0$ since the integrand is odd around $\pi/2$. For $m=2p,n=2q$ both even we have:
\begin{eqnarray}
J_{2p,2q} &=& \int_0^{\pi}\cos^{2p}x\sin^{2q}xdx=\frac12I_{2p,2q}\label{eq:appendix:J2p2q}
\end{eqnarray}
We also have:
\begin{eqnarray}
J_{2m,1} &=& \int_0^{\pi}\cos^{2m}x\sin xdx=\left.\frac{-\cos^{2m+1}x}{2m+1}\right|_0^{\pi}=\frac{2}{2m+1}
\end{eqnarray}
Similar to what was done for $I_{2m,2n}$, using integration by part we have:
\begin{eqnarray}
J_{2m,2n+1} &=& \frac{2n}{2m+1}J_{2m+2,2n-1}= \frac{2n}{2m+1}\frac{2n-2}{2m+3}\cdots\frac{2}{2m+2n-1}\frac{2}{2m+2n+1}= \frac{2^{2n+1}n!(2m)!(m+n)!}{m!(2m+2n+1)!}\label{eq:app:int_J_2m_2np1}
\end{eqnarray}

\section{\label{appendix:tetrahedral_nematic:constraint_equivalency} Equivalency of Rigidity constraints for regular tetrahedron}
We need to show that the following set of constraints are equivalent:
\begin{subequations}
		\begin{equation}
		\sum_{p=1}^4u_p^iu_p^j = \frac43\delta^{ij}
		\label{eq:appendix:tetrahedral_nematics:const1_a}
		\end{equation}
		\begin{equation}
		\sum_{p=1}^4u_p^i = 0 \label{eq:appendix:tetrahedral_nematics:const1_b}
		\end{equation}\label{eq:appendix:tetrahedral_nematics:const1}
\end{subequations}
\begin{eqnarray}
	u_p^iu_q^i = \frac43\delta_{pq}-\frac13, p,q=1,2,3,4
	\label{eq:appendix:tetrahedral_nematics:const2}
\end{eqnarray}
To prove that (\ref{eq:appendix:tetrahedral_nematics:const1}) implies (\ref{eq:appendix:tetrahedral_nematics:const2}), multiply (\ref{eq:appendix:tetrahedral_nematics:const1_a}) by $u _q$ and observe that $\sum_{p=1}^4(u_p^Tu_q)u_p=\frac43u_q$, which, after some rearrangement takes the form:
\begin{eqnarray}
	\sum_{p\neq q}\left[\frac43+\xi_{pq}-\xi_{pp}\right]u_p &=& 0 \notag
\end{eqnarray}
with $\xi_{pq}=a_p^Ta_q$. However since $\sum_{p=1}^4u_pu_p^T$ is a nonsingular matrix, any three of four $u_p$'s are linearly independent and $\xi_{qq}-\xi_{pq}=\frac43$. We also have $\norm{\sum_{p=1}^4u_pu_p^T}_F^2=16/3$ which yields $\xi_{qq}=1,\xi_{pq}=-\frac13 (p\neq q)$ and (\ref{eq:appendix:tetrahedral_nematics:const2}) follows.

To prove (\ref{eq:appendix:tetrahedral_nematics:const1}) from (\ref{eq:appendix:tetrahedral_nematics:const2}) take an arbitrary set of vectors satisfying (\ref{eq:appendix:tetrahedral_nematics:const2}) and show that (\ref{eq:appendix:tetrahedral_nematics:const1}) holds; however since $0$ and $\delta^{ij}$'s are isotropic tensors, they are invariant under any unitary transformation and the proof follows.

\section{Derivation Details of Uniaxial Nematics Order Parameter}
\label{appendix:nematic_OOP_derive}
The objective function of the associated optimization problem, $\mathfrak{M}_\Omega^{ijkl}\overline{M}^{ijkl}$, can be simplified by observing that $\overline{M}$ is both invariant under index permutation and traceless. More specifically:
\begin{eqnarray}
\overline{M}^{ijkl}(z^ix^jx^kz^l+z^iy^jy^kz^l) &=& \overline{M}^{ilkj}z^iz^l(x^kx^j+y^ky^j)= \overline{M}^{ijkl}z^iz^j(\delta^{kl}-z^kz^l)\notag
\end{eqnarray}
The same thing can be done for the six terms containing $\textbf{z, x}$ and $\textbf{y}$. With a similar argument one can verify that:
\begin{eqnarray}
\overline{M}^{ijkl}(x^ix^jx^kx^l+x^iy^jx^ky^l) &=& \overline{M}^{ijkl}x^ix^j(\delta^{kl}-z^kz^l)\notag
\end{eqnarray}
The objective function can therefore be written as:
\begin{eqnarray}
\mathfrak{M}^{ijkl}_{\Omega}\overline{M}^{ijkl} &=& \sum_{q=1}^4\left[\alpha_q^4z^iz^jz^kz^l+3\alpha_q^2\beta_q^2z^iz^j(\delta^{kl}-z^kz^l)+\frac{3\beta_q^4}{8}(\delta^{ij}-z^iz^j)(\delta^{kl}-z^kz^l)\right]\overline{M}^{ijkl}=\omega z^iz^jz^kz^l\overline{M}^{ijkl}\notag
\end{eqnarray}

\bibliographystyle{apsrev}
\bibliography{References}

\begin{thebibliography}{56}
\expandafter\ifx\csname natexlab\endcsname\relax\def\natexlab#1{#1}\fi
\expandafter\ifx\csname bibnamefont\endcsname\relax
  \def\bibnamefont#1{#1}\fi
\expandafter\ifx\csname bibfnamefont\endcsname\relax
  \def\bibfnamefont#1{#1}\fi
\expandafter\ifx\csname citenamefont\endcsname\relax
  \def\citenamefont#1{#1}\fi
\expandafter\ifx\csname url\endcsname\relax
  \def\url#1{\texttt{#1}}\fi
\expandafter\ifx\csname urlprefix\endcsname\relax\def\urlprefix{URL }\fi
\providecommand{\bibinfo}[2]{#2}
\providecommand{\eprint}[2][]{\url{#2}}

\bibitem[{\citenamefont{Ahmadi et~al.}(1996)\citenamefont{Ahmadi, Wang, Green,
  Henglein, and El-Sayed}}]{Ahmadi1996}
\bibinfo{author}{\bibfnamefont{T.~S.} \bibnamefont{Ahmadi}},
  \bibinfo{author}{\bibfnamefont{Z.~L.} \bibnamefont{Wang}},
  \bibinfo{author}{\bibfnamefont{T.~C.} \bibnamefont{Green}},
  \bibinfo{author}{\bibfnamefont{A.}~\bibnamefont{Henglein}}, \bibnamefont{and}
  \bibinfo{author}{\bibfnamefont{M.~A.} \bibnamefont{El-Sayed}},
  \bibinfo{journal}{Science} \textbf{\bibinfo{volume}{272}},
  \bibinfo{pages}{1924} (\bibinfo{year}{1996}),
  \urlprefix\url{http://dx.doi.org/ 10.1126/science.272.5270.1924}.

\bibitem[{\citenamefont{Peng}(2003)}]{PengAdvMatSemiconductor2003}
\bibinfo{author}{\bibfnamefont{X.}~\bibnamefont{Peng}}, \bibinfo{journal}{Adv.
  Mater.} \textbf{\bibinfo{volume}{15}}, \bibinfo{pages}{459}
  (\bibinfo{year}{2003}),
  \urlprefix\url{http://dx.doi.org/10.1002/adma.200390107}.

\bibitem[{\citenamefont{Jin et~al.}(2003)\citenamefont{Jin, Cao, Hao,
  M{\'e}traux, Schatz, and Mirkin}}]{JinMirkinNature2003}
\bibinfo{author}{\bibfnamefont{R.}~\bibnamefont{Jin}},
  \bibinfo{author}{\bibfnamefont{Y.~C.} \bibnamefont{Cao}},
  \bibinfo{author}{\bibfnamefont{E.}~\bibnamefont{Hao}},
  \bibinfo{author}{\bibfnamefont{G.~S.} \bibnamefont{M{\'e}traux}},
  \bibinfo{author}{\bibfnamefont{G.~C.} \bibnamefont{Schatz}},
  \bibnamefont{and} \bibinfo{author}{\bibfnamefont{C.~A.}
  \bibnamefont{Mirkin}}, \bibinfo{journal}{Nature}
  \textbf{\bibinfo{volume}{425}}, \bibinfo{pages}{487} (\bibinfo{year}{2003}),
  \urlprefix\url{http://dx.doi.org/10.1038/nature02020}.

\bibitem[{\citenamefont{Sau and Murphy}(2004)}]{MurphyJACS2004}
\bibinfo{author}{\bibfnamefont{T.~K.} \bibnamefont{Sau}} \bibnamefont{and}
  \bibinfo{author}{\bibfnamefont{C.~J.} \bibnamefont{Murphy}},
  \bibinfo{journal}{J. Am. Chem. Soc.} \textbf{\bibinfo{volume}{126}},
  \bibinfo{pages}{8648} (\bibinfo{year}{2004}),
  \urlprefix\url{http://dx.doi.org/10.1021/ja047846d}.

\bibitem[{\citenamefont{Champion et~al.}(2007)\citenamefont{Champion, Katare,
  and Mitragotri}}]{ChampoinPolymerPNAS2007}
\bibinfo{author}{\bibfnamefont{J.~A.} \bibnamefont{Champion}},
  \bibinfo{author}{\bibfnamefont{Y.~K.} \bibnamefont{Katare}},
  \bibnamefont{and}
  \bibinfo{author}{\bibfnamefont{S.}~\bibnamefont{Mitragotri}},
  \bibinfo{journal}{Proc. Natl. Acad. Sci. USA} \textbf{\bibinfo{volume}{104}},
  \bibinfo{pages}{11901} (\bibinfo{year}{2007}),
  \urlprefix\url{http://dx.doi.org/10.1073/pnas.0705326104}.

\bibitem[{\citenamefont{Yan and Yan}(2008)}]{YanJMaterChem2008}
\bibinfo{author}{\bibfnamefont{Z.-G.} \bibnamefont{Yan}} \bibnamefont{and}
  \bibinfo{author}{\bibfnamefont{C.-H.} \bibnamefont{Yan}},
  \bibinfo{journal}{J. Mater. Chem.} \textbf{\bibinfo{volume}{18}},
  \bibinfo{pages}{5046} (\bibinfo{year}{2008}),
  \urlprefix\url{http://dx.doi.org/10.1039/B810586C}.

\bibitem[{\citenamefont{Zhang et~al.}(2009)\citenamefont{Zhang, Li, Wu, Schatz,
  and Mirkin}}]{MirkinAngewChemIntEd2009}
\bibinfo{author}{\bibfnamefont{J.}~\bibnamefont{Zhang}},
  \bibinfo{author}{\bibfnamefont{S.}~\bibnamefont{Li}},
  \bibinfo{author}{\bibfnamefont{J.}~\bibnamefont{Wu}},
  \bibinfo{author}{\bibfnamefont{G.~C.} \bibnamefont{Schatz}},
  \bibnamefont{and} \bibinfo{author}{\bibfnamefont{C.~A.}
  \bibnamefont{Mirkin}}, \bibinfo{journal}{Angew. Chem. Int. Ed.}
  \textbf{\bibinfo{volume}{48}}, \bibinfo{pages}{7787} (\bibinfo{year}{2009}),
  \urlprefix\url{http://dx.doi.org/10.1002/ange.200903380}.

\bibitem[{\citenamefont{Sau and Rogach}(2010)}]{SauAdvancedMaterialRev2010}
\bibinfo{author}{\bibfnamefont{T.~K.} \bibnamefont{Sau}} \bibnamefont{and}
  \bibinfo{author}{\bibfnamefont{A.~L.} \bibnamefont{Rogach}},
  \bibinfo{journal}{Adv. Mater.} \textbf{\bibinfo{volume}{22}},
  \bibinfo{pages}{1781} (\bibinfo{year}{2010}),
  \urlprefix\url{http://dx.doi.org/10.1002/adma.200901271}.

\bibitem[{\citenamefont{Carbone and Cozzoli}(2010)}]{CarboneNanoToday2010}
\bibinfo{author}{\bibfnamefont{L.}~\bibnamefont{Carbone}} \bibnamefont{and}
  \bibinfo{author}{\bibfnamefont{P.~D.} \bibnamefont{Cozzoli}},
  \bibinfo{journal}{Nano Today} \textbf{\bibinfo{volume}{5}},
  \bibinfo{pages}{449} (\bibinfo{year}{2010}),
  \urlprefix\url{http://dx.doi.org/10.1016/j.nantod.2010.08.006}.

\bibitem[{\citenamefont{Gordon et~al.}(2012)\citenamefont{Gordon, Cargnello,
  Paik, Mangolini, Weber, Fornasiero, and Murray}}]{MurrayJACS2012}
\bibinfo{author}{\bibfnamefont{T.~R.} \bibnamefont{Gordon}},
  \bibinfo{author}{\bibfnamefont{M.}~\bibnamefont{Cargnello}},
  \bibinfo{author}{\bibfnamefont{T.}~\bibnamefont{Paik}},
  \bibinfo{author}{\bibfnamefont{F.}~\bibnamefont{Mangolini}},
  \bibinfo{author}{\bibfnamefont{R.~T.} \bibnamefont{Weber}},
  \bibinfo{author}{\bibfnamefont{P.}~\bibnamefont{Fornasiero}},
  \bibnamefont{and} \bibinfo{author}{\bibfnamefont{C.~B.}
  \bibnamefont{Murray}}, \bibinfo{journal}{J. Am. Chem. Soc.}
  \textbf{\bibinfo{volume}{134}}, \bibinfo{pages}{6751} (\bibinfo{year}{2012}),
  \urlprefix\url{http://dx.doi.org/10.1021/ja300823a}.

\bibitem[{\citenamefont{Henzie et~al.}(2012)\citenamefont{Henzie, Gr{\"u}nwald,
  Widmer-Cooper, Geissler, and Yang}}]{GeisslerYang2012}
\bibinfo{author}{\bibfnamefont{J.}~\bibnamefont{Henzie}},
  \bibinfo{author}{\bibfnamefont{M.}~\bibnamefont{Gr{\"u}nwald}},
  \bibinfo{author}{\bibfnamefont{A.}~\bibnamefont{Widmer-Cooper}},
  \bibinfo{author}{\bibfnamefont{P.~L.} \bibnamefont{Geissler}},
  \bibnamefont{and} \bibinfo{author}{\bibfnamefont{P.}~\bibnamefont{Yang}},
  \bibinfo{journal}{Nat. Mater.} \textbf{\bibinfo{volume}{11}},
  \bibinfo{pages}{131} (\bibinfo{year}{2012}),
  \urlprefix\url{http://dx.doi.org/10.1038/nmat3178}.

\bibitem[{\citenamefont{Langille et~al.}(2012)\citenamefont{Langille, Zhang,
  Personick, Li, and Mirkin}}]{MirkinScience2012}
\bibinfo{author}{\bibfnamefont{M.~R.} \bibnamefont{Langille}},
  \bibinfo{author}{\bibfnamefont{J.}~\bibnamefont{Zhang}},
  \bibinfo{author}{\bibfnamefont{M.~L.} \bibnamefont{Personick}},
  \bibinfo{author}{\bibfnamefont{S.}~\bibnamefont{Li}}, \bibnamefont{and}
  \bibinfo{author}{\bibfnamefont{C.~A.} \bibnamefont{Mirkin}},
  \bibinfo{journal}{Science} \textbf{\bibinfo{volume}{337}},
  \bibinfo{pages}{954} (\bibinfo{year}{2012}),
  \urlprefix\url{http://dx.doi.org/10.1126/science.1225653}.

\bibitem[{\citenamefont{Personick et~al.}(2013)\citenamefont{Personick,
  Langille, Wu, and Mirkin}}]{MirkinJACS2013}
\bibinfo{author}{\bibfnamefont{M.~L.} \bibnamefont{Personick}},
  \bibinfo{author}{\bibfnamefont{M.~R.} \bibnamefont{Langille}},
  \bibinfo{author}{\bibfnamefont{J.}~\bibnamefont{Wu}}, \bibnamefont{and}
  \bibinfo{author}{\bibfnamefont{C.~A.} \bibnamefont{Mirkin}},
  \bibinfo{journal}{J. Am. Chem. Soc.} \textbf{\bibinfo{volume}{135}},
  \bibinfo{pages}{3800} (\bibinfo{year}{2013}),
  \urlprefix\url{http://dx.doi.org/10.1021/ja400794q}.

\bibitem[{\citenamefont{Glotzer and
  Solomon}(2007)}]{GlotzerSolomonNatureMatrial2007}
\bibinfo{author}{\bibfnamefont{S.~C.} \bibnamefont{Glotzer}} \bibnamefont{and}
  \bibinfo{author}{\bibfnamefont{M.}~\bibnamefont{Solomon}},
  \bibinfo{journal}{Nat. Mater.} \textbf{\bibinfo{volume}{6}},
  \bibinfo{pages}{567} (\bibinfo{year}{2007}),
  \urlprefix\url{http://dx.doi.org/10.1038/nmat1949}.

\bibitem[{\citenamefont{Haji-Akbari et~al.}(2009)\citenamefont{Haji-Akbari,
  Engel, Keys, Zheng, Petschek, Palffy-Muhoray, and
  Glotzer}}]{HajiAkbariEtAl2009}
\bibinfo{author}{\bibfnamefont{A.}~\bibnamefont{Haji-Akbari}},
  \bibinfo{author}{\bibfnamefont{M.}~\bibnamefont{Engel}},
  \bibinfo{author}{\bibfnamefont{A.~S.} \bibnamefont{Keys}},
  \bibinfo{author}{\bibfnamefont{X.}~\bibnamefont{Zheng}},
  \bibinfo{author}{\bibfnamefont{R.}~\bibnamefont{Petschek}},
  \bibinfo{author}{\bibfnamefont{P.}~\bibnamefont{Palffy-Muhoray}},
  \bibnamefont{and} \bibinfo{author}{\bibfnamefont{S.~C.}
  \bibnamefont{Glotzer}}, \bibinfo{journal}{Nature}
  \textbf{\bibinfo{volume}{462}}, \bibinfo{pages}{773} (\bibinfo{year}{2009}),
  \urlprefix\url{http://dx.doi.org/10.1038/nature08641}.

\bibitem[{\citenamefont{Haji-Akbari
  et~al.}(2011{\natexlab{a}})\citenamefont{Haji-Akbari, Engel, and
  Glotzer}}]{HajiAkbaricondmat2011}
\bibinfo{author}{\bibfnamefont{A.}~\bibnamefont{Haji-Akbari}},
  \bibinfo{author}{\bibfnamefont{M.}~\bibnamefont{Engel}}, \bibnamefont{and}
  \bibinfo{author}{\bibfnamefont{S.~C.} \bibnamefont{Glotzer}},
  \bibinfo{journal}{J. Chem. Phys.} \textbf{\bibinfo{volume}{135}},
  \bibinfo{pages}{194101} (\bibinfo{year}{2011}{\natexlab{a}}),
  \urlprefix\url{http://dx.doi.org/10.1063/1.3651370}.

\bibitem[{\citenamefont{Haji-Akbari
  et~al.}(2011{\natexlab{b}})\citenamefont{Haji-Akbari, Engel, and
  Glotzer}}]{HajiAkbariDQC2011}
\bibinfo{author}{\bibfnamefont{A.}~\bibnamefont{Haji-Akbari}},
  \bibinfo{author}{\bibfnamefont{M.}~\bibnamefont{Engel}}, \bibnamefont{and}
  \bibinfo{author}{\bibfnamefont{S.~C.} \bibnamefont{Glotzer}},
  \bibinfo{journal}{Phys. Rev. Lett.} \textbf{\bibinfo{volume}{107}},
  \bibinfo{pages}{215702} (\bibinfo{year}{2011}{\natexlab{b}}),
  \urlprefix\url{http://dx.doi.org/10.1103/PhysRevLett.107.215702}.

\bibitem[{\citenamefont{Damasceno
  et~al.}(2012{\natexlab{a}})\citenamefont{Damasceno, Engel, and
  Glotzer}}]{PabloACSNanot2012}
\bibinfo{author}{\bibfnamefont{P.~F.} \bibnamefont{Damasceno}},
  \bibinfo{author}{\bibfnamefont{M.}~\bibnamefont{Engel}}, \bibnamefont{and}
  \bibinfo{author}{\bibfnamefont{S.~C.} \bibnamefont{Glotzer}},
  \bibinfo{journal}{ACS Nano} \textbf{\bibinfo{volume}{6}},
  \bibinfo{pages}{609} (\bibinfo{year}{2012}{\natexlab{a}}),
  \urlprefix\url{http://dx.doi.org/10.1021/nn204012y}.

\bibitem[{\citenamefont{Damasceno
  et~al.}(2012{\natexlab{b}})\citenamefont{Damasceno, Engel, and
  Glotzer}}]{PabloScience2012}
\bibinfo{author}{\bibfnamefont{P.~F.} \bibnamefont{Damasceno}},
  \bibinfo{author}{\bibfnamefont{M.}~\bibnamefont{Engel}}, \bibnamefont{and}
  \bibinfo{author}{\bibfnamefont{S.~C.} \bibnamefont{Glotzer}},
  \bibinfo{journal}{Science} \textbf{\bibinfo{volume}{337}},
  \bibinfo{pages}{453} (\bibinfo{year}{2012}{\natexlab{b}}),
  \urlprefix\url{http://dx.doi.org/10.1126/science.1220869}.

\bibitem[{\citenamefont{Haji-Akbari et~al.}(2013)\citenamefont{Haji-Akbari,
  Chen, Engel, and Glotzer}}]{HajiAkbariPRE2013}
\bibinfo{author}{\bibfnamefont{A.}~\bibnamefont{Haji-Akbari}},
  \bibinfo{author}{\bibfnamefont{E.~R.} \bibnamefont{Chen}},
  \bibinfo{author}{\bibfnamefont{M.}~\bibnamefont{Engel}}, \bibnamefont{and}
  \bibinfo{author}{\bibfnamefont{S.~C.} \bibnamefont{Glotzer}},
  \bibinfo{journal}{Phys. Rev. E} \textbf{\bibinfo{volume}{88}},
  \bibinfo{pages}{012127} (\bibinfo{year}{2013}),
  \urlprefix\url{http://dx.doi.org/10.1103/PhysRevE.88.012127}.

\bibitem[{\citenamefont{Ye et~al.}(2013)\citenamefont{Ye, Chen, Engel, Millan,
  Li, Qi, Xing, Collins, Kagan, Li et~al.}}]{MurrayNatChem2013}
\bibinfo{author}{\bibfnamefont{X.}~\bibnamefont{Ye}},
  \bibinfo{author}{\bibfnamefont{J.}~\bibnamefont{Chen}},
  \bibinfo{author}{\bibfnamefont{M.}~\bibnamefont{Engel}},
  \bibinfo{author}{\bibfnamefont{J.~A.} \bibnamefont{Millan}},
  \bibinfo{author}{\bibfnamefont{W.}~\bibnamefont{Li}},
  \bibinfo{author}{\bibfnamefont{L.}~\bibnamefont{Qi}},
  \bibinfo{author}{\bibfnamefont{G.}~\bibnamefont{Xing}},
  \bibinfo{author}{\bibfnamefont{J.~E.} \bibnamefont{Collins}},
  \bibinfo{author}{\bibfnamefont{C.~R.} \bibnamefont{Kagan}},
  \bibinfo{author}{\bibfnamefont{J.}~\bibnamefont{Li}}, \bibnamefont{et~al.},
  \bibinfo{journal}{Nat. Chem.} \textbf{\bibinfo{volume}{5}},
  \bibinfo{pages}{466} (\bibinfo{year}{2013}),
  \urlprefix\url{http://dx.doi.org/10.1038/nchem.1651}.

\bibitem[{\citenamefont{Millan et~al.}(2014)\citenamefont{Millan, Ortiz, van
  Anders, and Glotzer}}]{MilanACSNano2014}
\bibinfo{author}{\bibfnamefont{J.~A.} \bibnamefont{Millan}},
  \bibinfo{author}{\bibfnamefont{D.}~\bibnamefont{Ortiz}},
  \bibinfo{author}{\bibfnamefont{G.}~\bibnamefont{van Anders}},
  \bibnamefont{and} \bibinfo{author}{\bibfnamefont{S.~C.}
  \bibnamefont{Glotzer}}, \bibinfo{journal}{ACS Nano}
  \textbf{\bibinfo{volume}{8}}, \bibinfo{pages}{2918} (\bibinfo{year}{2014}),
  \urlprefix\url{http://dx.doi.org/10.1021/nn500147u}.

\bibitem[{\citenamefont{van Anders et~al.}(2014)\citenamefont{van Anders,
  Klotsa, Ahmed, Engel, and Glotzer}}]{GregPNAS2014}
\bibinfo{author}{\bibfnamefont{G.}~\bibnamefont{van Anders}},
  \bibinfo{author}{\bibfnamefont{D.}~\bibnamefont{Klotsa}},
  \bibinfo{author}{\bibfnamefont{N.~K.} \bibnamefont{Ahmed}},
  \bibinfo{author}{\bibfnamefont{M.}~\bibnamefont{Engel}}, \bibnamefont{and}
  \bibinfo{author}{\bibfnamefont{S.~C.} \bibnamefont{Glotzer}},
  \bibinfo{journal}{Proc. Natl. Acad. Sci. USA} \textbf{\bibinfo{volume}{111}},
  \bibinfo{pages}{E4812} (\bibinfo{year}{2014}),
  \urlprefix\url{http://dx.doi.org/10.1073/pnas.1418159111}.

\bibitem[{\citenamefont{Friedrich et~al.}(1913)\citenamefont{Friedrich,
  Knipping, and von Laue}}]{vonLaue1912}
\bibinfo{author}{\bibfnamefont{W.}~\bibnamefont{Friedrich}},
  \bibinfo{author}{\bibfnamefont{P.}~\bibnamefont{Knipping}}, \bibnamefont{and}
  \bibinfo{author}{\bibfnamefont{M.}~\bibnamefont{von Laue}},
  \bibinfo{journal}{Annalen der Physik} \textbf{\bibinfo{volume}{346}},
  \bibinfo{pages}{971} (\bibinfo{year}{1913}),
  \urlprefix\url{http://dx.doi.org/10.1002/andp.19133461004}.

\bibitem[{\citenamefont{Shechtman et~al.}(1984)\citenamefont{Shechtman, Blech,
  Gratias, and W.Cahn}}]{ShechtmanPRL1984}
\bibinfo{author}{\bibfnamefont{D.}~\bibnamefont{Shechtman}},
  \bibinfo{author}{\bibfnamefont{I.}~\bibnamefont{Blech}},
  \bibinfo{author}{\bibfnamefont{D.}~\bibnamefont{Gratias}}, \bibnamefont{and}
  \bibinfo{author}{\bibfnamefont{J.}~\bibnamefont{W.Cahn}},
  \bibinfo{journal}{Phys. Rev. Lett.} \textbf{\bibinfo{volume}{53}},
  \bibinfo{pages}{1951} (\bibinfo{year}{1984}),
  \urlprefix\url{http://dx.doi.org/10.1103/PhysRevLett.53.1951}.

\bibitem[{\citenamefont{Onsager}(1949)}]{Onsager1949}
\bibinfo{author}{\bibfnamefont{L.}~\bibnamefont{Onsager}},
  \bibinfo{journal}{Ann. NY Acad. Sci.} \textbf{\bibinfo{volume}{51}},
  \bibinfo{pages}{627} (\bibinfo{year}{1949}),
  \urlprefix\url{http://dx.doi.org/10.1111/j.1749-6632.1949.tb27296.x}.

\bibitem[{\citenamefont{Eppenga and Frenkel}(1984)}]{EppengaFrenkel1984}
\bibinfo{author}{\bibfnamefont{R.}~\bibnamefont{Eppenga}} \bibnamefont{and}
  \bibinfo{author}{\bibfnamefont{D.}~\bibnamefont{Frenkel}},
  \bibinfo{journal}{Mol. Phys.} \textbf{\bibinfo{volume}{52}},
  \bibinfo{pages}{1303} (\bibinfo{year}{1984}),
  \urlprefix\url{http://dx.doi.org/10.1080/0026897840010195}.

\bibitem[{\citenamefont{Bates and Frenkel}(1998)}]{BatesFrenkel1998}
\bibinfo{author}{\bibfnamefont{M.~A.} \bibnamefont{Bates}} \bibnamefont{and}
  \bibinfo{author}{\bibfnamefont{D.}~\bibnamefont{Frenkel}},
  \bibinfo{journal}{Phys. Rev. E} \textbf{\bibinfo{volume}{57}},
  \bibinfo{pages}{4824} (\bibinfo{year}{1998}),
  \urlprefix\url{http://dx.doi.org/http://dx.doi.org/10.1103/PhysRevE.57.4824}.

\bibitem[{\citenamefont{Steinhardt et~al.}(1983)\citenamefont{Steinhardt,
  Nelson, and Ronchetti}}]{SteinhardtPRB1983}
\bibinfo{author}{\bibfnamefont{P.~J.} \bibnamefont{Steinhardt}},
  \bibinfo{author}{\bibfnamefont{D.~R.} \bibnamefont{Nelson}},
  \bibnamefont{and}
  \bibinfo{author}{\bibfnamefont{M.}~\bibnamefont{Ronchetti}},
  \bibinfo{journal}{Phys. Rev. B} \textbf{\bibinfo{volume}{28}},
  \bibinfo{pages}{784} (\bibinfo{year}{1983}),
  \urlprefix\url{http://dx.doi.org/10.1103/PhysRevB.28.784}.

\bibitem[{\citenamefont{ten Wolde et~al.}(1995)\citenamefont{ten Wolde,
  Ruiz-Montero, and Frenkel}}]{FrenkelPRL1995}
\bibinfo{author}{\bibfnamefont{P.~R.} \bibnamefont{ten Wolde}},
  \bibinfo{author}{\bibfnamefont{M.~J.} \bibnamefont{Ruiz-Montero}},
  \bibnamefont{and} \bibinfo{author}{\bibfnamefont{D.}~\bibnamefont{Frenkel}},
  \bibinfo{journal}{Phys. Rev. Lett.} \textbf{\bibinfo{volume}{75}},
  \bibinfo{pages}{2714} (\bibinfo{year}{1995}),
  \urlprefix\url{http://dx.doi.org/10.1103/PhysRevLett.75.2714}.

\bibitem[{\citenamefont{Chopra et~al.}(2006)\citenamefont{Chopra, M\"{u}ller,
  and de~Pablo}}]{ChopraJCP2006}
\bibinfo{author}{\bibfnamefont{M.}~\bibnamefont{Chopra}},
  \bibinfo{author}{\bibfnamefont{M.}~\bibnamefont{M\"{u}ller}},
  \bibnamefont{and} \bibinfo{author}{\bibfnamefont{J.~J.}
  \bibnamefont{de~Pablo}}, \bibinfo{journal}{J. Chem. Phys.}
  \textbf{\bibinfo{volume}{124}}, \bibinfo{pages}{134102}
  (\bibinfo{year}{2006}), \urlprefix\url{http://dx.doi.org/10.1063/1.2178324}.

\bibitem[{\citenamefont{Lechner and Dellago}(2008)}]{DelagoJCP2008}
\bibinfo{author}{\bibfnamefont{W.}~\bibnamefont{Lechner}} \bibnamefont{and}
  \bibinfo{author}{\bibfnamefont{C.}~\bibnamefont{Dellago}},
  \bibinfo{journal}{J. Chem. Phys.} \textbf{\bibinfo{volume}{129}},
  \bibinfo{pages}{114707} (\bibinfo{year}{2008}),
  \urlprefix\url{http://dx.doi.org/10.1063/1.2977970}.

\bibitem[{\citenamefont{Moore and Molinero}(2011)}]{MolineroPCCP2011}
\bibinfo{author}{\bibfnamefont{E.~B.} \bibnamefont{Moore}} \bibnamefont{and}
  \bibinfo{author}{\bibfnamefont{V.}~\bibnamefont{Molinero}},
  \bibinfo{journal}{Phys. Chem. Chem. Phys.} \textbf{\bibinfo{volume}{13}},
  \bibinfo{pages}{20008} (\bibinfo{year}{2011}),
  \urlprefix\url{http://dx.doi.org/10.1039/C1CP22022E}.

\bibitem[{\citenamefont{Haji-Akbari et~al.}(2014)\citenamefont{Haji-Akbari,
  DeFever, Sarupria, and Debenedetti}}]{HajiAkbariFilmMolinero2014}
\bibinfo{author}{\bibfnamefont{A.}~\bibnamefont{Haji-Akbari}},
  \bibinfo{author}{\bibfnamefont{R.~S.} \bibnamefont{DeFever}},
  \bibinfo{author}{\bibfnamefont{S.}~\bibnamefont{Sarupria}}, \bibnamefont{and}
  \bibinfo{author}{\bibfnamefont{P.~G.} \bibnamefont{Debenedetti}},
  \bibinfo{journal}{Phys. Chem. Chem. Phys.} \textbf{\bibinfo{volume}{16}},
  \bibinfo{pages}{25916} (\bibinfo{year}{2014}),
  \urlprefix\url{http://dx.doi.org/10.1039/C4CP03948C}.

\bibitem[{\citenamefont{Palmer et~al.}(2014)\citenamefont{Palmer, Martelli,
  Liu, Car, Panagiotopoulos, and Debenedetti}}]{PalmerNature2014}
\bibinfo{author}{\bibfnamefont{J.~C.} \bibnamefont{Palmer}},
  \bibinfo{author}{\bibfnamefont{F.}~\bibnamefont{Martelli}},
  \bibinfo{author}{\bibfnamefont{Y.}~\bibnamefont{Liu}},
  \bibinfo{author}{\bibfnamefont{R.}~\bibnamefont{Car}},
  \bibinfo{author}{\bibfnamefont{A.~Z.} \bibnamefont{Panagiotopoulos}},
  \bibnamefont{and} \bibinfo{author}{\bibfnamefont{P.~G.}
  \bibnamefont{Debenedetti}}, \bibinfo{journal}{Nature}
  \textbf{\bibinfo{volume}{510}}, \bibinfo{pages}{385} (\bibinfo{year}{2014}),
  \urlprefix\url{http://dx.doi.org/10.1038/nature13405}.

\bibitem[{\citenamefont{Nguyen and Molinero}(2015)}]{MolineroJPCB2014}
\bibinfo{author}{\bibfnamefont{A.~H.} \bibnamefont{Nguyen}} \bibnamefont{and}
  \bibinfo{author}{\bibfnamefont{V.}~\bibnamefont{Molinero}},
  \bibinfo{journal}{J. Phys. Chem. B} \textbf{\bibinfo{volume}{119}},
  \bibinfo{pages}{9369} (\bibinfo{year}{2015}),
  \urlprefix\url{http://dx.doi.org/10.1021/jp510289t}.

\bibitem[{\citenamefont{Haji-Akbari and
  Debenedetti}(2015)}]{HajiAkbariPNAS2015}
\bibinfo{author}{\bibfnamefont{A.}~\bibnamefont{Haji-Akbari}} \bibnamefont{and}
  \bibinfo{author}{\bibfnamefont{P.~G.} \bibnamefont{Debenedetti}},
  \bibinfo{journal}{Proc. Natl. Acad. Sci. USA} \textbf{\bibinfo{volume}{112}},
  \bibinfo{pages}{10582} (\bibinfo{year}{2015}),
  \urlprefix\url{http://dx.doi.org/10.1073/pnas.1509267112}.

\bibitem[{\citenamefont{Keys et~al.}(2011{\natexlab{a}})\citenamefont{Keys,
  Iacovella, and Glotzer}}]{KeysAnnRevCondMatPhys2011}
\bibinfo{author}{\bibfnamefont{A.~S.} \bibnamefont{Keys}},
  \bibinfo{author}{\bibfnamefont{C.~R.} \bibnamefont{Iacovella}},
  \bibnamefont{and} \bibinfo{author}{\bibfnamefont{S.~C.}
  \bibnamefont{Glotzer}}, \bibinfo{journal}{Ann. Rev. Cond. Mat. Phys.}
  \textbf{\bibinfo{volume}{2}}, \bibinfo{pages}{263}
  (\bibinfo{year}{2011}{\natexlab{a}}),
  \urlprefix\url{http://dx.doi.org/10.1146/annurev-conmatphys-062910-140526}.

\bibitem[{\citenamefont{Keys et~al.}(2011{\natexlab{b}})\citenamefont{Keys,
  Iacovella, and Glotzer}}]{KeysJCompPhys2011}
\bibinfo{author}{\bibfnamefont{A.~S.} \bibnamefont{Keys}},
  \bibinfo{author}{\bibfnamefont{C.~R.} \bibnamefont{Iacovella}},
  \bibnamefont{and} \bibinfo{author}{\bibfnamefont{S.~C.}
  \bibnamefont{Glotzer}}, \bibinfo{journal}{J. Comp. Phys.}
  \textbf{\bibinfo{volume}{230}}, \bibinfo{pages}{6438}
  (\bibinfo{year}{2011}{\natexlab{b}}),
  \urlprefix\url{http://dx.doi.org/10.1016/j.jcp.2011.04.017}.

\bibitem[{\citenamefont{Zannoni}(1979)}]{Zannoni1979}
\bibinfo{author}{\bibfnamefont{C.}~\bibnamefont{Zannoni}},
  \emph{\bibinfo{title}{Distribution Functions and Order Parameters}}
  (\bibinfo{publisher}{Academic Press, London}, \bibinfo{year}{1979}),
  chap.~\bibinfo{chapter}{3}, pp. \bibinfo{pages}{51--83}.

\bibitem[{\citenamefont{Bisia et~al.}(2010)\citenamefont{Bisia, Longa, Pajak,
  and Rossoa}}]{RossoaMolCrystLiqCryst2010}
\bibinfo{author}{\bibfnamefont{F.}~\bibnamefont{Bisia}},
  \bibinfo{author}{\bibfnamefont{L.}~\bibnamefont{Longa}},
  \bibinfo{author}{\bibfnamefont{G.}~\bibnamefont{Pajak}}, \bibnamefont{and}
  \bibinfo{author}{\bibfnamefont{R.}~\bibnamefont{Rossoa}},
  \bibinfo{journal}{Mol. Cryst. Liq. Cryst.} \textbf{\bibinfo{volume}{525}},
  \bibinfo{pages}{12} (\bibinfo{year}{2010}),
  \urlprefix\url{http://dx.doi.org/10.1080/15421401003795670}.

\bibitem[{\citenamefont{Rosso}(2007)}]{Rosso2007}
\bibinfo{author}{\bibfnamefont{R.}~\bibnamefont{Rosso}}, \bibinfo{journal}{Liq.
  Cryst.} \textbf{\bibinfo{volume}{34}}, \bibinfo{pages}{737}
  (\bibinfo{year}{2007}),
  \urlprefix\url{http://dx.doi.org/10.1080/02678290701284303}.

\bibitem[{\citenamefont{John et~al.}(2008)\citenamefont{John, Juhlin, and
  Escobedo}}]{JohnEscobedo2008}
\bibinfo{author}{\bibfnamefont{B.~S.} \bibnamefont{John}},
  \bibinfo{author}{\bibfnamefont{C.}~\bibnamefont{Juhlin}}, \bibnamefont{and}
  \bibinfo{author}{\bibfnamefont{F.~A.} \bibnamefont{Escobedo}},
  \bibinfo{journal}{J. Chem. Phys.} \textbf{\bibinfo{volume}{128}},
  \bibinfo{pages}{044909} (\bibinfo{year}{2008}),
  \urlprefix\url{http://dx.doi.org/10.1063/1.2819091}.

\bibitem[{\citenamefont{Stone}(1978)}]{StoneMolPhys1978}
\bibinfo{author}{\bibfnamefont{A.~J.} \bibnamefont{Stone}},
  \bibinfo{journal}{Mol. Phys.} \textbf{\bibinfo{volume}{36}},
  \bibinfo{pages}{241} (\bibinfo{year}{1978}),
  \urlprefix\url{http://dx.doi.org/10.1080/00268977800101541}.

\bibitem[{\citenamefont{Indenbom et~al.}(1976)\citenamefont{Indenbom, Pikin,
  and Loginov}}]{Indenbom1976SovPhysCrystallogr}
\bibinfo{author}{\bibfnamefont{V.~L.} \bibnamefont{Indenbom}},
  \bibinfo{author}{\bibfnamefont{S.~A.} \bibnamefont{Pikin}}, \bibnamefont{and}
  \bibinfo{author}{\bibfnamefont{E.~B.} \bibnamefont{Loginov}},
  \bibinfo{journal}{Sov. Phys. Crystallogr.} \textbf{\bibinfo{volume}{21}},
  \bibinfo{pages}{633} (\bibinfo{year}{1976}).

\bibitem[{\citenamefont{Fel}(1995)}]{Fel1995PhysRevE}
\bibinfo{author}{\bibfnamefont{L.~G.} \bibnamefont{Fel}},
  \bibinfo{journal}{Phys. Rev. E} \textbf{\bibinfo{volume}{52}},
  \bibinfo{pages}{702} (\bibinfo{year}{1995}),
  \urlprefix\url{http://dx.doi.org/10.1103/PhysRevE.52.702}.

\bibitem[{\citenamefont{Leclerc}(1996)}]{Leclerc1996DiscMath}
\bibinfo{author}{\bibfnamefont{B.}~\bibnamefont{Leclerc}},
  \bibinfo{journal}{Disc. Math.} \textbf{\bibinfo{volume}{153}},
  \bibinfo{pages}{213} (\bibinfo{year}{1996}),
  \urlprefix\url{http://dx.doi.org/10.1016/0012-365X(95)00138-M}.

\bibitem[{\citenamefont{Hahn}(2006)}]{IntTabCrystWyckoff}
\bibinfo{author}{\bibfnamefont{T.}~\bibnamefont{Hahn}},
  \emph{\bibinfo{title}{International Tables for Crystallography}}
  (\bibinfo{publisher}{REIDEL Publ. Co. Dordrecht, Holland/Boston USA},
  \bibinfo{year}{2006}), vol.~\bibinfo{volume}{A}, chap.
  \bibinfo{chapter}{8.3.2.}, pp. \bibinfo{pages}{732--734}.

\bibitem[{\citenamefont{Radu et~al.}(2009)\citenamefont{Radu, Pfleiderer, and
  Schilling}}]{RaduSchilling2009}
\bibinfo{author}{\bibfnamefont{M.}~\bibnamefont{Radu}},
  \bibinfo{author}{\bibfnamefont{P.}~\bibnamefont{Pfleiderer}},
  \bibnamefont{and}
  \bibinfo{author}{\bibfnamefont{T.}~\bibnamefont{Schilling}},
  \bibinfo{journal}{J. Chem. Phys.} \textbf{\bibinfo{volume}{131}},
  \bibinfo{pages}{164513} (\bibinfo{year}{2009}),
  \urlprefix\url{http://dx.doi.org/10.1063/1.3251054}.

\bibitem[{\citenamefont{Agarwal and
  Escobedo}(2011)}]{EscobedoNatureMaterials2011}
\bibinfo{author}{\bibfnamefont{U.}~\bibnamefont{Agarwal}} \bibnamefont{and}
  \bibinfo{author}{\bibfnamefont{F.~A.} \bibnamefont{Escobedo}},
  \bibinfo{journal}{Nat. Mater.} \textbf{\bibinfo{volume}{10}},
  \bibinfo{pages}{230} (\bibinfo{year}{2011}),
  \urlprefix\url{http://dx.doi.org/10.1038/nmat2959}.

\bibitem[{\citenamefont{Fia{\l}kowski and Hess}(2000)}]{FiakowskiPhysicaA2000}
\bibinfo{author}{\bibfnamefont{M.}~\bibnamefont{Fia{\l}kowski}}
  \bibnamefont{and} \bibinfo{author}{\bibfnamefont{S.}~\bibnamefont{Hess}},
  \bibinfo{journal}{Physica A} \textbf{\bibinfo{volume}{284}},
  \bibinfo{pages}{59} (\bibinfo{year}{2000}),
  \urlprefix\url{http://dx.doi.org/10.1016/S0378-4371(00)00235-1}.

\bibitem[{\citenamefont{Wojciechowski and
  Frenkel}(2004)}]{WojciechowskiCompMetSciTechnol2004}
\bibinfo{author}{\bibfnamefont{K.~W.} \bibnamefont{Wojciechowski}}
  \bibnamefont{and} \bibinfo{author}{\bibfnamefont{D.}~\bibnamefont{Frenkel}},
  \bibinfo{journal}{Comp. Met. Sci. Technol.} \textbf{\bibinfo{volume}{10}},
  \bibinfo{pages}{235} (\bibinfo{year}{2004}).

\bibitem[{\citenamefont{Stillinger and Weber}(1985)}]{StillingerWeberJCP1985}
\bibinfo{author}{\bibfnamefont{F.~H.} \bibnamefont{Stillinger}}
  \bibnamefont{and} \bibinfo{author}{\bibfnamefont{T.~A.} \bibnamefont{Weber}},
  \bibinfo{journal}{J. Chem. Phys.} \textbf{\bibinfo{volume}{83}},
  \bibinfo{pages}{4767} (\bibinfo{year}{1985}),
  \urlprefix\url{http://dx.doi.org/10.1063/1.449840}.

\bibitem[{\citenamefont{Kalman}(2000)}]{KalmanMathMag2000}
\bibinfo{author}{\bibfnamefont{D.}~\bibnamefont{Kalman}},
  \bibinfo{journal}{Math. Mag.} \textbf{\bibinfo{volume}{73}},
  \bibinfo{pages}{313} (\bibinfo{year}{2000}),
  \urlprefix\url{http://dx.doi.org/10.2307/2690982}.

\bibitem[{\citenamefont{Karush}(1939)}]{Karush1939}
\bibinfo{author}{\bibfnamefont{W.}~\bibnamefont{Karush}}, Master's thesis,
  \bibinfo{school}{University of Chicago} (\bibinfo{year}{1939}).

\bibitem[{\citenamefont{Kuhn and Tucker}(1951)}]{KuhnTucker1951}
\bibinfo{author}{\bibfnamefont{H.~W.} \bibnamefont{Kuhn}} \bibnamefont{and}
  \bibinfo{author}{\bibfnamefont{A.~W.} \bibnamefont{Tucker}}, in
  \emph{\bibinfo{booktitle}{{Proceedings of the 2nd Berkeley Symposium}}}
  (\bibinfo{publisher}{Berkeley: University of California Press},
  \bibinfo{year}{1951}), pp. \bibinfo{pages}{481--492}.

\end{thebibliography}

\end{document}